\documentclass[preprint,3p,12pt,sort&compress]{elsarticle}

\usepackage[utf8]{inputenc}
\usepackage{mathrsfs}
\usepackage{amsmath}
\usepackage{stmaryrd}
\usepackage{bbding}
\usepackage{dcolumn}
\usepackage{graphicx}
\usepackage{amsfonts}
\usepackage{amssymb}
\usepackage{psfrag}
\usepackage{wrapfig}
\usepackage{subfigure}
\usepackage{makeidx}
\usepackage{bm}
\usepackage{epsf}
\usepackage{epsfig}
\usepackage{setspace}
\usepackage{graphicx}
\usepackage{epstopdf}
\usepackage{psfrag}
\usepackage{subfigure}
\usepackage{color}
\usepackage{enumerate}
\usepackage{verbatim} % comment-environment
\usepackage{booktabs}
\usepackage[utf8]{inputenc}
\usepackage{amsthm}
\usepackage{url}

\newtheorem{theorem}{Theorem}

\usepackage[ruled,vlined]{algorithm2e}
\SetKwRepeat{Do}{do}{while}

\begin{document}

\title{RelaxNet: A structure-preserving neural network to approximate the Boltzmann collision operator}

\author{Tianbai Xiao\corref{cor}}
\ead{tianbaixiao@gmail.com}

\author{Martin Frank}

\address{Karlsruhe Institute of Technology, Karlsruhe, Germany}
\cortext[cor]{Corresponding author}

\begin{abstract}

The extremely high dimensionality and nonlinearity in the Boltzmann equation bring tremendous difficulties to the study of rarefied gas dynamics.
This paper addresses a neural network-based surrogate model that provides a structure-preserving approximation for the fivefold collision integral.
The idea and notion originate from the similarity in structure between the BGK-type relaxation model and residual neural network (ResNet) when a particle distribution function is treated as the input to the neural network function.
Therefore, we extend the ResNet architecture and construct what we call the relaxation neural network (RelaxNet).
Specifically, two feed-forward neural networks with physics-informed connections and activations are introduced as building blocks in RelaxNet, which provide bounded and physically realizable approximations of the equilibrium distribution and velocity-dependent relaxation time respectively.
The evaluation of the collision term is thus significantly accelerated due to the fact that the convolution in the fivefold integral is replaced by tensor multiplication in the neural network.
We fuse the mechanical advection operator and the RelaxNet-based collision operator into a unified model named the universal Boltzmann equation (UBE).
We prove that UBE preserves the key structural properties in a many-particle system, i.e., positivity, conservation, invariance, H-theorem, and correct fluid dynamic limit.
These properties promise that RelaxNet is superior to strategies that naively approximate the right-hand side of the Boltzmann equation using a machine learning model.
The implementation of RelaxNet and training strategy are demonstrated.
The construction of the RelaxNet-based UBE and its solution algorithm are presented in detail.
Several numerical experiments, where the ground-truth datasets for the supervised learning task are produced by the Shakhov model, velocity-dependent $\nu$-BGK model, and the full Boltzmann equation, are investigated.
The capability of the current approach for simulating non-equilibrium flow physics is validated through excellent in- and out-of-distribution performance.
The open-source codes to reproduce the numerical results are available under the MIT license \cite{xiao2021kinetic}.

%hidden blocks
\end{abstract}


\maketitle
%\newpage

% Nomenclature
\begin{table}
	\centering
	\caption{Nomenclature.}
	\begin{tabular*}{16cm}{lll}
		\hline
		%\hline
		$t$, $\mathbf x$, $\mathbf v$ & time, space, and particle velocity variables \\
		$f$ & particle distribution function \\
		$\mathcal Q$, $\mathcal A$ & collision and advection operators of the Boltzmann equation \\
		$\nu$ & relaxation frequency \\
		$\boldsymbol\psi$ & collision invariants \\
		$H$ & function used in the H-theorem \\
		$\mathcal M$ & Maxwellian distribution function \\
		$\rho$, $\mathbf V$, $T$, $E$ & density, bulk velocity, temperature, and energy \\
		$\mathbf c$ & peculiar velocity \\
		$m$ & molecular mass \\
		$k$ & Boltzmann constant \\
		$\mathcal E$ & equilibrium state in the relaxation model \\
		$\sigma$ & activation function \\
		$\mathcal E$-net & subnet in RelaxNet to approximate equilibrium state \\
		$\tau$-net & subnet in RelaxNet to approximate relaxation time \\
		$D$ & dimension \\
		$\tau$ & relaxation time \\
		$\bar\nu$, $\bar\tau$ & mean relaxation frequency and time \\
		$\mu$ & viscosity \\
		$p$ & pressure \\
		$\mathcal M^+$ & correction term for equilibrium distribution \\
		$\boldsymbol\alpha$, $\boldsymbol\beta$ & outputs of $\mathcal E$-net and $\tau$-net \\
		ELU, ReLU & Exponential Linear Unit and Rectified Linear Unit functions \\
		$\theta$ & trainable parameters \\
		$\mathrm{NN}_\theta$ & RelaxNet function \\
		$\theta_1$, $\theta_2$ & The parts of $\theta$ belonging to $\mathcal E$-net and $\tau$-net \\
		$C$ & cost function \\
		$\lambda$,\ $\xi$ & regularization strengths \\
		$\varepsilon$ & minimum value of cost function \\
		$\hat{\boldsymbol\alpha}$ & Coefficients of Maxwellian in summation form of collision invariants \\
		$\tilde{\boldsymbol\alpha}$ & product of ${\boldsymbol\alpha}$ and $\hat{\boldsymbol\alpha}$ \\
		$\mathbf P$, $\mathbf q$ & pressure tensor and heat flux \\
		$\mathbf m$ & moment variables \\
		$\boldsymbol\phi$ & basis functions \\
		$N_m$ & polynomial degree of $\boldsymbol\phi$ \\
		$\boldsymbol \gamma$ & Lagrange multipliers of the entropy closure problem \\
		$\mathscr H$ & Heaviside step function \\
		$N_i$ & number of samples \\
		$f_\mathrm{ref}$, $\mathcal Q_\mathrm{ref}$ & referenced distribution function and collision term in the dataset \\
		$N_h$ & number of hidden layers \\
		CFL & Courant–Friedrichs–Lewy number \\
		\hline
		%\hline
	\end{tabular*}
	\label{table:nomenclature}
\end{table}

\newpage

\section{Introduction}

The advances in machine learning are increasingly incorporated in computational physics to enable versatile data-driven modeling and simulation.
Interestingly, not so long ago, many were convinced that there are significant differences between scientific computing and machine learning tasks.
For example, scientific computing generally processes smaller but more computationally intensive amounts of data.
These data from different locations are expected to be linked by physical principles, whereas classical machine learning mainly deals with discrete and localized data.
However, recent progress, e.g., physics-informed neural networks \cite{raissi2019physics}, suggests that there is ample room for building data-efficient and physics-enhanced machine learning approaches for scientific applications at the intersection of classical numerical methods.

Among possible applications, a key use of scientific machine learning (SciML) is to help establish reliable physical models.
This approach is especially crucial when first-principle models are missing, e.g.,
fitting potential-energy hypersurfaces whose topographical features are sufficiently close to those of real but unknown surfaces such that the simulation results from molecular dynamics (MD) or Monte-Carlo (MC) methods are experimentally meaningful \cite{raff2012neural,zhang2018deep}.
Another promising direction is to seek simplified models for high-dimensional and nonlinear terms.
An appropriate surrogate model can significantly reduce the dimensionality and complexity of differential and integral equations, thereby improving computational efficiency while maintaining satisfactory accuracy \cite{karniadakis2021physics,rackauckas2020universal}.

The SciML-integrated model is expected to be as universally accurate as possible to the ground-truth physics.
It should state fundamental physical laws, e.g., the conservation of mass in fluid dynamic equations.
It is desirable to preserve the structural properties under physical constraints, e.g., symmetry, invariance, and non-decreasing entropy.
Different from the brute-force solution, this requires us to design physics-oriented architectures, parameters, and connections in the machine learning model, or seek efficient regularization approaches \cite{greydanus2019hamiltonian,bekkers2019b,shuaibi2021rotation,schotthofer2022structure}.
Ideally, one would like to perform a small set of experiments under idealized situations and obtain a model that works under more general conditions \cite{han2020integrating}.
In other words, we expect the trained model to extrapolate well for both in- and out-of-distribution cases \cite{xu2020neural} and the generalization gap to be as small as possible \cite{shen2021towards}.
This further requires that the model should have clear physical meaning and interpretability rather than pure black-box nature.

In this paper, we focus on the kinetic theory of gases.
Lying at its core, the Boltzmann equation describes the evolution of a many-particle system in a statistical manner,
which serves as a cornerstone of many hydrodynamic and thermodynamic theories \cite{torrilhon2016modeling,eu1992kinetic}.
The Boltzmann equation is an integro-differential equation defined in a high-dimensional phase space.
The extremely high dimensionality and nonlinearity pose a notorious challenge for theoretical analysis and numerical simulation.
Starting from the introduction of the relaxation-type model \cite{bhatnagar1954model}, tremendous efforts have been made to seek a sufficiently accurate simplified model of the Boltzmann equation.
To further improve the accuracy in highly dissipative regions and preserve the correct fluid dynamic limit, without changing the form of the relaxation model, explorations have been conducted to modify the objective equilibrium function \cite{holway1966new,shakhov1968generalization} and the relaxation frequency \cite{mieussens2004numerical,haack2021consistent}.
Although these improved models can provide the correct Prandtl number, the accuracy in describing strongly non-equilibrium flows is less satisfactory \cite{shen2006rarefied}.

The burgeoning scientific machine learning offers an alternative for establishing simplified models of the Boltzmann equation.
The existence of the multifold integral term precludes direct use of automatic differentiation-based techniques \cite{raissi2019physics,han2018solving}.
Preliminary attempts have been made to build simplified surrogate models for the Boltzmann collision integral.
The underlying techniques of related work involve neural differential equations \cite{xiao2021using}, reduced order modeling \cite{alekseenko2022fast}, and spectral representations \cite{miller2022neural}.
Despite the good results obtained by these methods, very limited proof has been provided showing that the numerical method preserves the key structures, e.g., the H-theorem \cite{chapman1970mathematical}, of a many-particle system.
Therefore, the preservation of such physical constraints must rely on the quality of training and the reliability of the dataset.
It is notoriously known that both tasks are difficult to be perfect, and the generalization performance of machine learning models based on this paradigm can be questionable.

In this paper, we exploit the structural similarity of the BGK-type relaxation model and the residual neural network (ResNet) to construct the relaxation neural network (RelaxNet), which plays as a surrogate model for the Boltzmann collision integral with quantitative interpretability.
Specifically, two feed-forward neural networks with physics-informed connections and bounded activation functions are introduced as building blocks to provide physically realizable approximations of the parameterized equilibrium distribution and velocity-dependent relaxation time.
The evaluation of collision term is thus significantly accelerated since the computation of convolutions in the fivefold integral is replaced by the tensor multiplication in the neural network.
The advection operator (left-hand side) of the original Boltzmann equation and the neural network-based collision term are fused into a trainable framework, i.e., the so-called universal Boltzmann equation (UBE) \cite{rackauckas2020universal,xiao2021using}.
We prove that the RelaxNet-based UBE preserves key structural properties.
The positivity, invariance, correct fluid dynamic limit, and hyperbolicity of advection operator hold independent of trainable parameters.
The conservation and H-theorem are preserved up to the training error.
The implementation of the RelaxNet-based UBE, training strategy, and solution algorithm are presented in detail.

The paper is organized as follows.
In Section 2, we introduce some fundamental concepts in the kinetic theory of gases.
Section 3 presents the idea and architecture of RelaxNet as well as the training method.
Section 4 describes the methodology for generating datasets for training and testing.
Section 5 details the numerical solution algorithm for the RelaxNet-based universal Boltzmann equation.
Section 6 contains the numerical experiments for both spatially homogeneous and inhomogeneous cases to validate the current methodology. 
The last section is the conclusion.
The nomenclature of this paper can be found in Table \ref{table:nomenclature}.

\section{Kinetic Theory}\label{sec:theory}

The kinetic theory of gases describes the evolution of a many-particle system in terms of the particle distribution function $f(t,\mathbf x,\mathbf v)$, where $t\in \Gamma \subseteq \mathbb R^+$ is time variable, $\mathbf x \in \mathcal D \subseteq \mathbb R^3$ is space variable, and $\mathbf v \in \mathbb R^3$ is particle velocity..
For dilute monatomic gas in the absence of external force, the Boltzmann equation can be written as
\begin{equation}
    \partial_t f + \mathbf v \cdot \nabla_\mathbf x f = \mathcal Q(f,f) = \int_{\mathbb{R}^{3}} \int_{\mathbb S^{2}}\left[f\left(\mathbf{v}^{\prime}\right) f\left(\mathbf{v}_{*}^{\prime}\right)-f(\mathbf{v}) f\left(\mathbf{v}_{*}\right)\right] \mathcal{B}(\cos \theta, g) d \mathbf \Omega d \mathbf{v}_{*},
\label{eqn:boltzmann}
\end{equation}
where $\{\mathbf v, \mathbf v_*\}$ denotes the pre-collision velocities of two classes of colliding particles, and $\{\mathbf v', \mathbf v_*'\}$ is the corresponding post-collision velocities.
The collision kernel $\mathcal{B}(\cos \theta, g)$ measures the probability of collisions in different directions, where $\theta$ is the deflection angle and $g = |\mathbf g| = |\mathbf v - \mathbf v_*|$ is the magnitude of relative pre-collision velocity.
The solid angle $\mathbf \Omega$ denotes the unit vector along the relative post-collision velocity $\mathbf v' - \mathbf v_*'$, and the deflection angle satisfies the relation $\theta=\mathbf \Omega \cdot \mathbf g / g$.
Note that if we define the collision frequency as
\begin{equation}
    \nu(\mathbf v)=\int_{\mathbb{R}^{3}} \int_{\mathbb{S}^{2}} f\left(\mathbf v_{*} \right) \mathcal B\left(\cos \theta,g\right) d \mathbf \Omega d \mathbf v_{*},
    \label{eqn:collision frequency}
\end{equation}
the collision term in Eq.(\ref{eqn:boltzmann}) can be written as a combination of gain and loss, i.e.,
\begin{equation}
\begin{aligned}
    &\mathcal Q(f,f) = \mathcal Q^+(f,f) + \mathcal Q^-(f), \\
    &\mathcal Q^+ = \int_{\mathbb{R}^{3}} \int_{\mathbb S^{2}} f\left(\mathbf{v}^{\prime}\right) f\left(\mathbf{v}_{*}^{\prime}\right) \mathcal{B}(\cos \theta, g) d \mathbf \Omega d \mathbf{v}_{*}, \quad
    \mathcal Q^- = \nu(\mathbf v) f(\mathbf v).
\end{aligned}
\label{eqn:boltzmann gainloss}
\end{equation}

The Boltzmann equation possesses structural properties that reflect the physical constraints of a many-particle system.
Such structure is manifested in the properties of the collision operator $\mathcal Q$ and the advection operator $\mathcal A = \partial_t + \mathbf v \cdot \nabla_{\mathbf x}$, and it plays a key role in design of the analytical and numerical solution algorithms of the Boltzmann equation.
These structural properties can be highlighted as follows \cite{alldredge2019regularized}.

\begin{enumerate}
  \item[(1)] \textit{Invariant range}:
  There exists a set $B \subseteq [0, \infty )$ consistent with the physical bounds of $f$ such that $\mathrm{Range}(f(t, \cdot , \cdot )) \subseteq B$ if $\mathrm{Range}(f(0, \cdot , \cdot )) \subseteq B$. A minimum requirement for the invariant range is $f \geq 0$ since $f$ represents the density of particles.
  \item[(2)] \textit{Conservation}: There exists the collision invariants $\boldsymbol\psi=(1,\mathbf v,\mathbf v^2/2)^T$ such that
  \begin{equation}
    \langle \boldsymbol\psi \mathcal Q(g,g) \rangle = 0, \quad \forall g\in \mathrm{Dom}(\mathcal Q),
  \end{equation}
  where $\langle\cdot\rangle=\int_{\mathbb R^3} \cdot d\mathbf v$ denotes the integral over particle velocity space.
  This property leads to the conservation laws
  \begin{equation}
      \partial_t \langle \boldsymbol\psi f \rangle + \nabla_\mathbf x \cdot \langle \mathbf v \boldsymbol\psi f \rangle = 0,
      \label{eqn:conservation laws}
  \end{equation}
  where $\mathbf{W}(t, \mathbf{x})=(\rho,\rho\mathbf V,\rho E)^T=\langle \boldsymbol\psi f \rangle$ are named as conservative variables in fluid mechanics.
  These velocity-space integrals are also commonly referred as moments of particle distribution function $f$.
  \item[(3)] \textit{Hyperbolicity}: For each fixed $\mathbf v$, the advection operator $\mathcal A$ is hyperbolic over $(t, \mathbf x) \in \Gamma \times \mathcal D$.
  \item[(4)] \textit{H-theorem}: There exists a differentiable function $H(f)=-f\log f$ such that
  \begin{equation}
      \langle H'(g) \mathcal Q(g,g) \rangle \leq 0, \quad \forall g\in \mathrm{Dom}(\mathcal Q) \ \mathrm{s.t.} \ \mathrm{Range}(g) \subseteq \mathcal D,
  \end{equation}
  which further results in the dissipation law
  \begin{equation}
      \partial_t \langle H(f) \rangle + \nabla_\mathbf x \cdot \langle \mathbf v H(f) \rangle \leq 0.
  \end{equation}
  The H-theorem implies that $H$ is a Lyapunov function for the Boltzmann equation, and the logarithm of its minimizer $\mathcal M$ must be a linear combination of the collision invariants $\boldsymbol\psi$,
    which can be further proved to have the following form, i.e., the Maxwellian distribution function \cite{bouchut2000kinetic},
    \begin{equation}
    	\mathcal M := \rho\left(\frac{m}{2\pi k T}\right)^{3/2} \exp\left(-\frac{m}{2kT} \mathbf c^2\right),
    	\label{eqn:maxwellian}
    \end{equation}
    where $m$ is molecular mass, $k$ is the Boltzmann constant, $T$ is temperature, and $\mathbf c=\mathbf v-\mathbf V$ is the peculiar velocity.
  \item[(5)] \textit{Galilean invariance}: There exists the Galilean transformation defined by 
  \begin{equation}
      \mathcal G(g)(t,\mathbf x,\mathbf v)=g(t,\mathbf x-\boldsymbol\omega t,\mathbf v-\boldsymbol\omega),
      \label{eqn:galilean}
  \end{equation}
  where $\boldsymbol \omega \in \mathbb R^3$ is a translational velocity, that commute with the advection and collision operators, i.e.,
  \begin{equation}
    \begin{aligned}
        &\mathcal A(\mathcal G(g)) = \mathcal G(\mathcal A(g)), \quad \forall g \in \mathrm{Dom}(\mathcal A), \\
        &\mathcal Q(\mathcal G(g)) = \mathcal G(\mathcal Q(g)), \quad \forall g \in \mathrm{Dom}(\mathcal Q). \\
    \end{aligned}
  \end{equation}
  As a consequence, the transformed particle distribution function $g=\mathcal G(f)$ still satisfies Eq.(\ref{eqn:boltzmann}) if $f$ is a solution of it.
\end{enumerate}

Following the structure in Eq.(\ref{eqn:boltzmann gainloss}) and the H-theorem, several relaxation models have been developed to simplify the computation of the fivefold integral on the right-hand side of the Boltzmann equation.
These simplified models take the form
\begin{equation}
    \mathcal Q(f) = \nu(\mathcal E - f),
    \label{eqn:relaxation model}
\end{equation}
where $\mathcal E$ denotes the target equilibrium state and $\nu$ is the relaxation frequency.
In the Bhatnagar–Gross–Krook (BGK) model \cite{bhatnagar1954model}, the equilibrium distribution $\mathcal E$ is set as Maxwellian, and $\nu$ is independent of $\mathbf v$, which implies that particles with different velocities share the same probability of collision.
Although the BGK model can satisfy the structural properties shown in Section \ref{sec:theory},
the underlying assumption becomes particularly unreasonable when the magnitude of $\nu$ in rarefied gases is moderate, which makes the BGK model inaccurate in describing highly dissipative flows.
One direction to improve the BGK model is to introduce $\mathbf v$-dependent relaxation frequency in Eq.(\ref{eqn:relaxation model}) \cite{mieussens2004numerical,haack2021consistent}.
Following the principle that high-speed particles should have a greater probability of collision, these models attempt to recover the true collision frequency in Eq.(\ref{eqn:collision frequency}).
Another idea is to supplement correction terms in $\mathcal E$ based on high-order moments of $f$, e.g., stress tensor \cite{holway1966new} and heat flux \cite{shakhov1968generalization}, which builds a multi-level relaxation mechanism towards Maxwellian.
It is shown that both strategies above can provide the correct Prandtl number.
However, no model has yet provided as satisfactory the solution as the full Boltzmann equation in describing the evolution of gases far away from equilibrium.
%Necessary and insufficient conditions

\section{Relaxation Neural Network}

\subsection{Deep residual learning}

%Universal approximation theorems demonstrates the approximation capabilities of neural networks on the space of continuous functions between two Euclidean spaces when given appropriate weights \cite{hornik1989multilayer}.
Our idea is to build a surrogate model of the nonlinear function described by the Boltzmann equation based on a neural network.
To preserve as many structural properties as shown in Section \ref{sec:theory}, we prefer a physics-oriented architecture rather than a naive feed-forward neural network.

We choose to construct the model on the basis of residual neural network (ResNet).
To clarify this, we briefly introduce here the basic principles of ResNets.
The concept of residual learning was proposed to mitigate the degradation problem in deep neural networks, i.e., as the network depth increases, the accuracy gets saturated and then degrades rapidly.
This phenomenon implies that it is non-trivial to learn a deeper model based on its shallower counterpart in which the added layers only perform identity mapping.
Therefore, instead of building each few stacked layers to fit a desired underlying mapping directly, one can explicitly let these layers fit a residual mapping.
It has been shown that it is easier to optimize the residual mapping than the original unreferenced mapping \cite{he2016deep}.

Figure \ref{fig:resnet} shows the schematic diagram of the building block in a ResNet.
We denote $\mathcal H(\mathbf u)$ as the underlying mapping to be
fit by a few stacked layers and $\mathbf u$ as the input to the first layer.
If a continuous function can be approximated asymptotically by multiple nonlinear layers, then we can equivalently conclude that they can asymptotically approximate the residual function, i.e., $\mathcal F(\mathbf u) = \mathcal H(\mathbf u) - \mathbf u$, provided that the input and output are of the same dimensions.
The original function thus becomes $\mathcal H(\mathbf u)=\mathcal F(\mathbf u)+\mathbf u$, and its formulation is realized by adding shortcut connections \cite{bishop1995neural} to feedforward or convolutional neural networks.
The shortcut can skip one or several hidden layers and performs identity mapping.
The nonlinear activation function inside the hidden layers is denoted as $\sigma_h$.
The addition operation can be regarded as the connection function which combines the contributions from $\mathcal F$ and identity mapping of $\mathbf u$.
The residual block is further activated with $\sigma$, and the final output can be written as
\begin{equation}
    \mathcal N(\mathbf u) = \sigma (\mathcal F(\mathbf u) + \mathbf u).
    \label{eqn:resnet}
\end{equation}

\begin{figure}[htb!]
    \centering
    \includegraphics[width=0.6\textwidth]{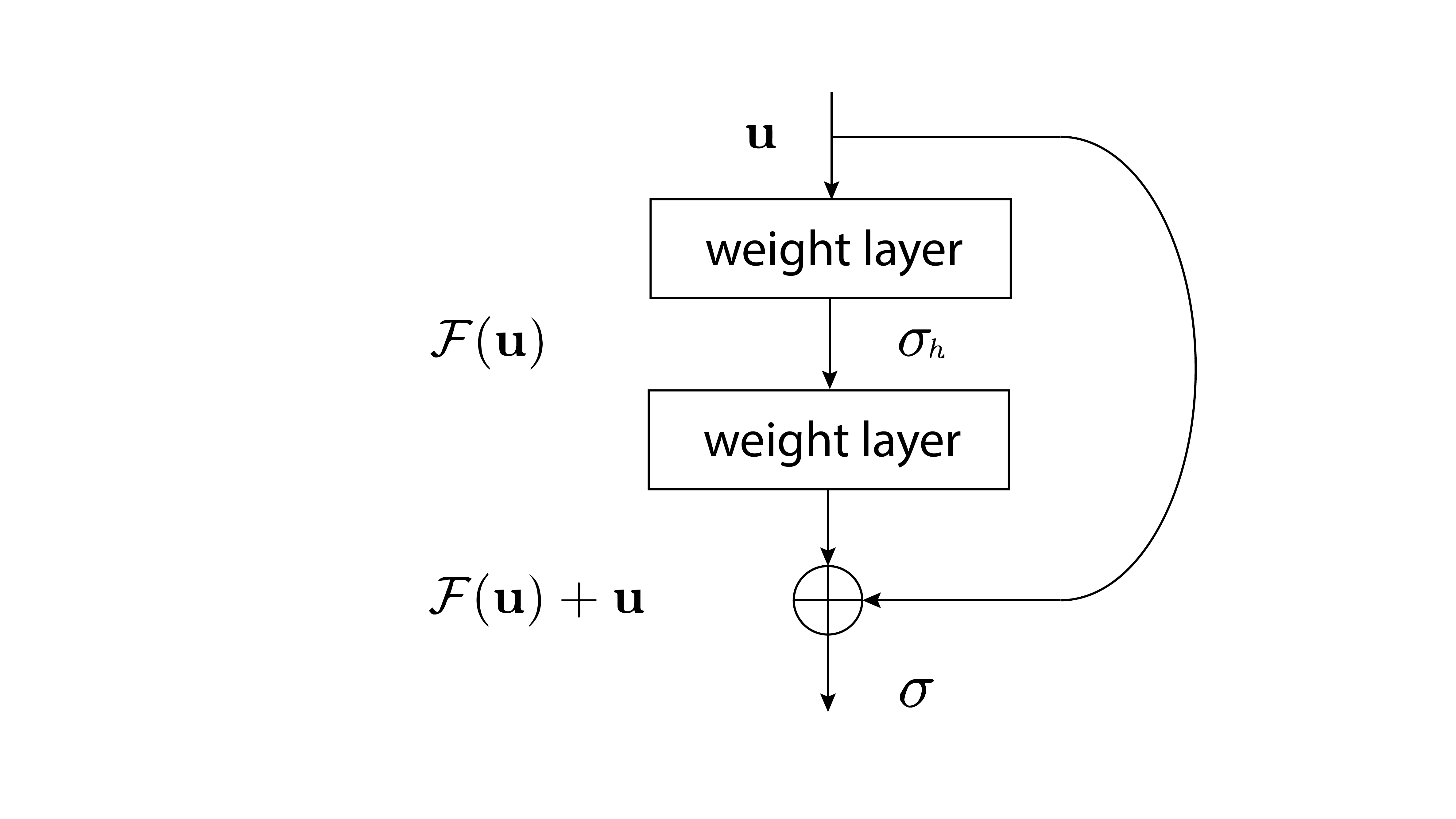}
    \caption{Schematic diagram of a building block in the residual neural network (ResNet).}
    \label{fig:resnet}
\end{figure}

It is noticeable that the relaxation model of the Boltzmann equation in Eq.(\ref{eqn:relaxation model}) has a similar structure to the ResNet.
Specifically, if the particle distribution function is considered as the input to the solver of relaxation term, then the equilibrium distribution function $\mathcal E = \mathcal E(f)$ plays an equivalent role to that of $\mathcal H$.
Next, unlike in neural networks where addition is employed as the connection function, subtraction is used in the relaxation model.
Finally, the relaxation frequency can be viewed as an activation function acting on the equilibrium function $\mathcal E$ and the shortcut connection from $f$.
The building block of the solver of relaxation term can thus be written as
\begin{equation}
    \mathcal Q(f) = \nu (\mathcal E(f) - f).
    \label{eqn:resnet relaxation model}
\end{equation}
The structural similarity between the relaxation model of Boltzmann equation in Eq.(\ref{eqn:resnet relaxation model}) and ResNet in Eq.(\ref{eqn:resnet}) can be easily noticed.
This provide us an opportunity to build a physics-oriented neural network based on the relaxation model and preserve the structural properties of the Boltzmann equation.

\subsection{Relaxation neural network}\label{sec:relaxnet}

Following the spirit of the relaxation model and ResNet, we build a novel neural network model which is named as relaxation neural network (RelaxNet).
The neural network is built on top of the discrete velocity formulation of the Boltzmann equation, i.e., the particle distribution function is discretized in the velocity space and the macroscopic variables as its moments can evaluated by numerical quadrature.
The schematic diagram of the model with forward pass is shown in Figure \ref{fig:relaxnet}.
As can be seen, RelaxNet consists of two feedforward subnets and the corresponding activation and connection functions.
We describe these components in detail below.

\begin{figure}[htb!]
    \centering
    \includegraphics[width=0.95\textwidth]{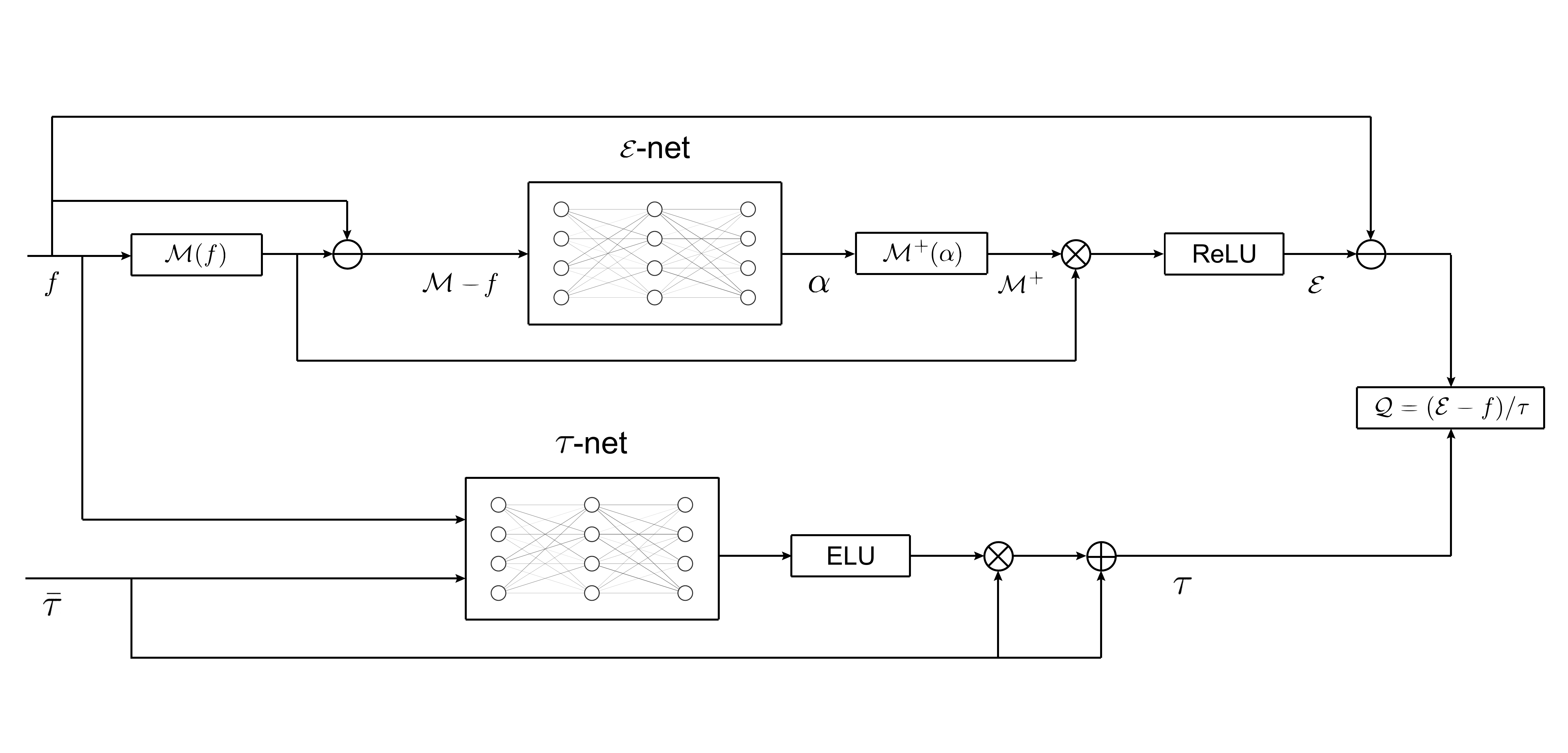}
    \caption{Schematic diagram of forward pass in the relaxation neural network (RelaxNet).}
    \label{fig:relaxnet}
\end{figure}

\begin{enumerate}
    \item[(1)] Neural networks
    \begin{itemize}
        \item{$\mathcal E$-net}: feedforward network $\mathcal N_\mathcal E : \mathbb R^{N_v} \rightarrow \mathbb R^{D+2}$, where $N_v$ is the number of quadrature points to discretize $f$ in the velocity space and $D\in\{1,2,3\}$ is the dimension of interest.
        The neural network takes the difference between Maxwellian and distribution function, i.e., $\mathcal M - f$, as input.
        The output of neural network is a vector $ \boldsymbol \alpha=(\alpha_1,\boldsymbol \alpha_2,\alpha_3)^T$, where $\boldsymbol \alpha_2$ is a vector of dimension $D$.
        The hyperbolic tangent (tanh) function is employed as the activation function for the hidden layers.
        We require the $\mathcal E$-net to have only weights without biases.
        \item{$\tau$-net}: feedforward network $\mathcal N_\tau : \mathbb R^{N_v+1} \rightarrow \mathbb R^{N_v}$.
        The input to the neural network $(\mathcal M-f, \bar \tau)^T$ is a vector of dimension $N_v+1$.
        Here $\bar\tau$ is the mean relaxation time calculated by
        \begin{equation}
            \bar \tau = \frac{\mu}{p},
        \end{equation}
        where $\mu$ is the viscosity that can be decided from a molecular model, and $p$ is the pressure.
        The relationship between mean relaxation time and frequency is $\bar\nu \bar\tau=1$.
        The output of neural network $\boldsymbol\beta$ is a vector of the same dimension as $f$.
        The hyperbolic tangent (tanh) function is used as the activation function for the hidden layers.
        We require the $\tau$-net to have only weights without biases.
    \end{itemize}
    \item[(2)] Structure functions
    \begin{itemize}
        \item{$\mathcal M(f)$}: Maxwellian function described in Eq.(\ref{eqn:maxwellian}), where the macroscopic variables can be evaluated as
        \begin{equation}
            \rho = \langle f \rangle, \quad 
            \mathbf V = \frac{\langle \mathbf v f \rangle}{\langle f \rangle}, \quad
            T = \frac{m}{3k\rho} \langle \mathbf c^2 f \rangle.
        \end{equation}
        \item{$\mathcal M^+(\boldsymbol\alpha)$}: Constructor of the correction term for equilibrium distribution via $\boldsymbol\alpha$, i.e.,
        \begin{equation}
            \mathcal M^+ = \exp \left(\alpha_1 + \boldsymbol\alpha_2 \cdot \mathbf v+\frac{1}{2}\alpha_3 \mathbf v^2\right),
            \label{eqn:maxwellian plus}
        \end{equation}
    \end{itemize}
    \item[(3)] Connection and activation functions
    \begin{itemize}
        \item{ELU}: Exponential Linear Unit function defined as
        \begin{equation}
            R(z)=\left\{\begin{array}{cc}
            z,  \quad & z>0, \\
            e^z-1, \quad & z<=0. \\
            \end{array}\right.
        \end{equation}
        \item{ReLU}: Rectified Linear Unit function defined as
        \begin{equation}
            R(z)=\left\{\begin{array}{cc}
            z,  \quad & z>0, \\
            0, \quad & z<=0. \\
            \end{array}\right.
        \end{equation}
        \item{$\oplus$}: Addition.
        \item{$\ominus$}: Subtraction.
        \item{$\otimes$}: Multiplication.
    \end{itemize}
\end{enumerate}

Based on the proposed RelaxNet, a neural network enhanced Boltzmann equation can be formulated as
\begin{equation}
    \partial_t f + \mathbf v \cdot \nabla_{\mathbf x} f = \mathrm{NN}_\theta (f,\bar\tau),
    \label{eqn:ube}
\end{equation}
where $\theta$ denotes trainable parameters in the model.
Such an equation is consistent of the idea of universal differential equation \cite{rackauckas2020universal}, where the mechanical advection operator $\mathcal A = \partial_t + \mathbf v \cdot \nabla_{\mathbf x}$ and the neural collision operator $\mathrm{NN}_\theta$ constitute a differentiable surrogate model.
Therefore, we can call Eq.(\ref{eqn:ube}) as the universal Boltzmann equation (UBE).

\subsection{Training strategy}\label{sec:loss}

UBE in Eq.(\ref{eqn:ube}) is a data-driven model and it forms a typical supervised training task.
Given a dataset consisting of reference solutions to the Boltzmann equation, an optimization problem can be set up which minimizes the difference between the current prediction of UBE model and ground-truth solution.
We define the cost function as
\begin{equation}
    C(\theta) = \sum_{n,i,j} \Arrowvert \mathrm{NN}_\theta - \mathcal Q_\mathrm{ref} \Arrowvert (t^n,\mathbf x_i,\mathbf v_j) + \lambda \sum_{n,i} \Arrowvert \langle \boldsymbol\psi \mathrm{NN}_\theta \rangle \Arrowvert (t^n,\mathbf x_i) + \xi \sum_l^L \Arrowvert \theta^l \Arrowvert,
    \label{eqn:cost function}
\end{equation}
where $\Arrowvert\cdot\Arrowvert$ denotes the Euclidean distance, and $\{n,i,j\}$ are the indices of the solution collocation points at $\{t^n, \mathbf x_i, \mathbf v_j\}$.
The generation of datasets used for the above function will be detailed in Section \ref{sec:data}.

The latter two terms in Eq.(\ref{eqn:cost function}) serve as regularization terms to mitigate overfitting.
The second term in Eq.(\ref{eqn:cost function}) comes naturally from the physical constraint (2) shown in Section \ref{sec:theory},
i.e., the collision term should conserve mass, momentum and energy in the Boltzmann equation.
This physics-informed regularization \cite{nabian2018physics} leads to more accurate preservation of conservation laws, and improves the physical interpretability of the model.
The last term in Eq.(\ref{eqn:cost function}) sums over the squared weight parameters of the network, where $\theta^l$ denotes the weight parameters of $l$-th layers and $L$ is the total number of layers.
The parameters $\{\lambda,\xi\}$ are the regularization strengths.
We require
\begin{equation}
    \lambda N_t N_x > 1, \quad \xi N_t N_x N_v > 1,
    \label{eqn:regularization strength}
\end{equation}
where $\{N_t,N_x,N_v\}$ are the number of time steps, physical cells, and velocity points respectively in the dataset.

The cost function in Eq.(\ref{eqn:cost function}) can then be minimized by different optimization algorithms.
In this paper, we adopt the Adam \cite{kingma2014adam} method which is an improvement of the stochastic gradient descent method with adaptive moment estimation.
The gradient of Eq.(\ref{eqn:cost function}) with respect to $\theta$ is computed with the help of reverse-mode automatic differentiation (AD).
Consider a smooth function $y = \mathscr F(x)$, the reverse-mode AD computes the dual (conjugate-transpose) matrix of Jacobian $\mathcal J = \nabla \mathscr F$ at $x = x_0$ with the chain rule
\begin{equation}
    \left(\mathcal{J}(\mathscr F)\left(x_{0}\right)\right)^{*}=\left(\mathcal{J}\left(G_{1}\right)\left(x_{0}\right)\right)^{*} \times \cdots \times\left(\mathcal{J}\left(G_{k}\right)\left(x_{k-1}\right)\right)^{*},
\end{equation}
with $x_{i}:=G_{i}\left(x_{i-1}\right) \text { for } i=1, \ldots, k-1$.
The implementation of the reverse-mode AD can be found in the open-source library Zygote \cite{innes2019differentiable}.

\subsection{Structure of RelaxNet}

The universal Boltzmann equation based on RelaxNet preserve the structural properties illustrated in Section \ref{sec:theory}.
We prove these properties in detail below.

\begin{theorem}\label{theorem:1}
Let \(f\) be the solution of the RelaxNet-based universal Boltzmann equation given in Eq.(\ref{eqn:ube}), then the range of $f$ satisfy
\begin{equation}
    \mathrm{Range}(f(t,\mathbf x,\mathbf v)) \subseteq B=\left\{ f: 0 \leq f <\infty \right\},
\end{equation}
given $\mathrm{Range}(f(0,\mathbf x,\mathbf v)) \subseteq B$.
\end{theorem}

\begin{proof}
We define shifted $f$ along characteristics as
\begin{equation}
    f^{\#}(t,\mathbf x,\mathbf v)=f(t,\mathbf x+\mathbf v t,\mathbf v).
\end{equation}
Then UBE in Eq.(\ref{eqn:ube}) can be written as
\begin{equation}
    \frac{d}{dt} f^{\#} + \nu f^{\#} = \nu \mathcal E^{\#},
    \label{eqn:bgk ode}
\end{equation}
where $d/dt$ denotes the full derivative.
Given the initial value
\begin{equation}
    f(0,\mathbf x,\mathbf v)=f_0(\mathbf x,\mathbf v),
    \label{eqn:initial value}
\end{equation}
then the solution of the initial value problem consisting of Eq.(\ref{eqn:bgk ode}) and (\ref{eqn:initial value}) can be written as \cite{babovsky1998boltzmann}
\begin{equation}
    f^{\#}(t,\mathbf x,\mathbf v) = f_0(\mathbf x,\mathbf v) e^{-\nu t} + \int_0^t e^{-\nu (t-s)} \nu \mathcal E^{\#}(s,\mathbf x,\mathbf v) ds.
\end{equation}
Proving Theorem \ref{theorem:1} is equivalent to proving that $\mathrm{Range}(f^{\#}) \subseteq B$.
As shown in Section \ref{sec:relaxnet}, $\mathcal E$ and $\nu$ can be expressed as
\begin{equation}
\begin{aligned}
    &\mathcal E=\mathrm{ReLU}(\mathcal M \mathcal M^+), \quad \mathcal M^+=\exp (\boldsymbol\alpha_{\theta_1} \cdot \boldsymbol\psi), \\
    &\nu = \frac{1}{\tau}, \quad \tau = \bar\tau (1 + \mathrm{ELU} (\boldsymbol\beta_{\theta_2})),
\end{aligned}
\label{eqn:etau calculation}
\end{equation}
where $\theta=(\theta_1,\theta_2)^T$ denotes all the trainable parameters.
With the help of the physics-informed activation functions, i.e., ELU and ReLU, it is clear that $\mathcal E$ and $\nu$ in Eq.(\ref{eqn:etau calculation}) are non-negative independent of $\theta$.
Besides, since the hidden layers inside $\mathcal E$-net and $\tau$-net are activated by the tanh function, the following relation holds,
\begin{equation}
    \mathrm{Range}(\mathcal E) \subseteq B, \quad \mathrm{Range}(\nu) \subseteq B
\end{equation}
given any finite layers and parameters.
So far, Theorem \ref{theorem:1} is proved.
\end{proof}

\begin{theorem}\label{theorem:2}
Let $\mathrm{NN}_\theta$ be the optimized RelaxNet, then the moments of $\mathrm{NN}_\theta$ with respect to the collision variants $\boldsymbol\psi$ satisfy
\begin{equation}
    \Arrowvert \langle \boldsymbol\psi \mathrm{NN}_\theta (g,\bar\tau) \rangle \Arrowvert \leq \varepsilon , \quad \forall (g,\bar\tau) \in \mathrm{Dom}(\mathrm{NN}_\theta),
\end{equation}
where $\varepsilon$ is the minimum value of the cost function.
\end{theorem}

\begin{proof}
The optimized cost function in Eq.(\ref{eqn:cost function}) satisfy
\begin{equation}
    C(\theta) = \sum_{n,i,j} \Arrowvert \mathrm{NN}_\theta - \mathcal Q \Arrowvert (t^n,\mathbf x_i,\mathbf v_j) + \lambda \sum_{n,i} \Arrowvert \langle \boldsymbol\psi \mathrm{NN}_\theta \rangle \Arrowvert (t^n,\mathbf x_i) + \xi \sum_l^L \Arrowvert \theta^l \Arrowvert = \varepsilon,
\end{equation}
which results in
\begin{equation}
    \lambda \sum_{n,i} \Arrowvert \langle \boldsymbol\psi \mathrm{NN}_\theta \rangle \Arrowvert (t^n,\mathbf x_i) \leq \varepsilon.
\end{equation}
Due to the requirement in Eq.(\ref{eqn:regularization strength}), we have
\begin{equation}
    \Arrowvert \langle \boldsymbol\psi \mathrm{NN}_\theta \rangle \Arrowvert \leq \varepsilon,
\end{equation}
at any solution point in the dataset.
\end{proof}

\begin{theorem}\label{theorem:3}
For each fixed $\mathbf v$, the advection operator $\mathcal A$ in the RelaxNet-based universal Boltzmann equation is hyperbolic over $(t,\mathbf x)\in \Gamma \times \mathcal D$.
\end{theorem}

\begin{proof}
As RelaxNet surrogates only the right-hand side of the Boltzmann equation, this theorem naturally holds according to Property (3) in Section \ref{sec:theory}.
\end{proof}

\begin{theorem}\label{theorem:4}
The function $H(f)=-f\log f$ satisty
\begin{equation}
   \partial_t \langle H(f) \rangle + \nabla_\mathbf x \cdot \langle \mathbf v H(f) \rangle \leq \sum_i \chi_i \varepsilon,
\end{equation}
where $\boldsymbol\chi=\sum_i \chi_i$ is a vector of the same dimension of $\boldsymbol \alpha$.
The H-theorem holds asymptotically as $\varepsilon \rightarrow 0$.
\end{theorem}

\begin{proof}
Substituting $H'(f)=-(1+\log f)$ into the RelaxNet-based UBE in Eq.(\ref{eqn:ube}) and considering Eq.(\ref{eqn:etau calculation}) yield
\begin{equation}
\begin{aligned}
    &\langle (\partial_t f + \mathbf v \cdot \nabla_\mathbf x f)(1+\log f) \rangle = \partial_t \langle H(f) \rangle + \nabla_\mathbf x \cdot \langle \mathbf v H(f) \rangle = \langle \nu (\mathcal E - f)(1+\log f) \rangle.
\end{aligned}
\label{eqn:entropy law}
\end{equation}
Notice that the Maxwellian distribution in Eq.(\ref{eqn:maxwellian}) can be written as
\begin{equation}
    \mathcal M = \exp(\hat{\boldsymbol\alpha} \cdot \boldsymbol\psi),
\end{equation}
where 
\begin{equation}
\begin{aligned}
    &\hat{\boldsymbol\alpha}=(\hat\alpha_1,\hat{\boldsymbol\alpha}_2,\hat\alpha_3)^T, \\
    &\hat\alpha_1=\log \left( \rho \left(\frac{m}{2\pi k T}\right)^{D/2} - \frac{m \mathbf V^2}{2kT} \right), \quad
    \hat{\boldsymbol\alpha}_2=\frac{m\mathbf V}{kT},\quad \hat\alpha_3=-\frac{m}{kT}.
\end{aligned}
\label{eqn:maxwellian alpha}
\end{equation}
Given the structure function $\mathcal M^+$ defined in Eq.(\ref{eqn:maxwellian plus}), the equilibrium state can be written as
\begin{equation}
    \mathcal E = \mathcal M \mathcal M^+ = \exp(\tilde{\boldsymbol\alpha} \cdot \boldsymbol\psi)
\end{equation}
where $\tilde\alpha_i=\alpha_i \hat\alpha_i$.
From Theorem \ref{theorem:2}, we obtain
\begin{equation}
    \langle \nu (\mathcal E - f)\log \mathcal E \rangle = \left\langle \nu (\mathcal E - f) \left(\tilde\alpha_1 + \tilde{\boldsymbol\alpha}_2 \cdot \mathbf v + \frac{1}{2} \tilde\alpha_3 \mathbf v^2 \right) \right\rangle \leq \sum_i \tilde\alpha_i \varepsilon.
    \label{eqn:entropy dissipation}
\end{equation}
The right-hand side of Eq.(\ref{eqn:entropy law}) can be written as
\begin{equation}
    \langle \nu (\mathcal E - f)(1+\log f) \rangle = \langle \nu (\mathcal E - f) \rangle + \langle \nu (\mathcal E - f)\log f \rangle \leq \varepsilon + \langle \nu (\mathcal E - f)\log f \rangle.
\end{equation}
We rewrite the last term of the above equation as
\begin{equation}
    \langle \nu (\mathcal E - f)\log f \rangle = \langle \nu (\mathcal E - f)(\log f - \log \mathcal E) \rangle +  \langle \nu (\mathcal E - f)\log \mathcal E \rangle.
    \label{eqn:entropy inequality}
\end{equation}
Since the logarithm is a monotonically increasing function, the first in the above equation is non-positive.
Considering Eq.(\ref{eqn:entropy law}), (\ref{eqn:entropy dissipation}), and (\ref{eqn:entropy inequality}), we end up with
\begin{equation}
    \partial_t \langle H(f) \rangle + \nabla_\mathbf x \cdot \langle \mathbf v H(f) \rangle \leq \left( 1 + \sum_i \tilde\alpha_i \right) \varepsilon.
\end{equation}
So far, Theorem \ref{theorem:4} is proved.
\end{proof}

\begin{theorem}\label{theorem:5}
The RelaxNet-based universal Boltzmann equation is Galilean invariant.
\end{theorem}

\begin{proof}
We assume $f$ is a solution to UBE in Eq.(\ref{eqn:ube}) and consider the Galilean transformation defined in Eq.(\ref{eqn:galilean}), i.e.,
\begin{equation}
    g(t,\mathbf x,\mathbf v) = \mathcal G(f)=f(t,\mathbf x-\boldsymbol\omega t,\mathbf v-\boldsymbol\omega).
\end{equation}
As the advection operator in UBE is the same as the original Boltzmann equation, the commute naturally holds, i.e.,
\begin{equation}
    \mathcal A(\mathcal G(f)) = \mathcal G(\mathcal A(f)).
\end{equation}
Besides, the inputs to RelaxNet, i.e., $\{f,\bar\tau\}$, are scalar and Galilean invariant, and thus we have
\begin{equation}
    \mathcal Q(\mathcal G(f),\bar\tau) = \mathcal G(\mathcal Q(f,\bar\tau)).
\end{equation}
Therefore, $g$ is also a soslution to Eq.(\ref{eqn:ube}).
\end{proof}

\begin{theorem}\label{theorem:6}
The RelaxNet-based universal Boltzmann equation preserves the correct continuum limit.
\end{theorem}

\begin{proof}
Based on Theorem \ref{theorem:2}, we obtain the transport equations for conservative variables
\begin{equation}
    \partial_t \langle \boldsymbol\psi f \rangle + \nabla_\mathbf x \cdot \langle \mathbf v \boldsymbol\psi f \rangle \leq \varepsilon,
    \label{eqn:ube macro}
\end{equation}
or specifically,
\begin{equation}
    \frac{\partial}{\partial t}\left(\begin{array}{c}
    \rho \\
    \rho \mathbf{V} \\
    \rho E
    \end{array}\right)+\nabla_{\mathbf{x}} \cdot\left(\begin{array}{c}
    \rho \mathbf{V} \\
    \rho \mathbf{V} \mathbf{V} \\
    \rho E \mathbf{V}
    \end{array}\right) \leq -\nabla_\mathbf x \cdot
    \left(\begin{array}{c}
    0 \\
    \mathbf P \\
    \mathbf P \cdot \mathbf V + \mathbf q
    \end{array}\right) + \left(\begin{array}{c}
    \varepsilon \\
    \varepsilon \\
    \varepsilon
    \end{array}\right),
\end{equation}
where
\begin{equation}
    \mathbf P = \langle \mathbf c \mathbf c f \rangle,\quad \mathbf q = \frac{1}{2} \langle \mathbf c \mathbf c^2 f \rangle,
\end{equation}
denote the pressure tensor and heat flux respectively.
UBE in Eq.(\ref{eqn:ube}) can be rewritten as
\begin{equation}
    f = \mathcal E - \tau \mathrm{NN}_\theta(f,\bar\tau),
\end{equation}
where $\tau=1/\nu$ is the relaxation time.
In the continuum limit, infinitely many intermolecular interactions happen in a unit of time, i.e., $\nu\rightarrow\infty$.
From Eq.(\ref{eqn:etau calculation}), we have $\tau\rightarrow 0$ and thus $f\rightarrow \mathcal E$.
We next prove that $\mathcal E$ converges to $\mathcal M$ in the continuum limit.
Note that the difference between Maxwellian and distribution function, i.e., $\mathcal M - f$, are used as input to the zero-bias $\mathcal E$-net in RelaxNet.
As a result, the $\mathcal E$-net produces $\boldsymbol\alpha \rightarrow \mathbf 0$ and $\mathcal M^+ \rightarrow 1$,
and thus we have
\begin{equation}
    \mathcal E = \mathcal M \mathcal M^+ \sim \mathcal M \quad (\mathrm{as} \ f \rightarrow \mathcal M).
\end{equation}
In addition, given $\langle \boldsymbol\psi (\mathcal M - f)\rangle=\mathbf 0$, the residuals in Eq.(\ref{eqn:ube macro}) converge to zero,
\begin{equation}
    \varepsilon \rightarrow 0 \quad (\mathrm{as} \ f \rightarrow \mathcal M),
\end{equation}
and we recover the exact Euler equations independent of $\theta$, i.e.,
\begin{equation}
    \frac{\partial}{\partial t}\left(\begin{array}{c}
    \rho \\
    \rho \mathbf{V} \\
    \rho E
    \end{array}\right)+\nabla_{\mathbf{x}} \cdot\left(\begin{array}{c}
    \rho \mathbf{V} \\
    \rho \mathbf{V} \mathbf{V} \\
    \rho E \mathbf{V}
    \end{array}\right) = -\nabla_\mathbf x \cdot
    \left(\begin{array}{c}
    0 \\
    p\mathbf I \\
    p\mathbf V
    \end{array}\right),
\end{equation}
given the following relationship as $f=\mathcal M$,
\begin{equation}
    \mathbf P = p\mathbf I,\quad \mathbf q = \mathbf 0.
\end{equation}
\end{proof}

\section{Data Sampling}\label{sec:data}

As presented in Section \ref{sec:loss}, a dataset with ground-truth distribution functions $f_\mathrm{ref}$ and collision terms $\mathcal Q_\mathrm{ref}$ is needed in Eq.(\ref{eqn:cost function}) to perform the supervised learning task.
We consider the space of $f$ as a sample space under the constraint,
\begin{equation}
    f(t,\mathbf x,\mathbf v)\in B = \left\{ f: 0\leq f < \infty \right\},
\end{equation}
where $B$ can be referred as the physically realizable set.
A data-distribution $p_f$ can be incorporated to generate $f$ from $B$.

A common strategy for data-driven modeling is to first perform deterministic simulations and extract solutions in the intermediate steps as sample data.
One disadvantage of this approach is that $p_f$ is implicitly determined by the chosen test cases and can thus be heavily biased.
It may not be possible to cover enough different particle distributions and flow regimes.
The high expense of CFD simulations makes it challenging to establish an all-round database.

In this paper, we provide an alternative sampling strategy to generate particle distribution functions and corresponding collision terms.
The theoretical basis of this approach is the closure hierarchies of moment system of the Boltzmann equation \cite{levermore1996moment}.
In the following we briefly explain the principle and implementation of the approach.

\subsection{Generation of distribution function}\label{sec:sample pdf}

A closure strategy of the Boltzmann equation aims to reconstruct the particle distribution function $f$ from moments
\begin{equation}
    \mathbf m = \langle \boldsymbol\phi f \rangle,
\end{equation}
under the constraint
\begin{equation}
    f(t,\mathbf x,\mathbf v)\in B_M = \left\{ f: f \geq 0 ,\ \Arrowvert \langle \boldsymbol\phi f \rangle \Arrowvert < \infty \right\},
\end{equation}
where $\boldsymbol\phi(\mathbf v)\in \mathbb R^{N_m}$ is a set of basis functions.
Here we choose the basis in such a way that the first moments in $\mathbf m$ are consistent with the conservative variables $\mathbf W=(\rho,\rho \mathbf V,\rho E)^T$.
Therefore, we can write $\boldsymbol\phi$ as
\begin{equation}
    \boldsymbol\phi(\mathbf v)=(\boldsymbol \psi, \tilde{\boldsymbol\phi}(\mathbf v))^T
\end{equation}
where $\tilde{\boldsymbol\phi}(\mathbf v)$ are chosen as monomials and mixed polynomials of $\mathbf v$ up to degree $N_m$.

The entropy closure seeks for a unique map $\mathbf m \mapsto f$ by forming an optimization problem.
The objective function of the optimization problem is built with the help of a convex entropy function.
For classical Maxwell-Statistics statistics, the minimal entropy closure problem writes
\begin{equation}
    \min _{g \in B_m} \langle g \log (g)-g \rangle, \quad \text { s.t. } \mathbf{m}=\langle \boldsymbol\phi g \rangle.
    \label{eqn:entropy closure}
\end{equation}
The solution of the above optimization problem can be expressed as
\begin{equation}
    f_u(\mathbf v)=\exp(\boldsymbol\gamma_u \cdot \boldsymbol\phi),
    \label{eqn:reconstruct pdf}
\end{equation}
where $\boldsymbol\gamma_u \in \mathbb R^{N_m}$ is a vector of Lagrange multipliers of the dual formulation of the optimization problem, which reads
\begin{equation}
    \boldsymbol\gamma_u = \underset{\boldsymbol\gamma \in \mathbb{R}^{N_m}}{\operatorname{argmax}}\{\boldsymbol\gamma \cdot \mathbf{m}-\langle\exp (\boldsymbol\gamma \cdot \boldsymbol\phi)\rangle\}.
\end{equation}

The idea is to sample particle distribution functions from Eq.(\ref{eqn:reconstruct pdf}).
Note that the condition number of the minimal entropy closure at a moment $\mathbf u$ can be computed via the positive semi-definite Hessian of the dual problem, i.e.,
\begin{equation}
    \mathbf H_u = \langle \boldsymbol\phi \boldsymbol\phi \exp(\boldsymbol\gamma_u \cdot \boldsymbol\phi) \rangle.
\end{equation}
where $\boldsymbol\phi \boldsymbol\phi$ denotes the union of $\boldsymbol\phi$.
As analyzed in \cite{curto1991recursiveness,junk2000maximum}, generally the lower the condition number, the less the reconstructed particle distribution differs from the Maxwellian.
Therefore, by controlling the condition number of $\mathbf H_u$, we can generate distribution functions that are close and far from equilibrium.
We let the first three elements of $\boldsymbol\gamma$ be consistent with $\hat{\boldsymbol\alpha}$ in Eq.(\ref{eqn:maxwellian alpha}), and sample $\gamma_n$ for $n>3$ with a prescribed probability density.

It is worthwhile to mention that the condition number cannot be arbitrarily large due to the constraint of realizability.
The realizable set of the minimal entropy problem in Eq.(\ref{eqn:entropy closure}) writes
\begin{equation}
    \mathcal R=\{ \mathbf m: \mathbf m = \langle \boldsymbol\phi g \rangle, g\in B_m \}.
\end{equation}
It is known that the minimal entropy problem near the boundary of $\mathcal R$, denoted as $\partial \mathcal R$, may become ill-conditioned \cite{xiao2022predicting}.
Therefore, we keep the condition number within a certain range, and employ an upwind reconstruction to enable the capability to generate highly non-equilibrium distribution functions that may exist in highly dissipative flows.
The distribution functions produced from Eq.(\ref{eqn:reconstruct pdf}) are divided randomly into two parts, and the referenced $f$ is reconstructed as
\begin{equation}
    f_\mathrm{ref} = f_L \mathscr H(\mathbf n \cdot \mathbf v) + f_R (1- \mathscr H(\mathbf n \cdot \mathbf v)),
    \label{eqn:upwind reconstruction}
\end{equation}
where $\{f_L,f_R\}$ are the distribution functions from these two parts,
$\mathbf n$ is a randomly generated unit vector, and $\mathscr H(x)$ is the heaviside step function.
Figure \ref{fig:reconstruction} shows typical near-equilibrium and non-equilibrium particle distribution functions reconstructed from Eq.(\ref{eqn:upwind reconstruction}) using the degree $N_m=4$.
As the polynomial degree increases, the non-equilibrium effect is expected to be stronger.
The detailed sampling strategy is summarized in Algorithm \ref{alg:sample}.
Note that the sampling strategy here can work together with classical simulation-based data samplers, where the data produced by a specific simulation can be employed as reinforcement data on top of the pre-generated data.
We will explain this further in Section \ref{sec:algorithm}.

\begin{figure}[htb!]
	\centering
	\subfigure[Near-equilibrium reconstructed distribution function]{
		\includegraphics[width=0.4\textwidth]{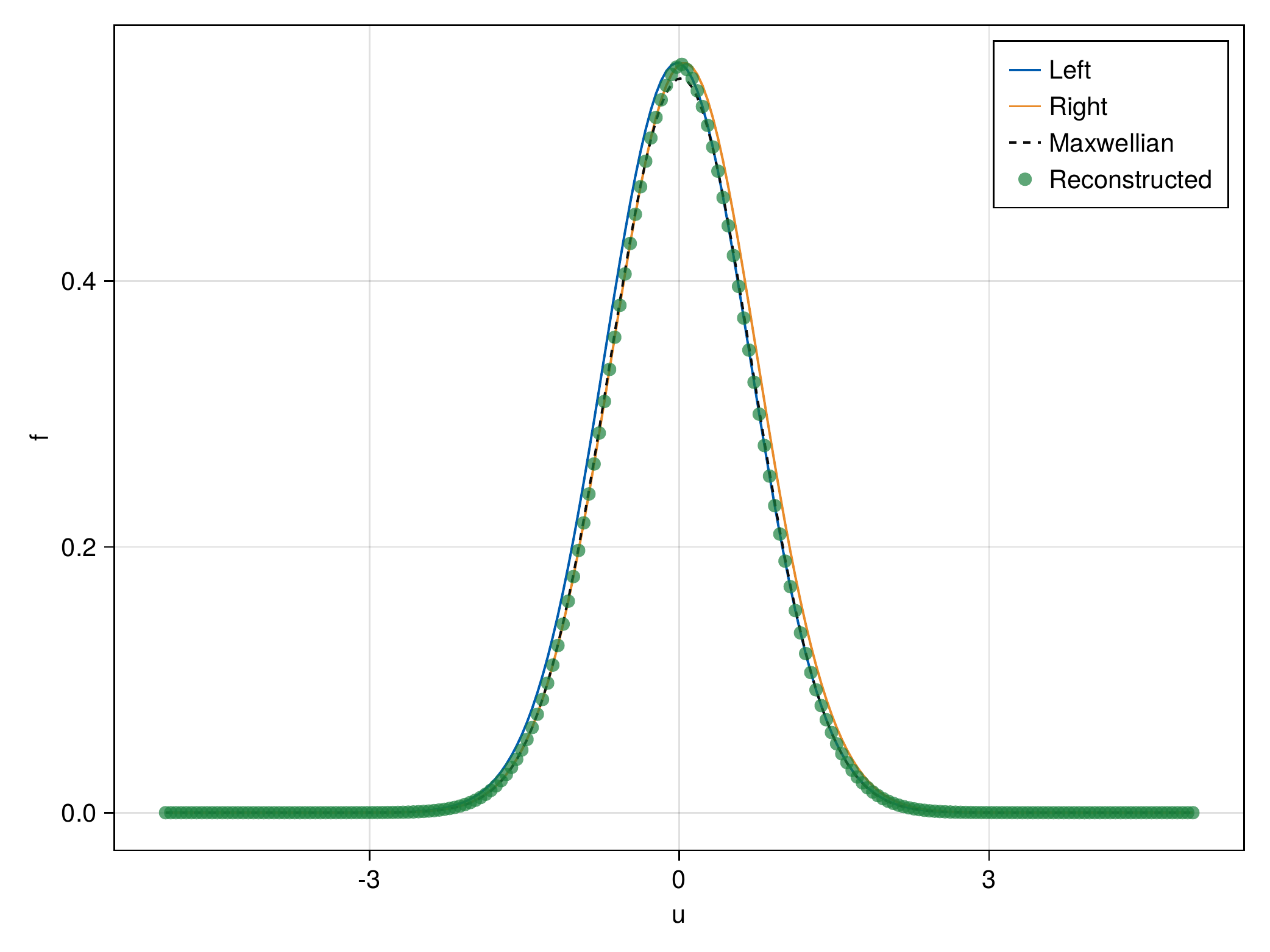}
	}
	\subfigure[Non-equilibrium reconstruction distribution function]{
		\includegraphics[width=0.4\textwidth]{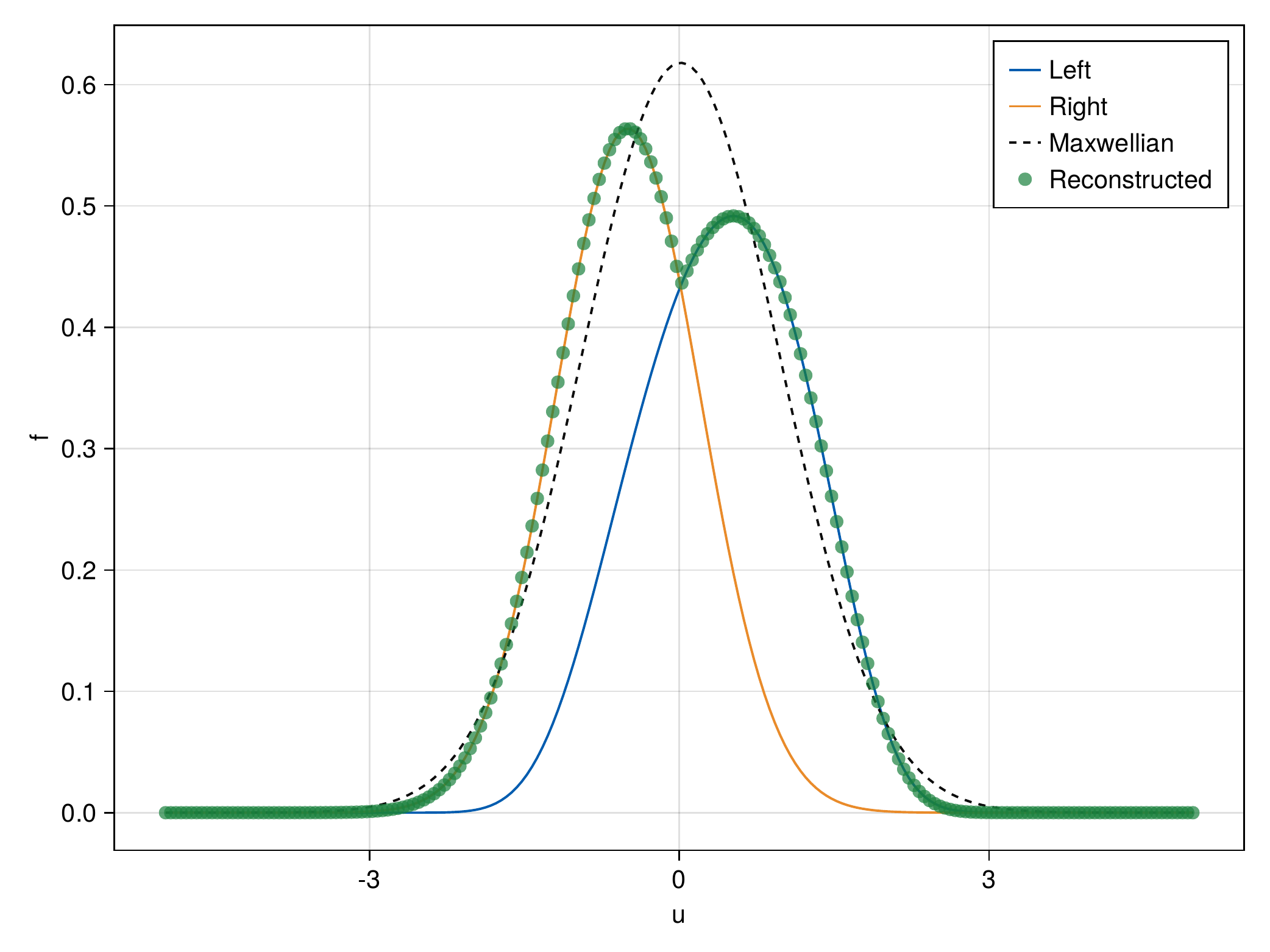}
	}
	\caption{Typical one-dimensional near-equilibrium and non-equilibrium particle distribution functions generated by the sampling strategy in Section \ref{sec:sample pdf} with the polynomial degree $N_m=4$. The values of velocity $u$ and distribution function $f$ are normalized by $(2RT_0)^{1/2}$ and $\rho_0/(2RT_0)^{1/2}$, respectively.}
    \label{fig:reconstruction}
\end{figure}

\begin{algorithm}\label{alg:sample}
\DontPrintSemicolon
\caption{Sampling of reference particle distribution functions }
\SetAlgoLined
\SetKwInOut{Input}{Input}
\Input{Number of dataset entries $N_i$, maximum moment order $N_m$, and
spatial dimension $D$ \newline
Range of density, velocity, and temperature to formulate $\hat{\boldsymbol\alpha}$ \newline
Standard deviation $\sigma$ to sample $\boldsymbol\gamma$\newline
Condition number threshold $c$
}
\KwResult{Dataset of reference distribution functions $f$}
\For {$i = 1,\dots,2N_i$}{
    Sample $\rho\in[\rho_{min},\rho_{max}]$\;
    Sample $\mathbf V\in[\mathbf V_{min},\mathbf V_{max}]$\;
    Sample $T\in[T_{min},T_{max}]$\;
    Compute mean for $(\gamma_{u,1},\boldsymbol \gamma_{u,2}, \gamma_{u,3})^T=\hat{\boldsymbol \alpha}$\;
    \Do {$\sigma_{\mathbf H_u(\boldsymbol\gamma_u)}<c$}{
        Sample $\gamma_{u,n}$, $n=4,\dots,N_m $ \;
        Compute $\mathbf H_u(\boldsymbol\gamma_u)$\;
    }
    Compute $f_{\text{ent},i} = \exp(\boldsymbol\gamma_u \cdot \boldsymbol\phi)$ \;
 }
 Divide $f_\mathrm{ent}$ into $f_L$ and $f_R$ equally \;
 \For {$i = 1,\dots,N_i$}{
    Sample unit vector $\mathbf n$ \;
    Reconstruct $f_{\mathrm{ref},i}=f_{L,i} \mathscr H(\mathbf n \cdot \mathbf v) + f_{R,i} (1 - \mathscr H(\mathbf n \cdot \mathbf v))$ \;
 }
\end{algorithm}

\subsection{Computation of collision term}\label{sec:collision data}

After the particle distribution functions are collected, the collision terms need to be computed and used in the cost function Eq.(\ref{eqn:cost function}).
In this paper, we consider three different collision terms for the Boltzmann and its model equations, which are briefly introduced below.

\begin{enumerate}
  \item[(1)] \textit{Full Boltzmann collision integral with fast spectral method}: The full Boltzmann term can be transformed into Carleman-type \cite{carleman1944integrale},
  \begin{equation}
    \begin{aligned}
    \mathcal Q(f,f)
    &=\int_{\mathbb{R}^{3}} \int_{\mathbb{S}^{2}} \Theta g\left[f\left(\mathbf{v}+\frac{g \Omega-\mathbf{g}}{2}\right) f\left(\mathbf{v}_{*}-\frac{g \Omega-\mathbf{g}}{2}\right)-f(\mathbf{v}) f\left(\mathbf{v}_{*}\right)\right] d \Omega d \mathbf{v}_{*} \\
    &=2 \int_{\mathbb{R}^{3}} \int_{\mathbb{R}^{3}} \Theta \delta\left(2 \mathbf{y} \cdot \mathbf{g}+\mathbf{y}^{2}\right)\left[f\left(\mathbf{v}+\frac{\mathbf{y}}{2}\right) f\left(\mathbf{v}_{1}-\frac{\mathbf{y}}{2}\right)-f(\mathbf{v}) f\left(\mathbf{v}_{*}\right)\right] d \mathbf{y} d \mathbf{v}_{*} \\
    &=4 \int_{\mathbb{R}^{3}} \int_{\mathbb{R}^{3}} \Theta \delta(\mathbf{y} \cdot \mathbf{z})[f(\mathbf{v}+\mathbf{y}) f(\mathbf{v}+\mathbf{z})-f(\mathbf{v}) f(\mathbf{v}+\mathbf{y}+\mathbf{z})] d \mathbf{y} d \mathbf{z},
    \end{aligned}
    \label{eqn:carleman}
    \end{equation}
    which can be solved with a discrete Fourier transform-based spectral method with a computational cost of $O(N^3 \log N)$ \cite{mouhot2006fast}.
  Specifically, the particle distribution function and collision term discretized with $N_v$ quadrature points are expanded into the Fourier series,
    \begin{equation}
    \begin{aligned}
        &f(t,\mathbf x,\mathbf{v})=\sum_{k=-N_v / 2}^{N_v / 2-1} {f}_{k}(t,\mathbf x) \exp \left(i \xi_{k} \cdot \mathbf{v}\right),\\
        &{f}_{k}=\frac{1}{L^{D}} \langle f(t,\mathbf x,\mathbf{v}) \exp \left(-i \xi_{k} \cdot \mathbf{v}\right) \rangle,\\
        &\mathcal Q(t,\mathbf x,\mathbf{v})=\sum_{k=-N_v / 2}^{N_v / 2-1} \mathcal{Q}_{k}(t,\mathbf x) \exp \left(i \xi_{k} \cdot \mathbf{v}\right),\\
        &\mathcal{Q}_{k}=\sum_{l, m=-N_v / 2,(l+m=k)}^{N_v / 2-1} {f}_{l} {f}_{m} [\beta(l, m)-\beta(m, m)],
    \end{aligned}
    \end{equation}
  where $L$ is the span of velocity space in each dimension, and $\beta$ is the kernel mode,
  and the convolutions defined in Eq.(\ref{eqn:carleman}) can be solved in the frequency domain.
  The details of this method can be found in \cite{mouhot2006fast}, and its numerical implementation is available in \cite{xiao2021kinetic}.
  \item[(2)] \textit{Shakhov relaxation model}: The Shakhov relaxation model builds a heat flux-based correction term on top of the BGK model to ensure that the Chapman-Enskog expansion of the model has the correct Prandtl number \cite{shakhov1968generalization}.
  The equilibrium state in the Shakhov model writes
  \begin{equation}
      \mathcal E = \mathcal M \left[1+(1-\mathrm{Pr}) \mathbf{c} \cdot \mathbf{q}\left(\frac{\mathbf c^2}{R T}-5\right) /(5 p R T)\right],
  \end{equation}
    where Pr is the Prandtl number, $\mathbf c$ is the peculiar velocity, $\mathbf q$ is the heat flux, and $R=k/m$ is the gas constant.
    It has been shown that the Shakhov model is more accurate in predicting highly non-equilibrium flows in normal shock wave compared to the BGK model \cite{xu2011improved}.
  With the inclusion of heat flux, the Shakhov model loses some of the structural properties shown in Section \ref{sec:theory}.
  For example, non-negative solution of particle distribution function can appear, and the H-theorem is no longer strictly valid in the Shakhov model.
  Using the Shakhov model to construct the dataset is to validate the ability of RelaxNet to correctly approximate $\mathcal E$.
  \item[(3)] \textit{$\nu$-BGK model}: Another strategy to improve the BGK model is to introduce velocity-dependent relaxation frequency. Here we choose the $\nu$-BGK model \cite{mieussens2004numerical}, where the relaxation term writes
  \begin{equation}
      \mathcal Q = \nu(\mathbf v)(\mathcal M - f).
  \end{equation}
  Different curves of $\nu$ as a function of $\mathbf v$ can be fitted to provide the correct Prandtl number of monatomic gas.
  Here we adopt the relaxation frequency defined in \cite{yuan2022capturing}, which is an improvement over the frequencies provided in \cite{mieussens2004numerical} and can provide more accurate solution in non-equilibrium flows.
  In this model, the relaxation frequency is defined as
  \begin{equation}
  \begin{aligned}
      &\nu(|\mathbf v|) = A\frac{p}{\mu}\left[ \nu_{eq}^0(\xi) + 2\nu_{eq}^0(\xi) \right],\\
      &\nu_{eq}^0(\xi) = \frac{3}{2} \left[ \exp(-\xi^2) + \frac{\sqrt \pi}{2} \left( \frac{1}{\xi} + 2\xi \right) \mathrm{erf}(\xi)\right],
    \end{aligned}
    \label{eqn:nubgk frequency}
  \end{equation}
  where $\xi=|\mathbf c|/\sqrt{2kT/m}$, and $A$ is a parameter depending on the viscosity index.
  Using the $\nu$-BGK model to construct the dataset is to validate the ability of RelaxNet to correctly approximate velocity-dependent $\nu$.
\end{enumerate}

\section{Solution Algorithm}\label{sec:algorithm}

\subsection{Update algorithm}

The finite volume method is employed to build the solution algorithm of the RelaxNet-based universal Boltzmann equation.
We denote the mean value of particle distribution function in a control volume as
\begin{equation}
    f(t^n,\mathbf x_i,\mathbf v_j)=f_{i,j}^n=\frac{1}{\Omega_{i}(\mathbf x)\Omega_{j}(\mathbf v)} \int_{\Omega_{i}} \int_{\Omega_{j}} f(t^n,\mathbf x,\mathbf v) d\mathbf xd\mathbf v,
\end{equation}
where $\Omega_{i}$ and $\Omega_{j}$ are cell areas in the discrete physical and velocity space.
The update algorithm of distribution function in each control volume can be written as
\begin{equation}
    f_{i,j}^{n+1}=f_{i,j}^n+\frac{1}{\Omega_{i}}\int_{t^n}^{t^{n+1}} \sum_{r=1}^{n_f} F_r \Delta S_r dt+ \int_{t^n}^{t^{n+1}} \mathrm{NN}_\theta(f_{i},\bar\tau_i)_j dt,
    \label{eqn:update}
\end{equation}
where $F_r$ is the time-dependent flux function of distribution function at $r$-th cell interface, $\Delta  S_r$ is the interface area, and $n_f$ is the number of interfaces.
Details of the update algorithm can be found in \cite{xiao2021stochastic}.

\subsection{Numerical flux}

The numerical flux in Eq.(\ref{eqn:update}) is evaluated from the reconstructed particle distribution function at the cell interface.
At the cell interface $\mathbf x_{i+1/2}$, the distribution function is constructed in an upwind fashion,
\begin{equation}
    f_{i+1/2,j}=f_{i+1/2,j}^L \mathscr H\left(\mathbf n_{i+1/2} \cdot \mathbf v_j\right) + f_{i+1/2,j}^R (1-\mathscr H\left(\mathbf n_{i+1/2} \cdot \mathbf v_j\right)),
    \label{eqn:interface distribution}
\end{equation}
where $\mathbf n_{i+1/2}$ is the unit normal vector of the interface, and $\mathscr H(x)$ is the heaviside step function.
The left and right value of distribution function $f_{i+1/2,j}^{L,R}$ can be obtained through reconstruction and interpolation.
For example, a second-order interpolation results in
\begin{equation}
\begin{aligned}
    &f_{i+1/2,j}^L = f_{i,j} + \nabla_\mathbf x f_{i,j} \cdot (\mathbf x_{i+1/2}-\mathbf x_i), \\
    &f_{i+1/2,j}^R = f_{i+1,j} + \nabla_\mathbf x f_{i+1,j} \cdot (\mathbf x_{i+1/2}-\mathbf x_{i+1}),
\end{aligned}
\end{equation}
where $\nabla_\mathbf x f$ is the reconstructed gradient with slope limiters.
The principle of higher-order interpolation can be found in \cite{xiao2021flux}.
The numerical flux of particle distribution function can then be evaluated as
\begin{equation}
    F_{i+1/2,j}= f_{i+1/2,j} \mathbf n_{i+1/2} \cdot \mathbf v_j.
\end{equation}

\subsection{Collision term}

The RelaxNet-based collision operator in each cell reads
\begin{equation}
    \mathcal Q_{i,j} = \mathrm{NN}_\theta(f_{i},\bar\tau_i)_j.
\end{equation}
According to the kinetic theory of gases, the mean relaxation time is defined as $\bar\tau = \mu/{p}$ where $p$ is pressure.
In this paper, we adopt the variational hard-sphere (VHS) molecule model to compute the viscosity coefficient, i.e.,
\begin{equation}
    \mu = \mu_\mathrm{ref} \left(\frac{T}{T_\mathrm{ref}}\right)^\omega,
    \label{eqn:hs viscosity}
\end{equation}
where $\mu_\mathrm{ref}$ and $T_\mathrm{ref}$ are the viscosity and temperature in the reference state, and $\omega$ is the viscosity index.

Once the parameters in RelaxNet get optimized, the collision term can be integrated with the help of time-integral algorithms.
While the forward Euler method is the most straightforward integrator for solving Eq.(\ref{eqn:update}),
it is feasible to adopt other higher-order methods, e.g., Tsitouras' 5/4 Runge-Kutta method \cite{tsitouras2011modified}.
When the time step is much larger than mean relaxation time, it is preferable to employ implicit \cite{shampine1982implementation} and implicit-explicit (IMEX) methods \cite{ascher1995implicit} to improve numerical stability of the scheme.

Note that the update procedure in Eq.(\ref{eqn:update}) produces case-specific distribution functions that may be beyond the dataset used for training and testing.
The newly generated distribution functions can be supplemented to the dataset on-the-fly by computing the collision term $\mathcal Q_{i,j}=\mathcal Q(f_{i,j})$ with the methods shown in Section \ref{sec:collision data}.
We use the rise of residuals as a criterion to generate supplementary data.
A schematic diagram of the solution algorithm of the RelaxNet-based UBE is presented in Figure \ref{fig:solver}.

\begin{figure}[htb!]
    \centering
    \includegraphics[width=0.99\textwidth]{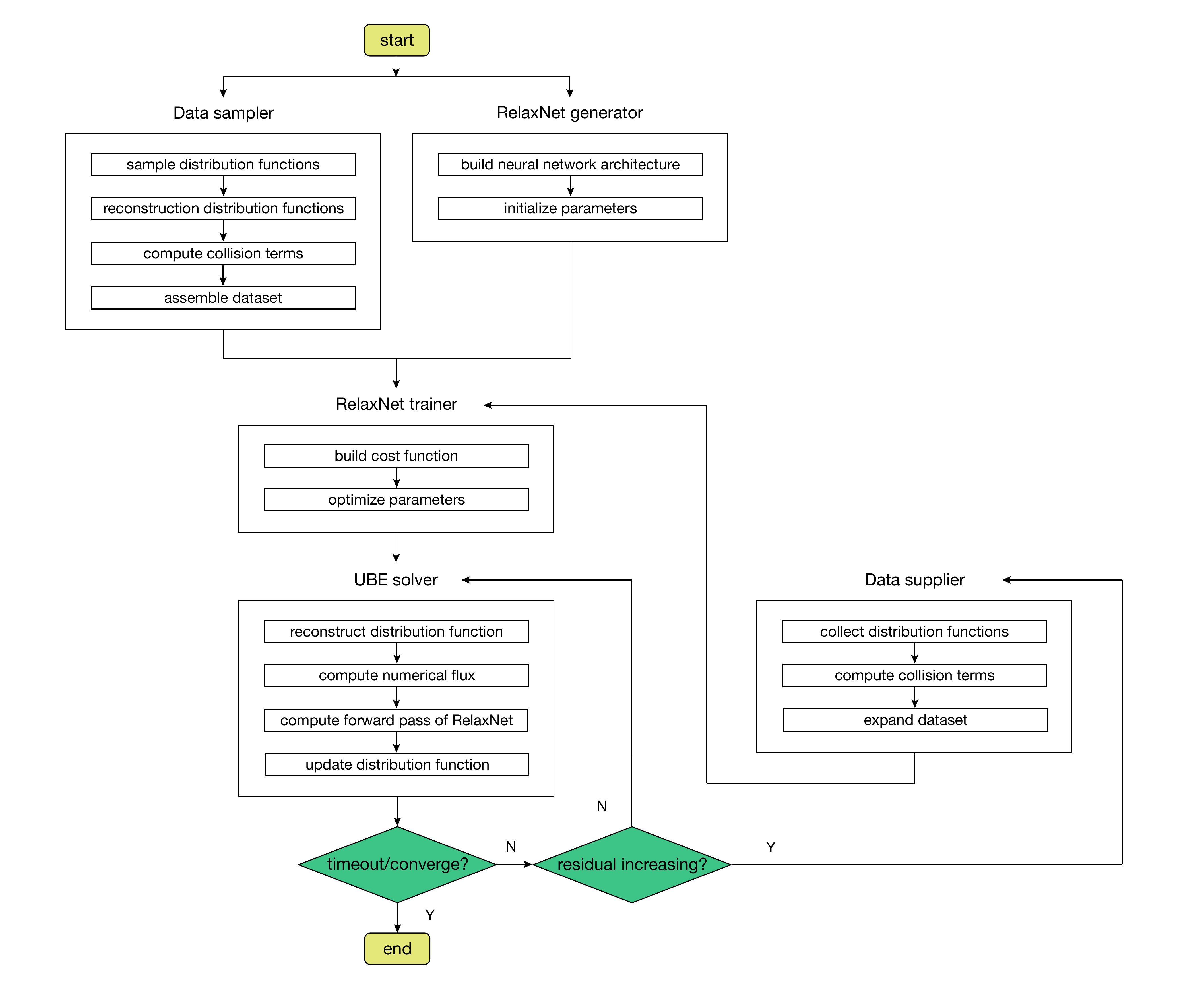}
    \caption{Schematic diagram of solution algorithm of the RelaxNet-based universal Boltzmann equation.}
    \label{fig:solver}
\end{figure}

\section{Numerical Experiments}

In this section, we present numerical experiments to validate the RelaxNet-based universal Boltzmann equation and its solution algorithm.
The dimensionless variables are introduced in the simulation, i.e.,
\begin{equation*}
\begin{aligned}
    & \tilde t = \frac{t}{L_0/(2RT_0)^{1/2}}, \ \tilde{\mathbf x} = \frac{\mathbf x}{L_0},\ \tilde \rho = \frac{\rho}{\rho_0}, \ \tilde{\mathbf v} = \frac{\mathbf v}{(2RT_0)^{1/2}}, \ \tilde{\mathbf V} = \frac{\mathbf V}{(2RT_0)^{1/2}}, \ \tilde T = \frac{T}{T_0}, \\
    & \tilde{E}=\frac{E}{2RT_0}, \ \tilde{\mathbf P} = \frac{\mathbf P}{2 \rho_0 RT_0}, \ \tilde{\mathbf q} = \frac{\mathbf q}{\rho_0 (2RT_0)^{3/2}}, \ \tilde \mu = \frac{\mu}{\rho_0 L_0 (2RT_0)^{1/2}}, \ \tilde f = \frac{f}{\rho_0 / (2RT_0)^{D/2}}.
\end{aligned}
\end{equation*}
The subscript zero represents the reference state. For brevity, the tilde notation for dimensionless variables will be removed henceforth.
In all simulations we employ monatomic hard-sphere gas, where the viscosity coefficient can be determined by Eq.(\ref{eqn:hs viscosity}).

\subsection{Data generation}

To further explain the principle of the data sampling strategy described in Section \ref{sec:data} and illustrate its superiority, we first compare the distribution of the data generated by Algorithm \ref{alg:sample} with the data collected directly from numerical simulation results of the Boltzmann equation.
The computational setup for Algorithm \ref{alg:sample} is presented in Table \ref{tab:sample}, where $\{\rho_s,U_s,T_s\}$ is the macroscopic variables used for sampling distribution functions, $\boldsymbol\gamma$ is the Lagrange multipliers in Eq.(\ref{eqn:entropy closure}), and $\mathcal U(a,b)$, $\mathcal N(\mu,\sigma)$, and $\mathcal T(\mu,\sigma,a,b)$ denote uniform, normal, and truncated normal distributions, respectively.
\begin{table}[htbp]
	\caption{Computational setup for sampling particle distribution functions in the data generation problem.}
	\centering
	\begin{tabular}{llllllll} 
		\hline
		$N_i$ & $N_m$ & $D$ & $\rho_s$ & $U_s$ & $T_s$ & $\gamma_4$ & $\gamma_5$ \\
		$3000$ & 4 & 1 & $\mathcal U(0,2)$ & $\mathcal U(-1.5,1.5)$ & $\mathcal U(0,4)$ & $\mathcal N(0,0.005)$ & $\mathcal T(0,0.005,-\infty,0)$ \\
		\hline
	\end{tabular} 
	\label{tab:sample}
\end{table}

Two problems are considered below as illustrative examples to produce simulation-based datasets.
\begin{enumerate}
    \item[(1)] Wave propagation problem. The initial particle distribution function is set to Maxwellian everywhere in correspondence with the following flow field,
    \begin{equation*}
        \left(\begin{array}{c}
        \rho \\
        U \\
        p \\
        \end{array}\right)_{{t=0}}=\left(\begin{array}{c}
        A(1+a \sin(2\pi x)) \\
        b \\
        c \\
        \end{array}\right).
    \end{equation*}
    The detailed computational setup is presented in Table \ref{tab:wave}.
    
    \item[(2)] Sod shock tube problem. The initial particle distribution function is set to Maxwellian everywhere in correspondence with the one-dimensional Riemann problem,
    \begin{equation*}
        \left(\begin{array}{c}
        \rho \\
        U \\
        p \\
        \end{array}\right)_{{t=0,x_L}}=\left(\begin{array}{c}
        \rho_L \\
        0 \\
        p_L \\
        \end{array}\right),\ 
        \left(\begin{array}{c}
        \rho \\
        U \\
        p \\
        \end{array}\right)_{{t=0,x_R}}=\left(\begin{array}{c}
        \rho_R \\
        0 \\
        p_R \\
        \end{array}\right).
    \end{equation*}
    The detailed computational setup is presented in Table \ref{tab:sod}.
\end{enumerate}

\begin{table}[htbp]
	\caption{Computational setup for wave propagation problem.} 
	\centering
	\begin{tabular}{llllllll} 
		\hline
		$t$ & $x$ & $N_x$ & $u$ & $N_u$ & Quadrature & Kn & $A_s$ \\ 
		$(0,5]$ & $[0,1]$ & $200$ & $[-8,8]$ & 80 & Rectangular & $10^{-3}$ & $\mathcal U(0.5,1.5)$ \\ 
		%\hline 
		\rule{0pt}{0ex}\\
		$a_s$ & $b_s$ & $c_s$ & Collision & Integrator & Boundary & CFL \\ 
		$\mathcal U(0,0.2)$ & $\mathcal U(-1,1)$ & $\mathcal U(0.5,1.5)$ & Shakhov & RK4 & Periodic & 0.5 \\ 
		\hline
	\end{tabular} 
	\label{tab:wave}
\end{table}

\begin{table}[htbp]
	\caption{Computational setup for Sod shock tube problem.} 
	\centering
	\begin{tabular}{llllllll} 
		\hline
		$t$ & $x$ & $N_x$ & $u$ & $N_u$ & Quadrature & Kn & $\rho_{Ls}$ \\
		$(0,0.2]$ & $[0,1]$ & $200$ & $[-8,8]$ & 80 & Rectangular & $10^{-4}$ & $\mathcal U(0,2)$ \\
		%\hline 
		\rule{0pt}{0ex}\\
		$\rho_{Rs}$ & $T_{Ls}$ & $T_{Rs}$ & Collision & Integrator & Boundary & CFL \\ 
		$\mathcal U(0,0.25)$ & $\mathcal U(1,3)$ & $\mathcal U(0,1.25)$ & Shakhov & IMEX & Dirichlet & 0.5 \\
		\hline
	\end{tabular} 
	\label{tab:sod}
\end{table}

We compare the datasets generated from the algorithmic sampler described in Section \ref{sec:data} and the above two numerical examples with the same sample size $N_i=3000$.
Two cases are considered.
In the first case, the input parameters of two examples are fixed as $\{A,a,b,c\}=\{1,0.1,1,1\}$ and $\{\rho_L,\rho_R,T_L,T_R\}=\{1,0.125,2,0.625\}$ to analyze the dataset produced by a single numerical simulation.
In the second case, we sample the parameters ten times according to the uniform distribution of the intervals given in Table \ref{tab:wave} and \ref{tab:sod} to characterize the datasets collected from multiple numerical simulations with various inputs.

Figure \ref{fig:sample single} and \ref{fig:sample multiple} present the phase diagrams between macroscopic variables generated by the algorithmic sampler using Algorithm \ref{alg:sample} and the simulation results.
As shown in Figure \ref{fig:sample single}, the samples from a single numerical simulation can suffer from a strong bias.
Such a phenomenon is caused by the physical constraints in a particular flow problem.
In the wave propagation problem, the density and temperature is linearly coupled due to the isobaric pressure.
In the Sod shock tube problem, the density, velocity and temperature are correlated around wave structures, i.e., the rarefaction wave, contact discontinuity, and shock wave.
The discontinuous flow field in the Sod problem makes the samples show a even sparser distribution in the phase diagram.
As particle distribution functions generated from multiple sampling of parameters and numerical simulations are added to the dataset, as presented in Figure \ref{fig:sample multiple}, a wider range of flow variables can be obtained.
From this trend, it may be feasible to generate reliable datasets through numerical simulations if the sampling ranges of parameters are further extended, but this obviously requires the use of appropriate physical problems and adequate sample size.
To accomplish the goal, the use of trial-and-error processes is inevitable.

In contrast, the algorithmic sampler generates a dataset of particle distribution function with a much wider range of densities, velocities and temperatures.
The obtained samples are homogeneously distributed without a bias towards a preferred direction.
Therefore, the predefined domain of $\{\rho,U,T\}$ is fully covered with a relatively small sample size.
Note that this also implies a broader and unbiased distribution of particle distribution functions, since there is a one-to-one correspondence between distribution functions and macroscopic quantities.
Compared with the simulation-based sampling, the broader dataset provided by the algorithmic sampler reduces the risk of overfitting and helps improve the generalization of the neural network model.
Table \ref{tab:cost1} provides the computation time and allocations for the three methods (including the garbage collection with Julia programming language \cite{abelson1996structure}).
Since the presented sampling strategy does not require a full simulation of the Boltzmann equation with multiple cases or initial conditions, the computational cost is significantly reduced by three to five orders of magnitude.

\begin{table}[htbp]
	\caption{Computational cost of the data generation problem.} 
	\centering
	\begin{tabular}{llllllll} 
		\hline
		 & Time & Allocations & Memory & \\
		\hline
		Wave & $197.12$ s & $1.49\times 10^9$ & $485.94$ GB & \\
		Sod & $41.46$ s & $4.18\times 10^6$ & $13.15$ GB & \\ 
		Sampler & $0.03$ s & $7.27\times 10^3$ & $13.33$ MB & \\
		\hline
	\end{tabular} 
	\label{tab:cost1}
\end{table}

\subsection{Homogeneous relaxation}

In this numerical experiment, we study the homogeneous relaxation problem of non-equilibrium particle distribution functions.
To fully validate the ability of RelaxNet to approximate the solution of the Boltzmann and its model equations, three collision models, i.e., the Shakhov model, the velocity-dependent $\nu$-BGK model, and the full Boltzmann integral model described in Section \ref{sec:collision data}, are employed to provide ground-truth datasets for training the neural network.

\subsubsection{Shakhov model}\label{sec:relaxation shakhov}

First we consider the non-equilibrium distribution function in the initial state,
\begin{equation}
    f(t=0,u) = \frac{1}{2\sqrt{\pi}} \left( \exp\left(-\left(u-1.5\right)^2\right) + \frac{7}{10}\exp\left(-(u+1.5)^2 \right) \right).
    \label{eqn:f0 1d}
\end{equation}
We construct RelaxNet in such a way that there are two hidden layers of the same dimension as the input layer in both
$\mathcal E$-net and $\tau$-net.
A dataset is generated by Algorithm \ref{alg:sample} to train RelaxNet, which is then used to solve the universal Boltzmann equation.
The setups of the neural network, training and computation are provided in Table \ref{tab:shakhov}, where $N_h=2$ is the number of hidden layers and Tsit5 denotes Tsitouras’ 5/4 Runge-Kutta method \cite{tsitouras2011modified}.

\begin{table}[htbp]
	\caption{Computational setup for the homogeneous relaxation problem based on the Shakhov model.}
	\centering
	\begin{tabular}{lllllllll} 
	    \hline
	    $D$ & $u$ & $N_u$ & $N_i$ & $\rho_s$ & $U_s$ \\
	    1 & $[-8,8]$ & 80 & 10000 & $\mathcal U(0.1,10)$ & $\mathcal U(-1,1)$ \\
		%\hline
		\rule{0pt}{0ex}\\
		$V_s$ & $W_s$ & $T_s$ & $\gamma_4$ & $\gamma_5$ & $\mathrm{Kn}_s$ \\
		$0$ & $0$ & $\mathcal U(0.1,5)$ & $\mathcal N(0,0.005)$ & $\mathcal T(0,0.005,-\infty,0)$ & $\mathcal U(0.001,1)$ \\
		%\hline
		\rule{0pt}{0ex}\\
		$N_h$ & $N_m$ & Reference & Optimizer & $t$ & Kn \\
	    2 & 4 & Boltzmann & Adam & $(0,8]$ & 1 \\
		%\hline
		\rule{0pt}{0ex}\\
		Integrator & Quadrature & CFL \\
		Tsit5 & Rectangular & 0.3 \\
		\hline
	\end{tabular} 
	\label{tab:shakhov}
\end{table}

The values of cost function in Eq.(\ref{eqn:cost function}) for both training and test datasets during the training process are plotted in Figure \ref{fig:shakhov train}.
After the parameters of RelaxNet are optimized, the initial value problem of the UBE model is solved.
Note that the bimodal distribution in Eq.(\ref{eqn:f0 1d}) can fall outside the training set, so this example actually acts as a validation set.
Figure \ref{fig:shakhov solution} presents the solutions of particle distribution functions at different time instants computed by the BGK, Shakhov, and UBE model.
It can be seen that the numerical solution provided by UBE is closer to the Shakhov reference solution than the BGK model.
This is attributed to a more accurate approximation of the collision term of the Boltzmann equation, which is shown in Figure \ref{fig:shakhov collision}.
Figure \ref{fig:shakhov macro} presents the evolution of macroscopic density and energy with time.
It is clear that the conservation of mass is exactly satisfied, which provides a corroboration of Theorem \ref{theorem:2}.
The error in energy slightly exceeds the range of the Shakhov reference solution.
Due to the interpretability of RelaxNet, we can open the black box of forward pass in the UBE model and explain the issue.
Figure \ref{fig:shakhov equilibrium} presents the equilibrium state $\mathcal E$ approximated by RelaxNet at different time instants, which intuitively explains the mechanism of UBE.
Note that the Shakhov model does not satisfy $H$-theorem and the logarithm of equilibrium distribution corrected by heat flux cannot be written as a sum of collision invariants.
As a result, the approximation of RelaxNet deviates a bit from the reference Shakhov solution.

\subsubsection{$\nu$-BGK model}\label{sec:relaxation nubgk}

Next we consider the same setup as in Section \ref{sec:relaxation shakhov}, except that the velocity-dependent $\nu$-BGK model shown in Section \ref{sec:collision data} is employed to produce reference solution.
Figure \ref{fig:nubgk solution} presents the solutions of particle distribution functions at different time instants computed by the BGK, the $\nu$-BGK, and the UBE model.
It is clear that there is a more significant difference between the BGK and $\nu$-BGK solutions, while the solution provided by the UBE model is in excellent agreement with the reference $\nu$-BGK solution.
As we explained $\mathcal E$ in Section \ref{sec:relaxation shakhov}, RelaxNet allows us to quantitatively interpret the relaxation frequency $\nu$ in the UBE model.
Figure \ref{fig:nubgk collision} and \ref{fig:nubgk frequency} show the collision terms and relaxation frequencies used in RelaxNet, and compare them with the BGK and $\nu$-BGK models.
It is clear that the dependence between relaxation frequencies and particle velocities are precisely recovered, bringing a more accurate approximation of the collision term.
Figure \ref{fig:nubgk macro} presents the evolution of macroscopic density and energy with time.
It can be seen that although the parameters in RelaxNet are optimized based on the solution of $\nu$-BGK model in a supervised learning manner, the physics-informed regularization in the cost function in Eq.(\ref{eqn:cost function}) allows UBE to better preserve the conservation property, which provides a corroboration of Theorem \ref{theorem:2}.

\subsubsection{Full Boltzmann model}

We then turn to the full Boltzmann equation.
In this case, the initial particle distribution function is set as
\begin{equation*}
    f(t=0,u,v,w) = \frac{1}{\sqrt{\pi}} \left( \exp\left(-\left(u-1\right)^2\right) + \frac{7}{10}\exp\left(-(u+1)^2 \right) \right) \exp\left(-v^2\right) \exp\left(-w^2\right)  .
\end{equation*}
RelaxNet is constructed in the same way as in Section \ref{sec:relaxation shakhov} and \ref{sec:relaxation nubgk}.
Following Algorithm \ref{alg:sample}, the fast spectral method introduced in Section \ref{sec:collision data} is employed to generate the dataset used for training RelaxNet.
The computational setup is presented in Table \ref{tab:boltzmann}.

\begin{table}[htbp]
	\caption{Computational setup for the homogeneous relaxation problem based on the full Boltzmann model.}
	\centering
	\begin{tabular}{lllllllll} 
	    \hline
	    $D$ & $\mathbf v$ & $N_u$ & $N_v$ & $N_w$ & $N_i$ & \\
	    3 & $[-8,8]^3$ & 80 & 28 & 28 & 10000 & \\
		%\hline
		\rule{0pt}{0ex}\\
		$\rho_s$ & $U_s$ & $V_s$ & $W_s$ & $T_s$ & $\gamma_4$ \\
		$\mathcal U(0.1,10)$ & $\mathcal U(-1,1)$ & $0$ & $0$ & $\mathcal U(0.1,5)$ & $\mathcal N(0,0.005)$ \\
		%\hline
		\rule{0pt}{0ex}\\
		$\gamma_5$ & $\mathrm{Kn}_s$ & $N_h$ & $N_m$ & Reference & Optimizer \\
	    $\mathcal T(0,0.005,-\infty,0)$ & $\mathcal U(0.001,1)$ & 2 & 4 & Boltzmann & Adam \\
		%\hline
		\rule{0pt}{0ex}\\
		$t$ & Kn & Integrator & Quadrature & CFL \\
		$(0,8]$ & 1 & Tsit5 & Rectangular & 0.3 \\
		\hline
	\end{tabular} 
	\label{tab:boltzmann}
\end{table}

Figure \ref{fig:boltzmann solution} provides the evolution of particle distribution functions with time.
The variable used for display is the reduced distribution function on $u$ axis, i.e.,
\begin{equation*}
    h(t,u) = \int_{\mathbb R^2} f(t,u,v,w) dvdw.
\end{equation*}
As shown, UBE based on RelaxNet provides a solution consistent with the full Boltzmann equation, while the solution of the BGK model shows significant deviations.
Figure \ref{fig:boltzmann collision} presents the collision terms computed by the three approaches at different time instants.
Figure \ref{fig:boltzmann enu} shows the equilibrium distributions and relaxation frequencies approximated by RelaxNet.
It can be seen that the difference between the equilibrium distribution in RelaxNet and the Maxwellian distribution is slight, and thus the correction of the BGK solution is mainly achieved by the velocity-dependent relaxation frequency.
Different from predefined frequencies in the $\nu$-BGK model, e.g., in Eq.(\ref{eqn:nubgk frequency}), RelaxNet provides a variable relaxation frequency that depends on the solution of local particle distribution function, and thus provides an approximation that better fits the real collision frequency in Eq.(\ref{eqn:collision frequency}).
Figure \ref{fig:boltzmann entropy} shows the evolution of entropy density defined as $\eta = \langle f\log(f) \rangle$, 
from which the principle of increase of entropy and the H-theorem demonstrated in Section \ref{sec:theory} adn Theorem \ref{theorem:4} are validated.

Table \ref{tab:cost relaxation} presents the computational costs of the three methods to simulate the time-series solution of the homogeneous relaxation problem.
Since a large number of convolution operations for solving the Boltzmann collision integral are replaced by tensor multiplication, the computational cost of RelaxNet is significantly reduced compared.
Under the current numerical setup, the simulation time is reduced to $1/50$ of the fast spectral method for the full Boltzmann equation, and the memory usage is reduced to $1/17$.
This numerical experiment verifies the ability of the RelaxNet-based UBE to solve the Boltzmann equation accurately and efficiently.

\begin{table}[htbp]
	\caption{Computational cost of computing time-series solutions in the homogeneous relaxation problem.} 
	\centering
	\begin{tabular}{llllllll} 
		\hline
		 & Time & Allocations & Memory & \\
		\hline
		Boltzmann & $2.167$ s & $90322$ & $6.43$ GB & \\
		BGK & $1.076$ ms & $2939$ & $3.38$ MB & \\ 
		UBE & $43.472$ ms & $14344$ & $375.66$ MB & \\
		\hline
	\end{tabular} 
	\label{tab:cost relaxation}
\end{table}

\subsection{Normal shock structure}

We then study flow transport in inhomogeneous gases.
Due to the existence of non-Maxwellian particle distribution functions \cite{mott1951solution}, the normal shock wave problem is an ideal case to validate the RelaxNet-based UBE in solving highly non-equilibrium flows.
The initial flow field is initialized according to the Rankine-Hugoniot jump conditions, i.e.,
\begin{equation*}
    \left(\begin{array}{c}
    \rho \\
    U \\
    V \\
    W \\
    T \\
    \end{array}\right)_{{t=0,x_L}}=\left(\begin{array}{c}
    \rho_- \\
    U_- \\
    0 \\
    0 \\
    T_- \\
    \end{array}\right),\ 
    \left(\begin{array}{c}
    \rho \\
    U \\
    V \\
    W \\
    T \\
    \end{array}\right)_{{t=0,x_R}}=\left(\begin{array}{c}
    \rho_+ \\
    U_+ \\
    0 \\
    0 \\
    T_+ \\
    \end{array}\right),
\end{equation*}
where the upstream and downstream density, velocity, and temperature are denoted by $\{\rho_-, U_-, T_-\}$ and $\{\rho_+, U_+, T_+\}$, respectively. 
The initial particle distribution functions are set to Maxwellian according to the local flow variables,
\begin{equation*}
\begin{aligned}
    &\frac{\rho_+}{\rho_-}=\frac{(\varpi+1)\rm{Ma}^2}{(\varpi-1)\rm{Ma}^2+2},\\
    &\frac{U_+}{U_-}=\frac{(\varpi-1)\rm{Ma}^2+2}{(\varpi+1)\rm{Ma}^2},\\
    &\frac{T_+}{T_-}=\frac{ ((\varpi-1)\rm{Ma}^2+2) (2\varpi\rm{Ma}^2-\varpi+1) }{(\varpi+1)^2 \rm{Ma}^2},
\end{aligned}
\end{equation*}
where $\rm Ma$ is the Mach number, and $\varpi$ is the ratio of specific heat.
The upstream flow variables are used as references to dimensionaless the system.
The detailed computational setup is provided in Table \ref{tab:shock}.
Note that the dataset used to optimize the parameters of RelaxNet consists of two parts.
In addition to the data generated by Algorithm \ref{sec:algorithm} with predefined parameters, as the numerical simulation progresses, the case-specific data can be supplemented on the fly, as shown in Figure \ref{fig:solver}.

\begin{table}[htbp]
	\caption{Computational setup for the normal shock structure problem.}
	\centering
	\begin{tabular}{lllllllll} 
	    \hline
	    $D$ & $\mathbf v$ & $N_u$ & $N_v$ & $N_w$ & $N_i$ & \\
	    3 & $[-10,10]^3$ & 80 & 28 & 28 & 10000 & \\
		%\hline
		\rule{0pt}{0ex}\\
		$\rho_s$ & $U_s$ & $V_s$ & $W_s$ & $T_s$ & $\gamma_4$ \\
		$\mathcal U(0.5,5)$ & $\mathcal U(-3,3)$ & $0$ & $0$ & $\mathcal U(0.5,5)$ & $\mathcal N(0,0.005)$ \\
		%\hline
		\rule{0pt}{0ex}\\
		$\gamma_5$ & $\mathrm{Kn}_s$ & $N_h$ & $N_m$ & Reference & Optimizer \\
	    $\mathcal T(0,0.005,-\infty,0)$ & $\mathcal U(0.01,5)$ & 2 & 4 & Boltzmann & Adam \\
		%\hline
		\rule{0pt}{0ex}\\
	    $t$ & $x$ & $N_x$ & Kn & Ma & $\varpi$ \\
		$(0,250]$ & $[-35,35]$ & 100 & 1 & $[2,3]$ & $5/3$ \\
		\rule{0pt}{0ex}\\
		Integrator & Boundary & Quadrature & CFL \\
		Tsit5 & Dirichlet & Rectangular & 0.5 \\
		\hline
	\end{tabular} 
	\label{tab:shock}
\end{table}

Figure \ref{fig:shock macro ma2} shows the profiles of density, velocity, temperature, viscous stress, heat flux, and entropy density in the shock structure at $\rm Ma=2$.
The definition of viscous stress is the difference between the first element in the pressure tensor and the isotropic pressure, i.e.,
\begin{equation*}
    \Delta = \mathbf P_{xx} - p.
\end{equation*}
When the Mach number is small, the particle distribution functions inside the shock wave deviate moderately away from Maxwellian, and thus the Boltzmann, BGK, and UBE solutions are roughly equivalent.
The discrepancy of the viscous stress and heat flux provided by the Boltzmann and BGK models can be observed due to the fact that high-order moments are more sensitive to the nuances of distribution function than density and bulk velocity.
Figure \ref{fig:shock macro ma3} presents the macroscopic flow field at $\rm Ma=3$.
It is obvious that the difference between the Boltzmann and BGK solutions become more significant with the increasing Mach number.
The unit Prandtl number in the BGK model leads to an incorrect heat transfer rate and an early rise in the temperature profile.
Thanks to the solution-dependent equilibrium and relaxation frequency, the current UBE is able to provide solutions equivalent to the reference Boltzmann results in all cases.

Figure \ref{fig:shock pdf ma2} and \ref{fig:shock pdf ma3} provide the reduced distribution functions and collision terms in the $x-u$ phase space at $\rm Ma=2$ and $\rm Ma=3$.
An one-to-one correspondence can be clearly observed between the evolution of macroscopic quantities and particle distribution functions.
When $\rm Ma=3$, significant difference can be obtained between the BGK and UBE collision terms, resulting in different distributions of particle distribution functions and macroscopic flow variables.
Table \ref{tab:cost shock} shows the computational cost of computing the right-hand side operator once in the Boltzmann, BGK, and UBE models.
The benchmark verifies that the RelaxNet-based UBE increases the computational efficiency by a factor of 108 and reduces the memory load to $1/43$ compared to the fast spectral method of the Boltzmann equation.
\begin{table}[htbp]
	\caption{Computational cost of computing the right-hand side of the kinetic model in the normal shock structure problem.} 
	\centering
	\begin{tabular}{llllllll} 
		\hline
		 & Time & Allocations & Memory & \\
		\hline
		Boltzmann & $32.114$ ms & $1446$ & $102.98$ MB & \\
		BGK & $5.893$ $\rm \mu s$ & $36$ & $12.92$ KB & \\ 
		UBE & $295.947$ $\rm \mu s$ & $93$ & $2.38$ MB & \\
		\hline
	\end{tabular} 
	\label{tab:cost shock}
\end{table}

\subsection{Lid-driven cavity}

In the last numerical experiment, we study the lid-driven cavity to validate the ability of current approach to solve non-equilibrium flows under multi-dimensional geometry. 
The initial flow field inside the square cavity is set as
\begin{equation*}
    \left(\begin{array}{c}
    \rho \\
    U \\
    V \\
    W \\
    T \\
    \end{array}\right)_{{t=0}}=\left(\begin{array}{c}
    1 \\
    0 \\
    0 \\
    0 \\
    1 \\
    \end{array}\right),
\end{equation*}
and the particle distribution function is set to Maxwellian everywhere.
The flow domain is enclosed by four isothermal solid walls.
The upper wall moves in the tangent direction with $\mathbf V_w^u = (0.15, 0, 0)^T$, with the rest three walls kept still with $\mathbf V_w^{d,l,r} = (0, 0, 0)^T$.
Maxwell’s diffusive boundary is considered at all wall boundaries.
The detailed computational setup is provided in Table \ref{tab:cavity}.

\begin{table}[htbp]
	\caption{Computational setup for the lid-driven cavity problem.}
	\centering
	\begin{tabular}{lllllllll} 
	    \hline
	    $D$ & $\mathbf v$ & $N_u$ & $N_v$ & $N_w$ & $N_i$ & \\
	    3 & $[-5,5]^3$ & 28 & 28 & 28 & 10000 & \\
		%\hline
		\rule{0pt}{0ex}\\
		$\rho_s$ & $U_s$ & $V_s$ & $W_s$ & $T_s$ & $\gamma_4$ \\
		$\mathcal U(0.5,5)$ & $\mathcal U(-1,1)$ & $\mathcal U(-1,1)$ & $0$ & $\mathcal U(0.5,5)$ & $\mathcal N(0,0.005)$ \\
		%\hline
		\rule{0pt}{0ex}\\
		$\gamma_5$ & $\mathrm{Kn}_s$ & $N_h$ & $N_m$ & Reference & Optimizer \\
	    $\mathcal T(0,0.005,-\infty,0)$ & $\mathcal U(0.01,1)$ & 2 & 4 & Boltzmann & Adam \\
		%\hline
		\rule{0pt}{0ex}\\
	    $\mathbf x$ & $N_x$ & $N_y$ & $\mathbf V_w^{u}$ & $\mathbf V_w^{d,l,r}$ & $T_w$ \\
		$[0,1]^2$ & 45 & 45 & $(0.15,0,0)^T$ & $(0,0,0)^T$ & 1 \\
		\rule{0pt}{0ex}\\
		$\varpi$ & Kn & Integrator & Boundary & Quadrature & CFL \\
		$5/3$ & 0.075 & Tsit5 & Maxwell & Rectangular & 0.5 \\
		\hline
	\end{tabular} 
	\label{tab:cavity}
\end{table}

Figure \ref{fig:cavity contour} shows the contours of density with streamlines and temperature with heat flux vectors inside the cavity.
Consistent with the report in \cite{john2010investigation}, the  inverse-Fourier heat flux driven by the viscous stress in the non-equilibrium regime is accurately identified.
Figure \ref{fig:cavity line} presents the velocity profiles along the horizontal and vertical central lines of the cavity.
The validity of UBE and the corresponding numerical method is fully validated by a quantitative comparison with the reference DSMC solution.
Figure \ref{fig:cavity pdf} and \ref{fig:cavity collision} further provide the particle distribution functions and collision terms along the two central lines.
For low-speed microflows, the non-equilibrium effect is weak and the difference between the Boltzmann and BGK solutions is mild.
The correction provided by the RelaxNet-based UBE for the BGK model can be observed, which ensures the correct heat transfer rate and Prandtl number.

\section*{Conclusion}

Scientific machine learning is increasingly showing its power and versatility in modeling and simulation in computational physics.
In this paper, a novel relaxation neural network (RelaxNet) is proposed as a surrogate model of the collision operator in the Boltzmann equation.
Based on the architecture of residual neural networks, RelaxNet builds a parameterized relaxation model with solution-dependent equilibrium state and frequency.
The quantitative interpretability of RelaxNet makes it the first machine learning model that strictly preserves all the key structural properties of the Boltzmann equation, including invariance, conservation, H-theorem, and the correct continuum limit.
Since the convolution in the original Boltzmann equation is simplified by the tensor multiplication in RelaxNet, the dimensionality and computational complexity of the collision term are greatly reduced.
Based on RelaxNet, a universal Boltzmann equation (UBE) model is developed, which fuses mechanical and neural models into a single differentiable framework that is compatible with source-to-source automatic differentiation.
Based on the entropy closure of the Boltzmann moment system, an algorithmic sampler is designed to generate datasets for model training and testing, and the superiority of this strategy over simulation-based samplers is demonstrated through numerical experiments.
The solution algorithm for solving the RelaxNet-based UBE model is described in detail.
Both spatially homogeneous and inhomogeneous test cases are presented to validate UBE and its solution algorithm. 
The current approach provides an accurate and efficient tool for the study of
non-equilibrium flow physics.
Like the network chain that contains multiple ResNet blocks, it is promising to extend RelaxNet to build multi-level relaxation models \cite{d2002multiple}, which remains to be investigated in future work.

\section*{CRediT authorship contribution statement}

\textbf{Tianbai Xiao:} : Conceptualization, Formal analysis, Investigation, Methodology, Project administration, Software, Visualization, Writing – original draft, Writing – review \& editing. \textbf{Martin Frank}: Formal analysis, Project administration, Resources, Supervision, Writing – review \& editing.

\section*{Declaration of competing interest}

The authors declare that they have no known competing financial interests or personal relationships that could have appeared to influence the work reported in this paper.

\section*{Acknowledgement}

The current research is funded by the German Research Foundation in the frame of the priority programm SPP 2298 ”Theoretical Foundations of Deep Learning” with the project number 441826958.

\clearpage
\newpage

\bibliographystyle{unsrt}
\bibliography{main}

% sample
\begin{figure}[htb!]
	\centering
	\subfigure[$\rho-U$]{
		\includegraphics[width=0.4\textwidth]{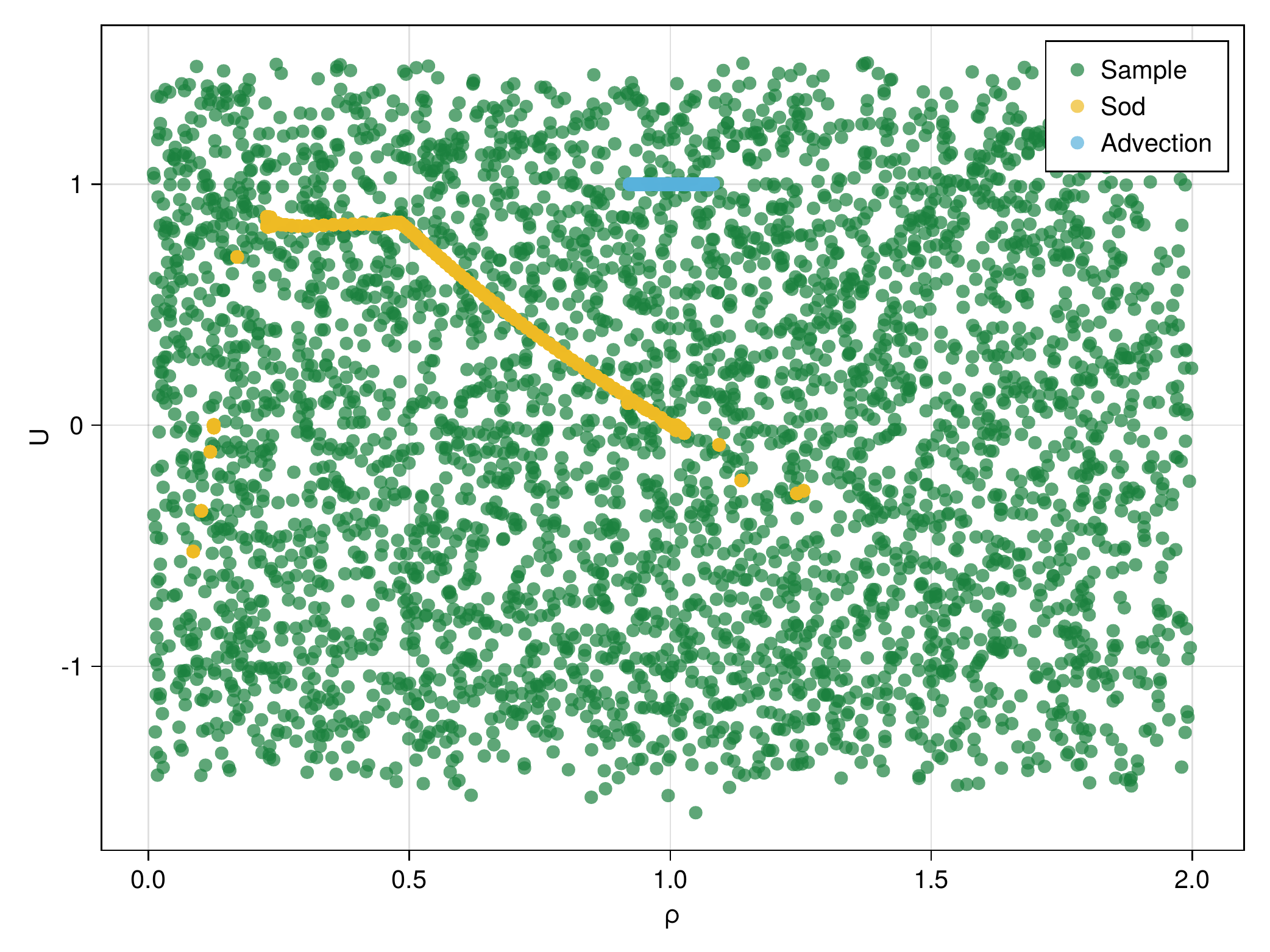}
	}
	\subfigure[$\rho-T$]{
		\includegraphics[width=0.4\textwidth]{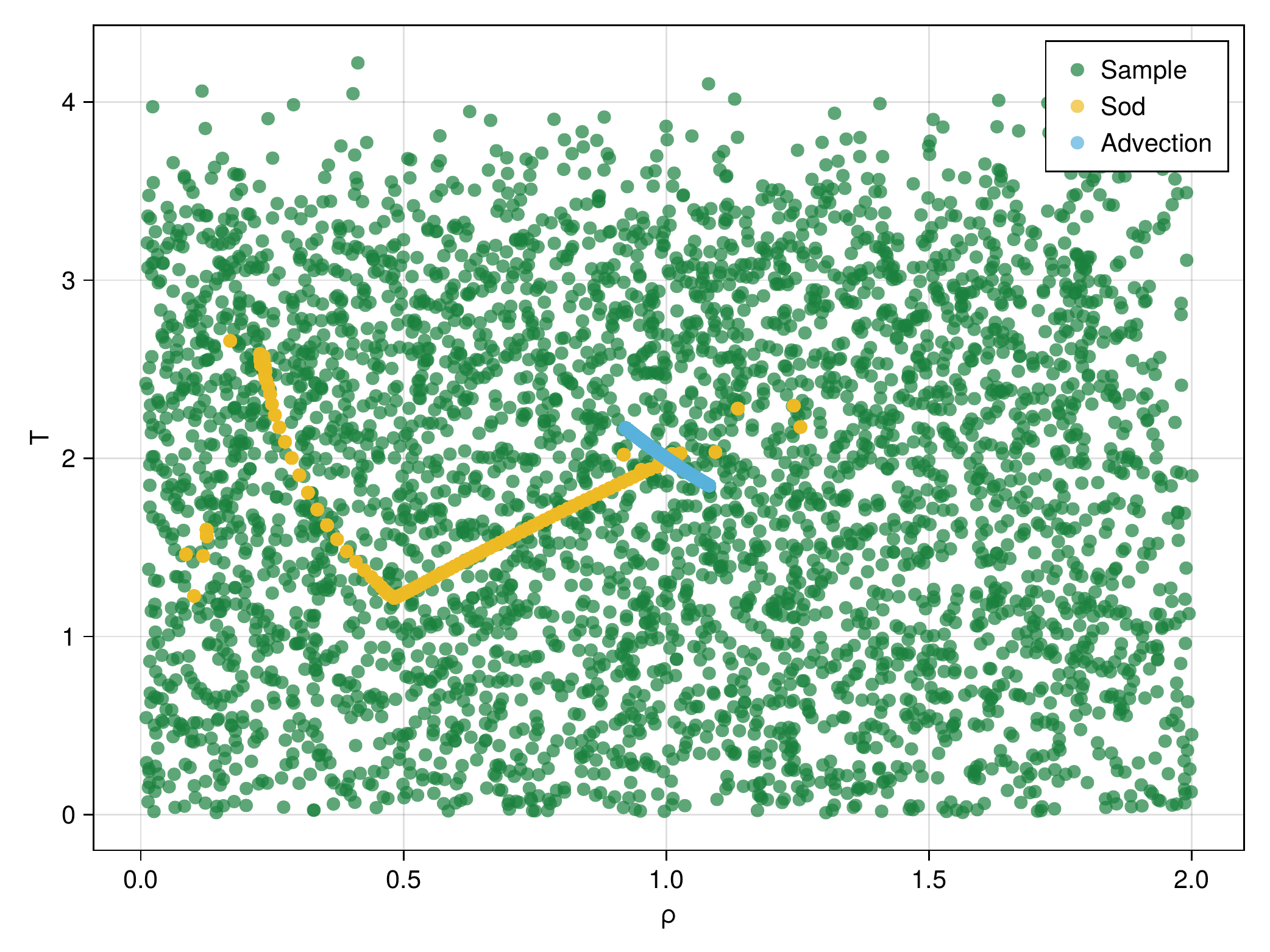}
	}
	\caption{phase diagrams of the data generated by the algorithmic sampler and a single simulation.}
    \label{fig:sample single}
\end{figure}

\begin{figure}[htb!]
	\centering
	\subfigure[$\rho-U$]{
		\includegraphics[width=0.4\textwidth]{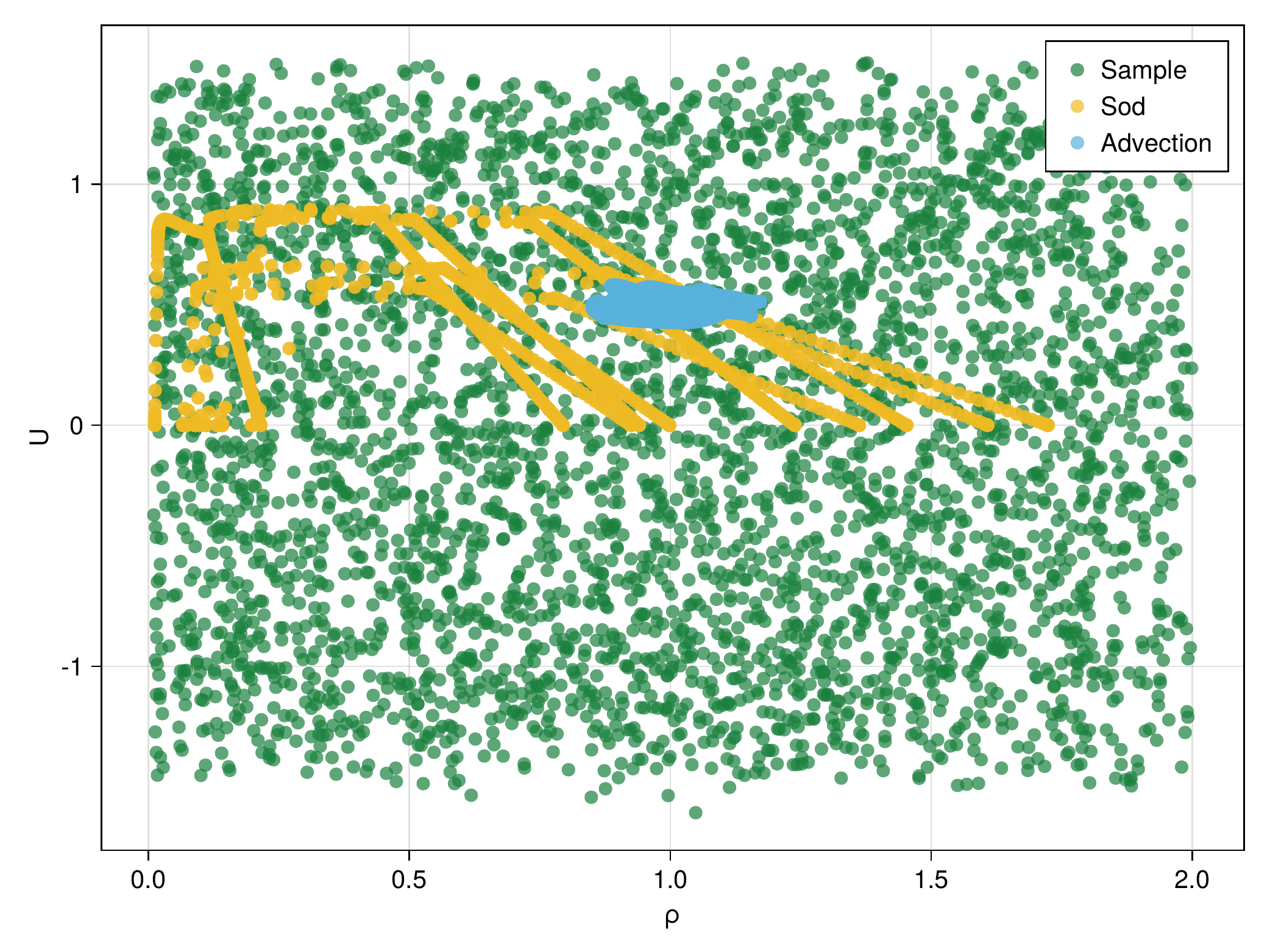}
	}
	\subfigure[$\rho-T$]{
		\includegraphics[width=0.4\textwidth]{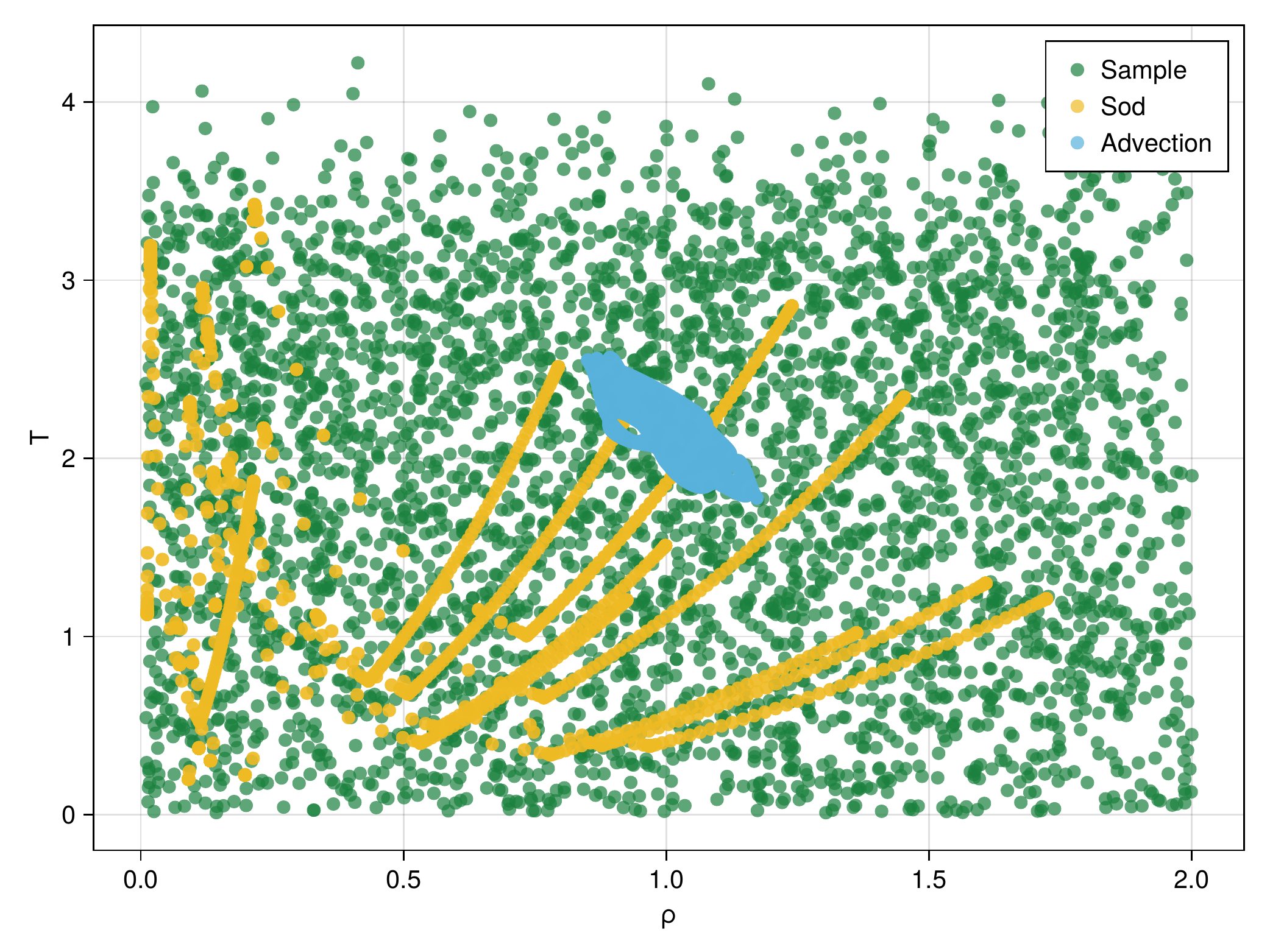}
	}
	\caption{phase diagrams of the data generated by the algorithmic sampler and multiple simulations.}
    \label{fig:sample multiple}
\end{figure}

\begin{figure}[htb!]
	\centering
	\includegraphics[width=0.4\textwidth]{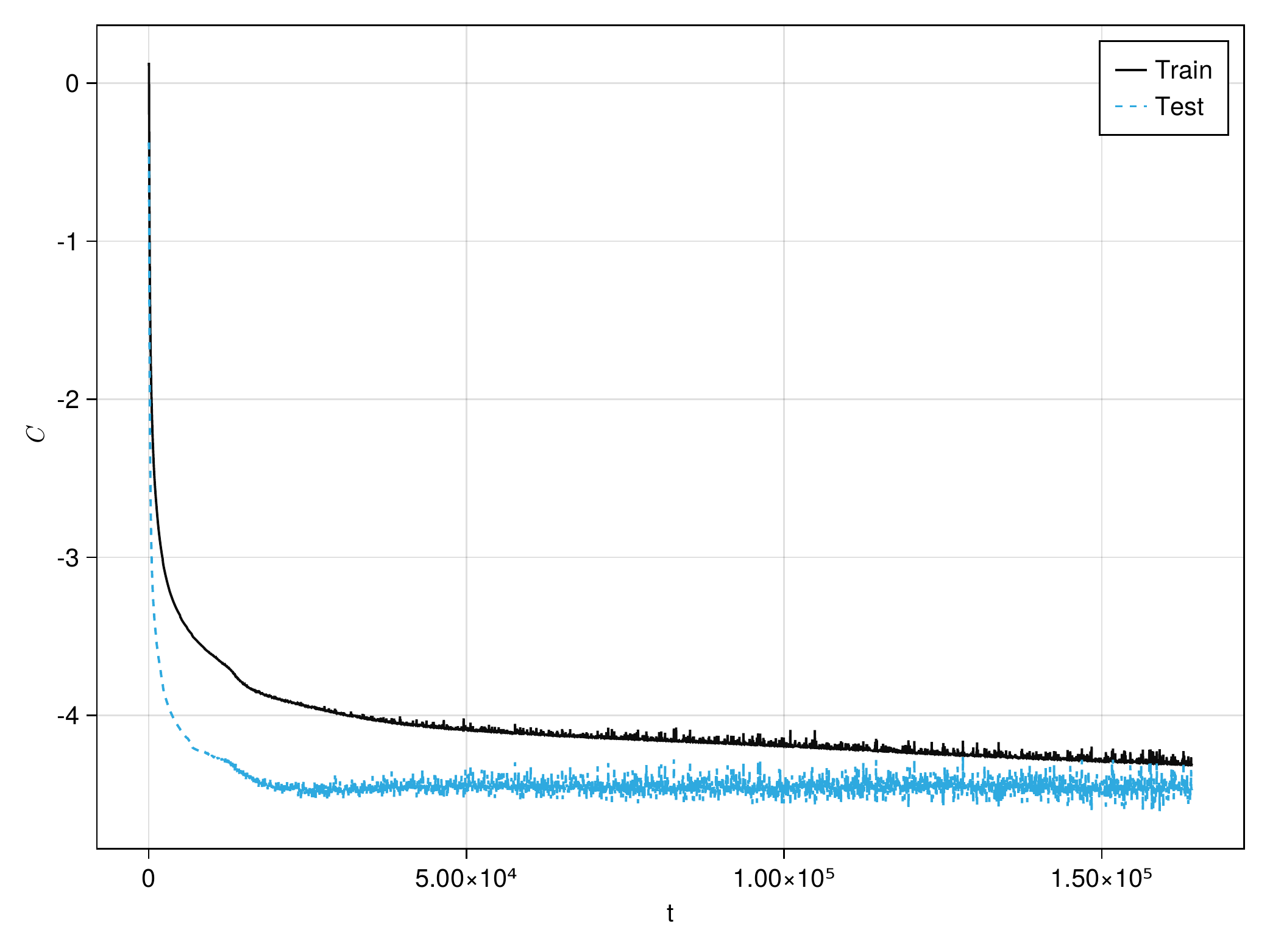}
	\caption{Values of cost function versus optimization time in the training and test datasets for the homogeneous relaxation problem.}
    \label{fig:shakhov train}
\end{figure}

\begin{figure}[htb!]
	\centering
	\subfigure[$t=0.5$]{
		\includegraphics[width=0.31\textwidth]{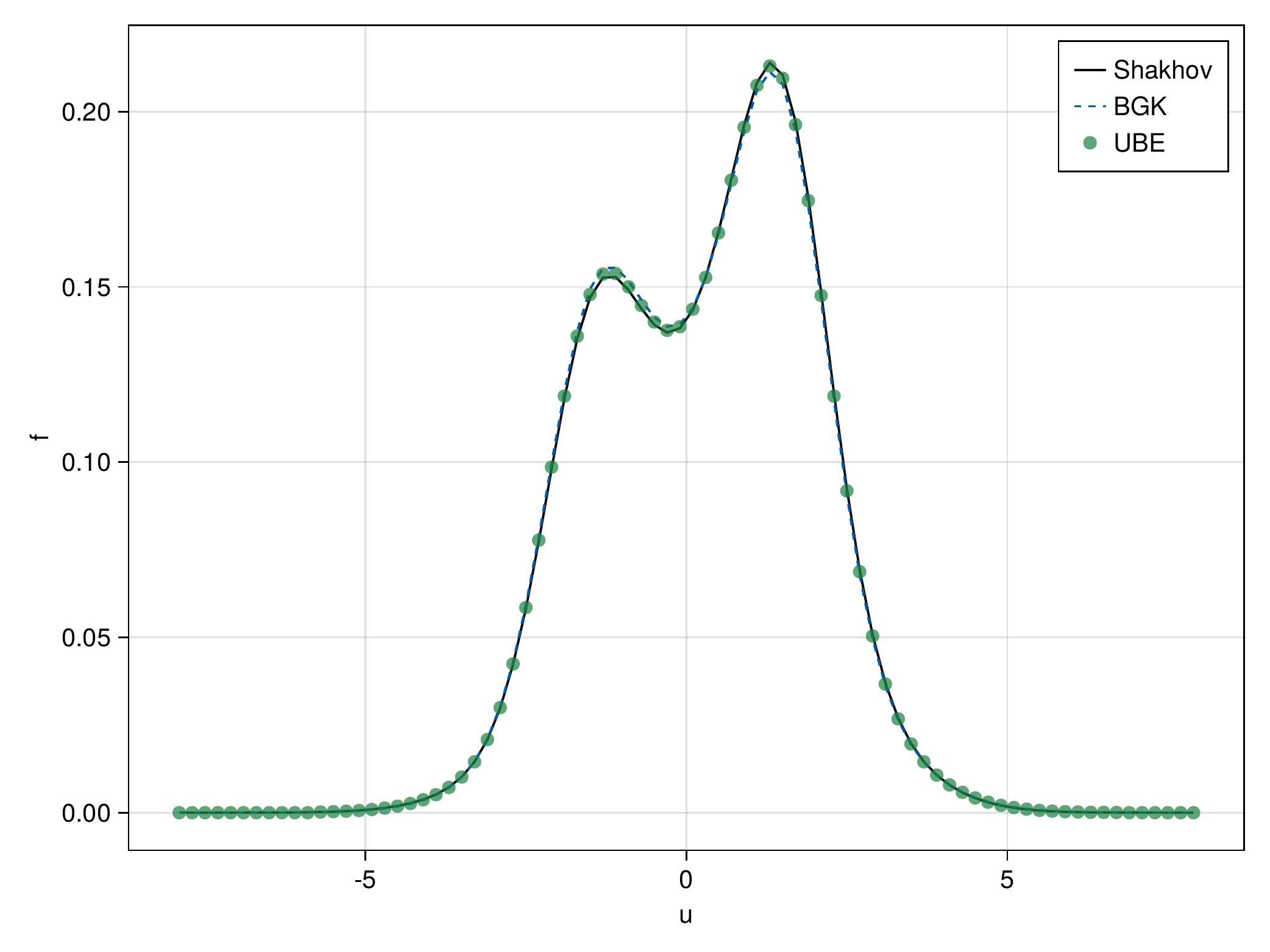}
	}
	\subfigure[$t=1$]{
		\includegraphics[width=0.31\textwidth]{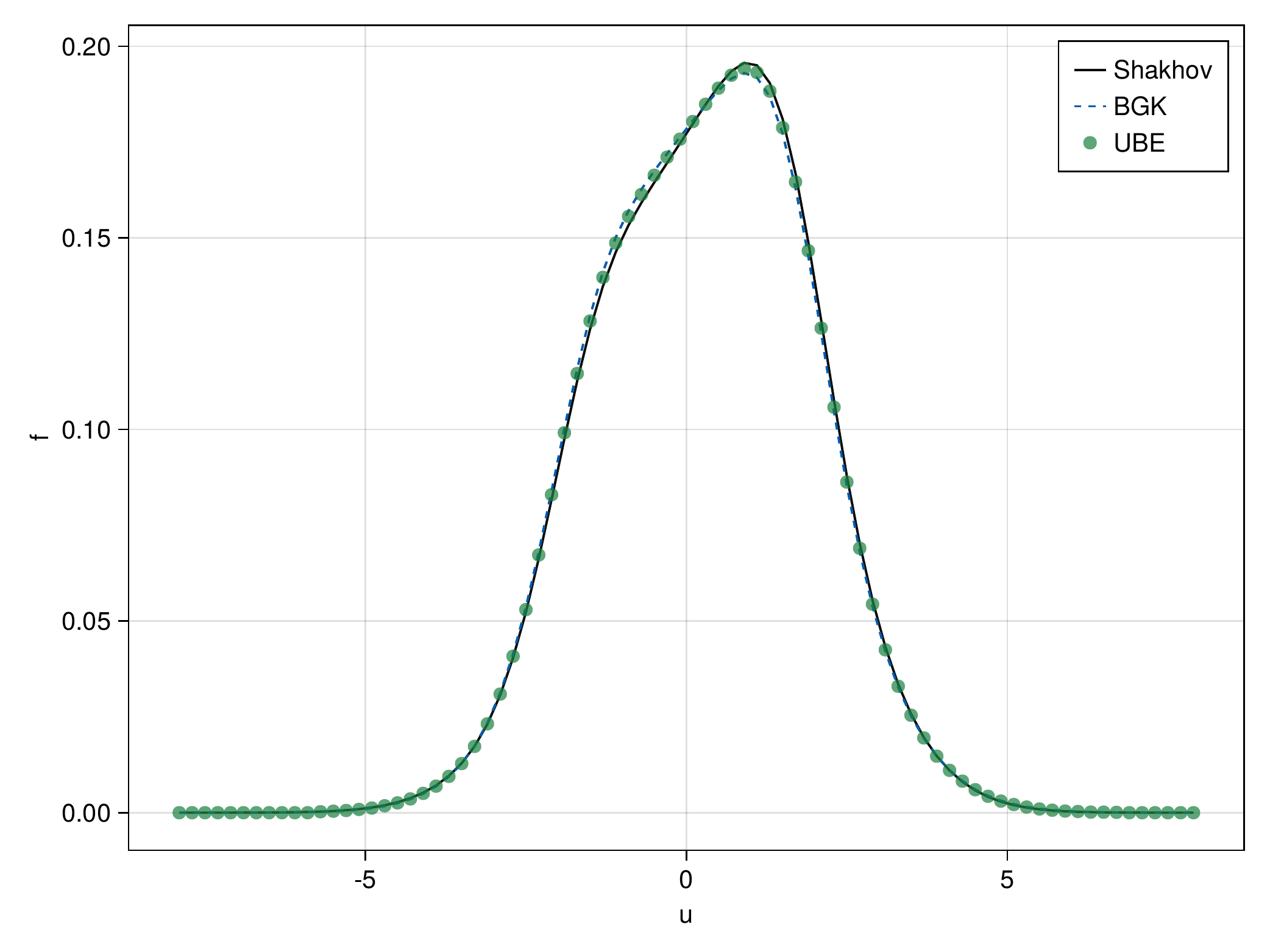}
	}
	\subfigure[$t=2$]{
		\includegraphics[width=0.31\textwidth]{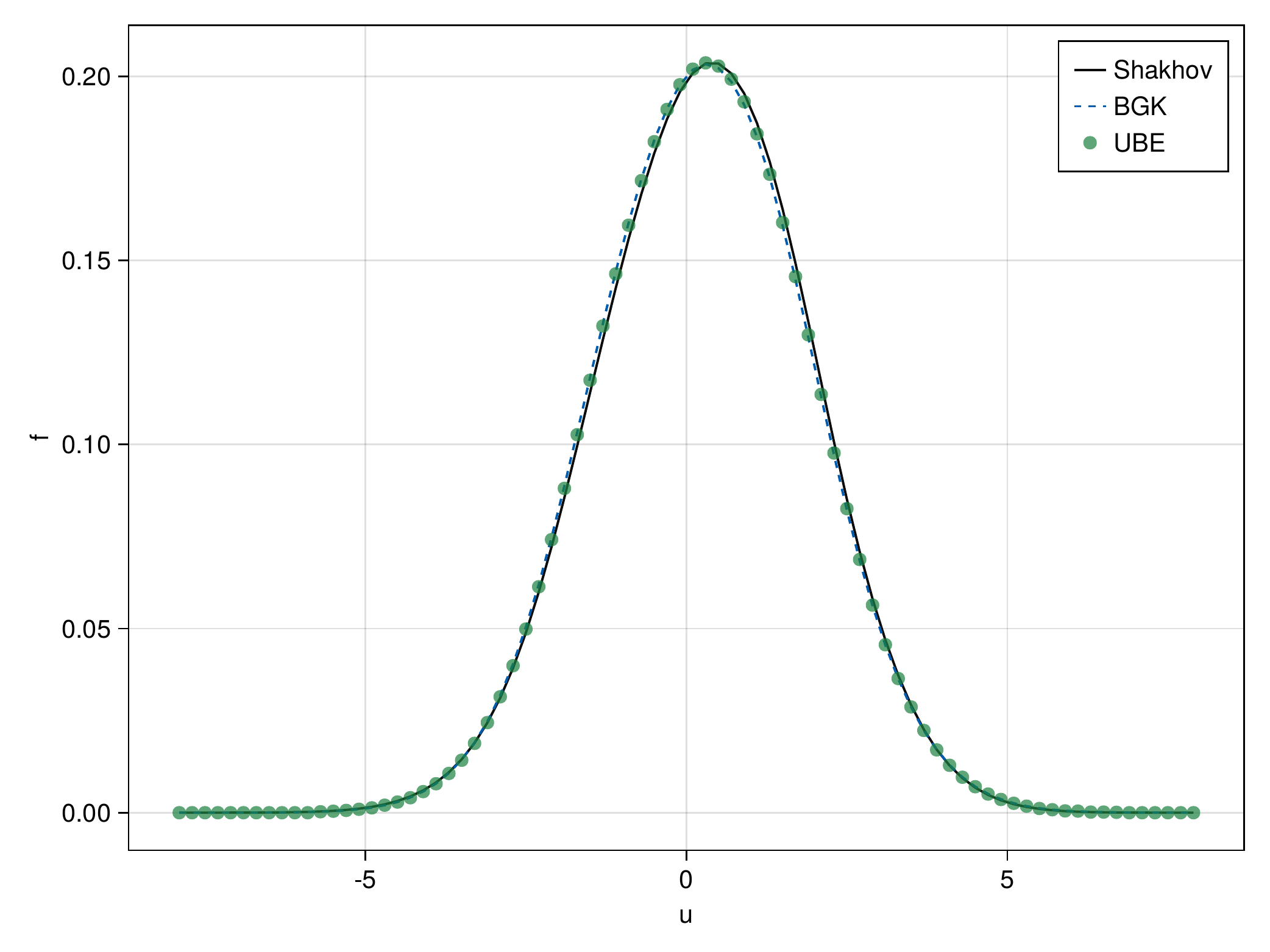}
	}
	\caption{Particle distribution functions at different time instants in the homogeneous relaxation problem (Shakhov model as reference solution).}
    \label{fig:shakhov solution}
\end{figure}

\begin{figure}[htb!]
	\centering
	\subfigure[$t=0$]{
		\includegraphics[width=0.31\textwidth]{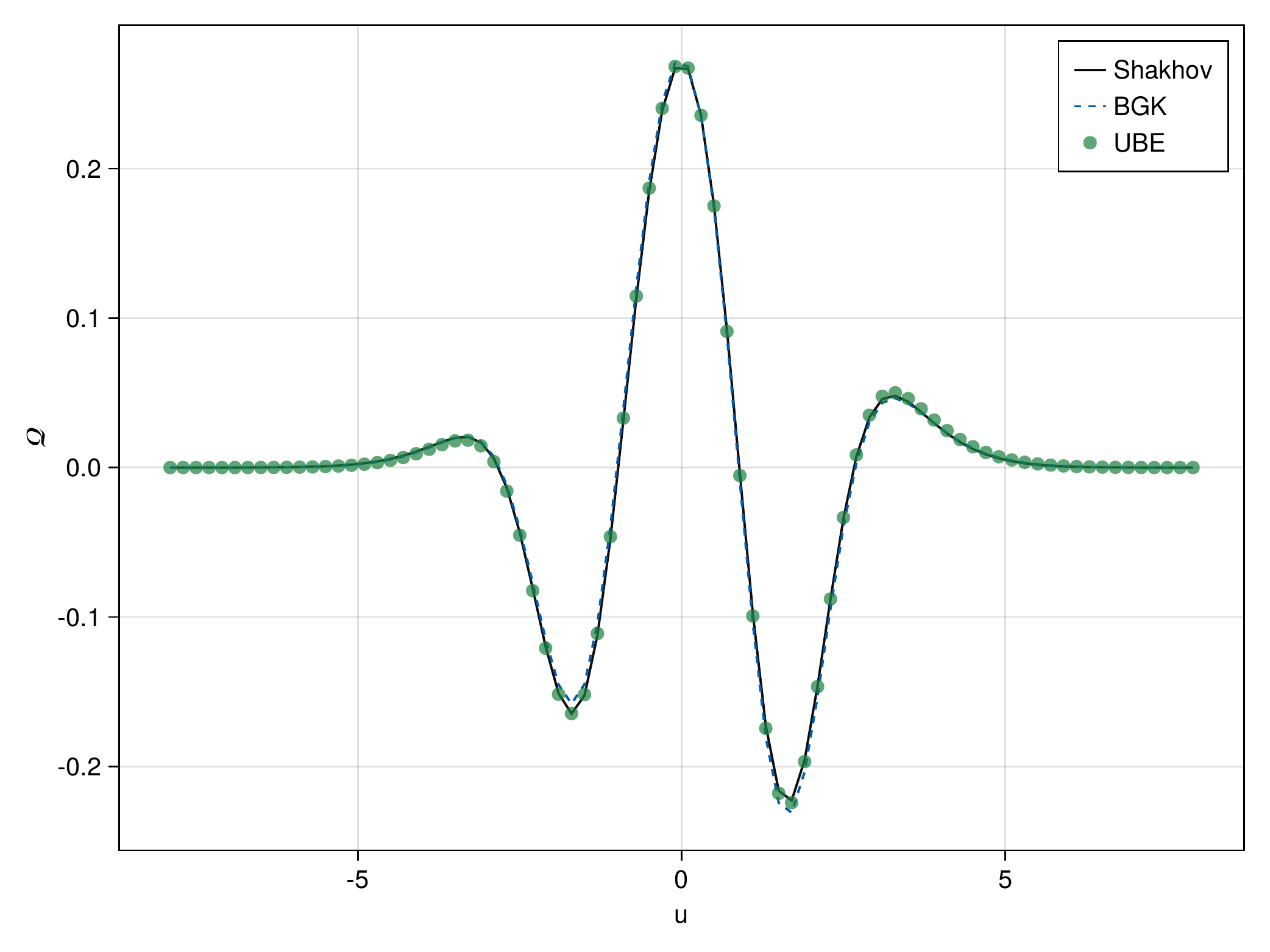}
	}
	\subfigure[$t=0.5$]{
		\includegraphics[width=0.31\textwidth]{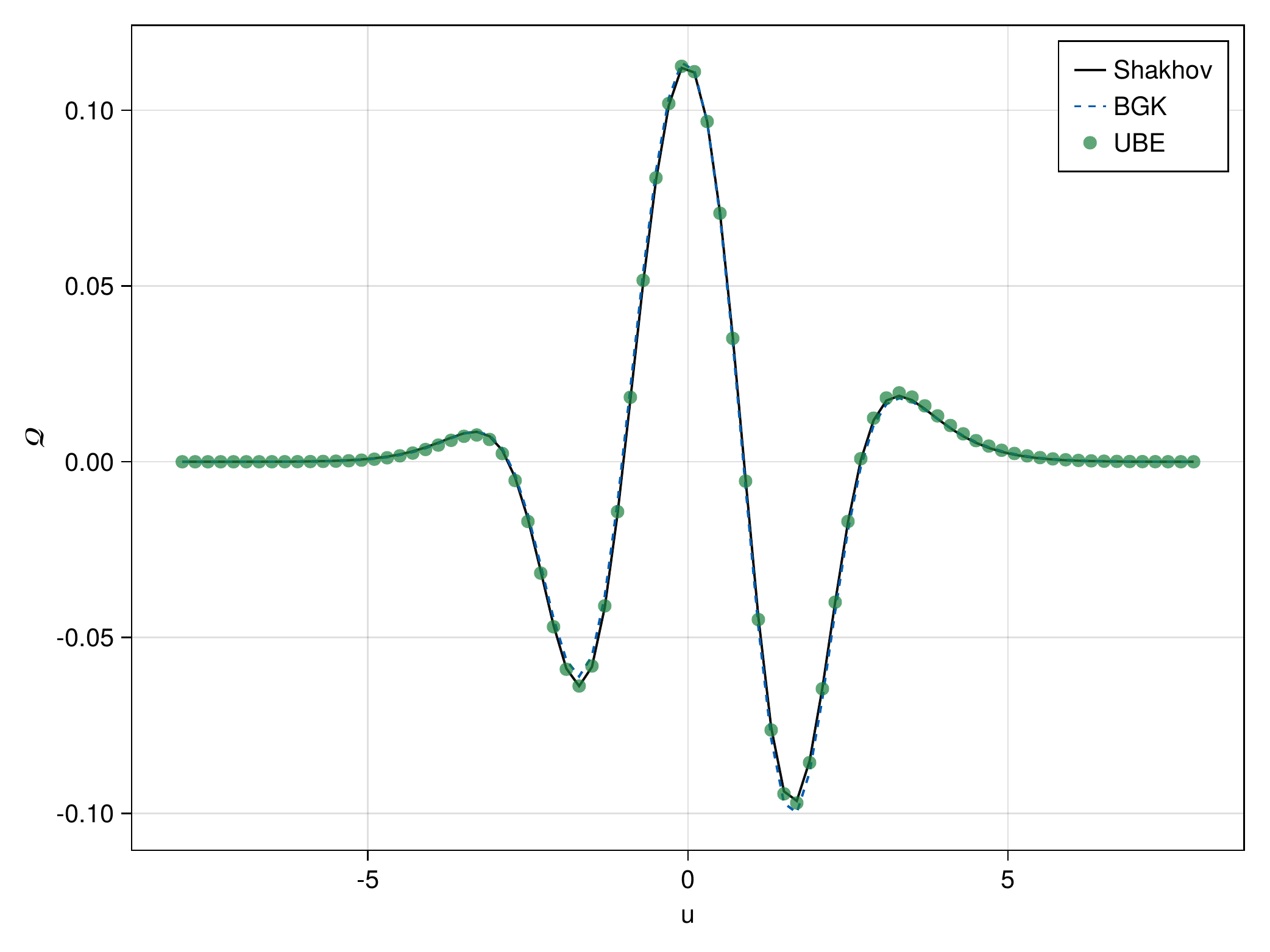}
	}
	\subfigure[$t=1$]{
		\includegraphics[width=0.31\textwidth]{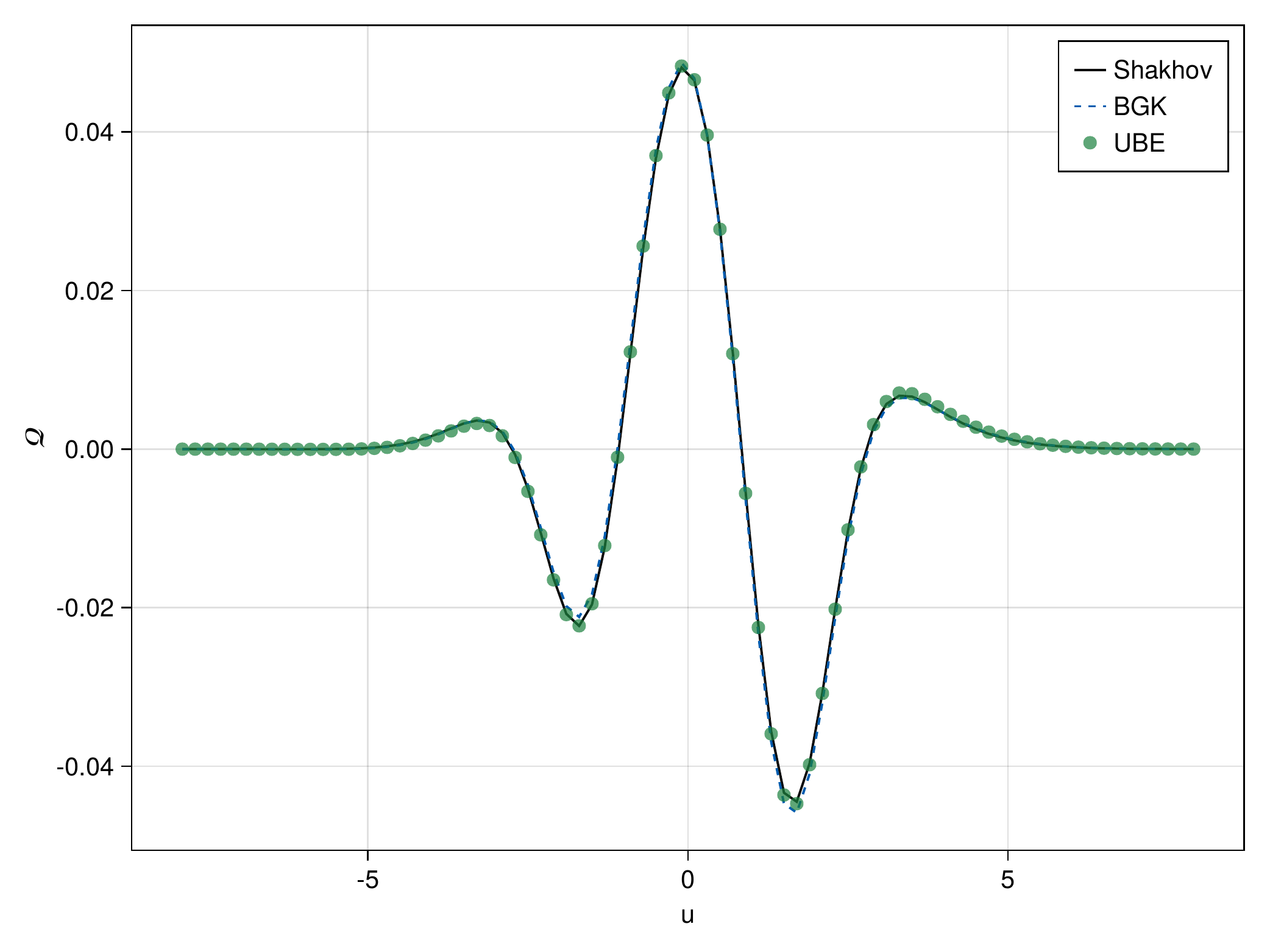}
	}
	\caption{Collision terms at different time instants in the homogeneous relaxation problem (Shakhov model as reference solution).}
    \label{fig:shakhov collision}
\end{figure}

\begin{figure}[htb!]
	\centering
	\subfigure[Density]{
		\includegraphics[width=0.4\textwidth]{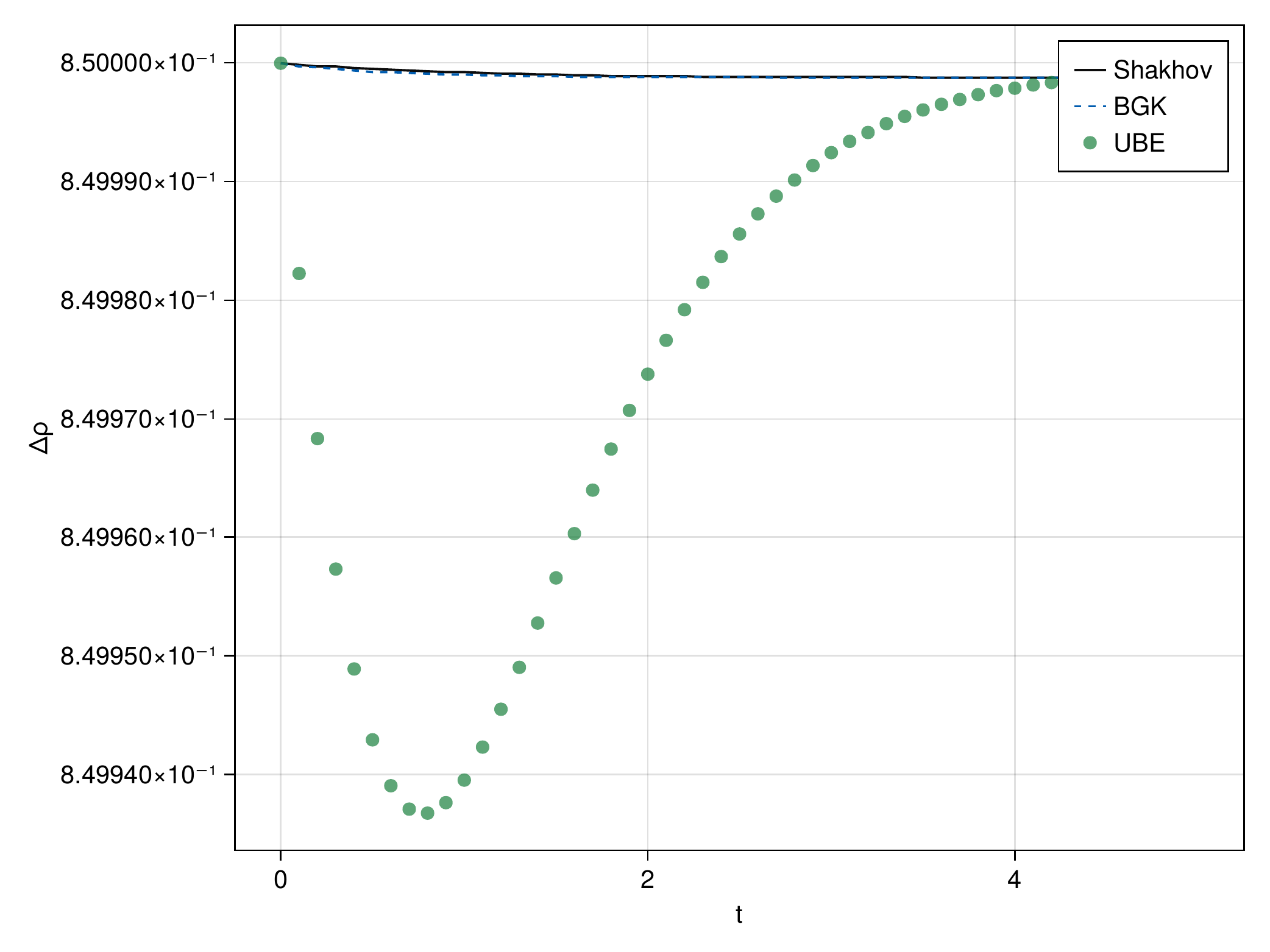}
	}
	\subfigure[Energy]{
		\includegraphics[width=0.4\textwidth]{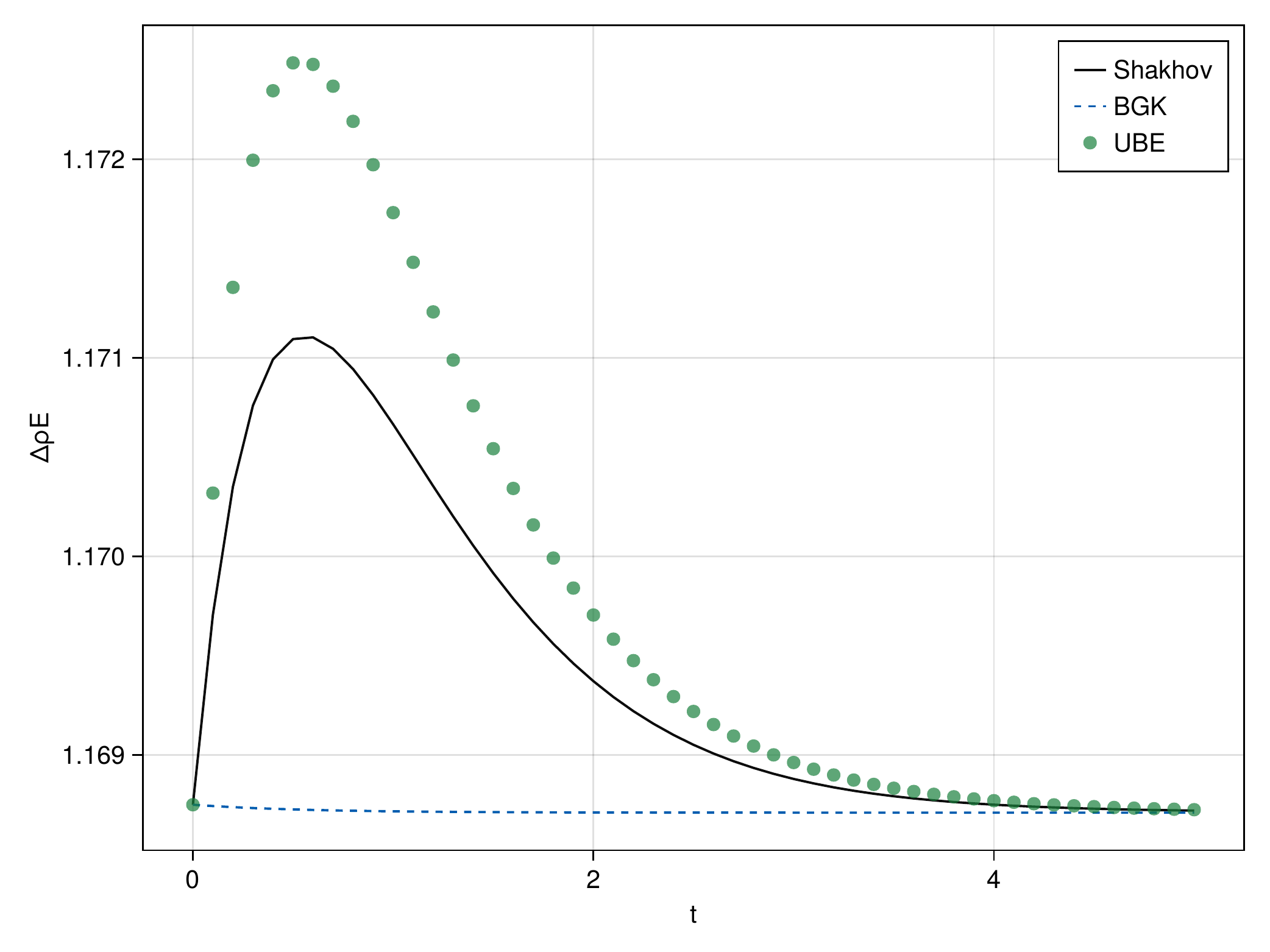}
	}
	\caption{Evolution of total density and energy with time in the homogeneous relaxation problem (Shakhov model as reference solution).}
    \label{fig:shakhov macro}
\end{figure}

\begin{figure}[htb!]
	\centering
	\subfigure[$t=0$]{
		\includegraphics[width=0.31\textwidth]{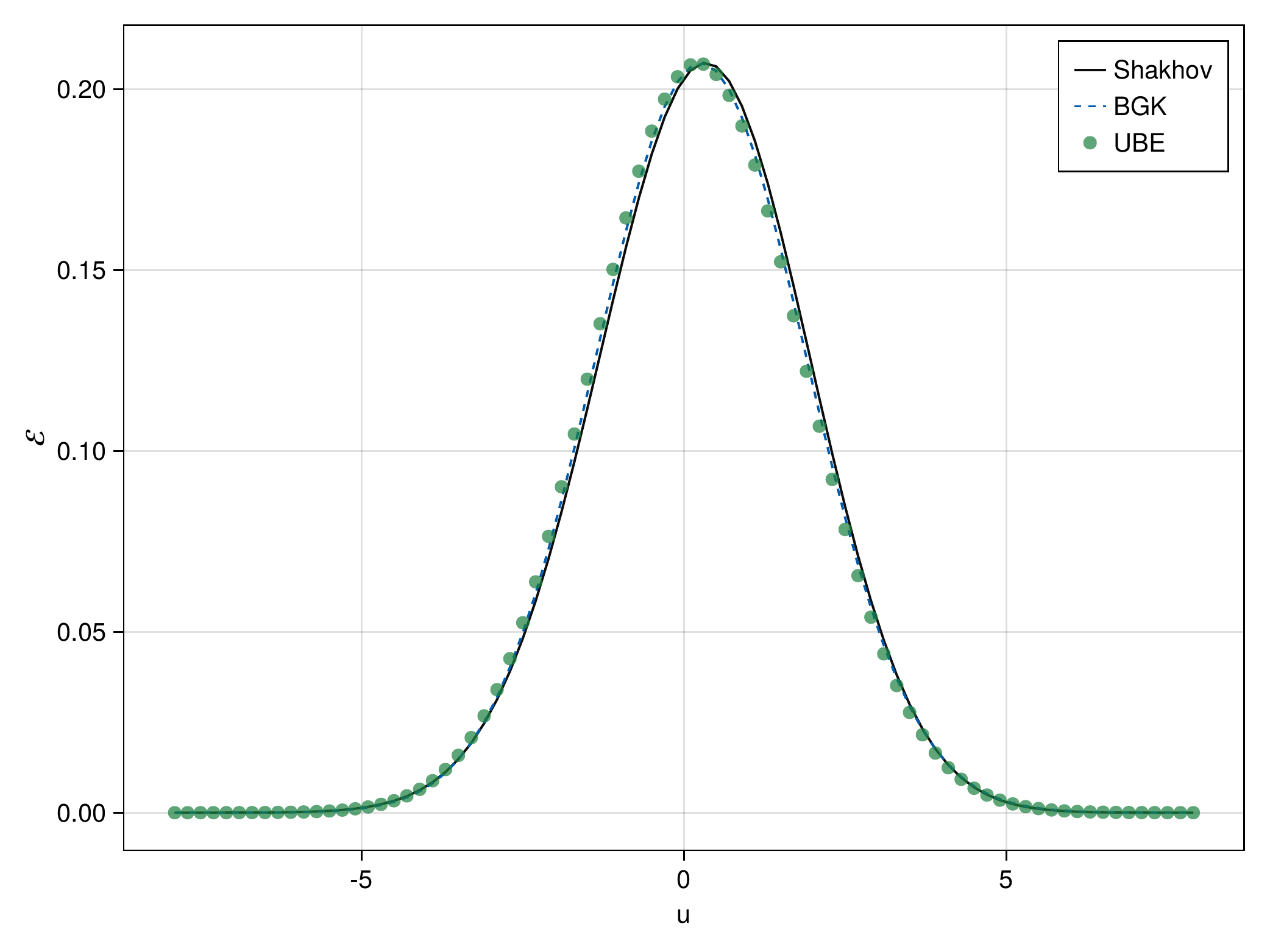}
	}
	\subfigure[$t=0.5$]{
		\includegraphics[width=0.31\textwidth]{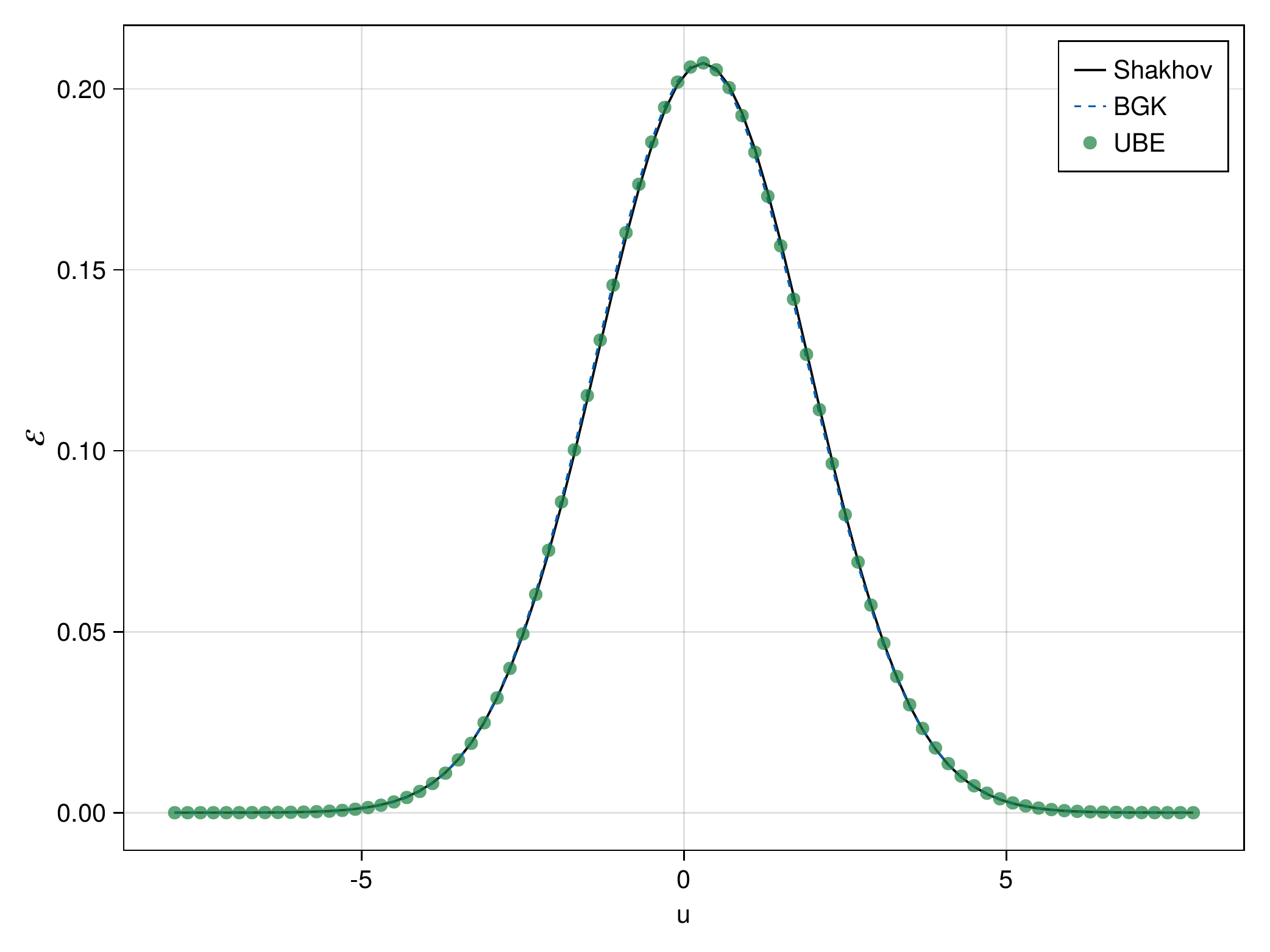}
	}
	\subfigure[$t=1$]{
		\includegraphics[width=0.31\textwidth]{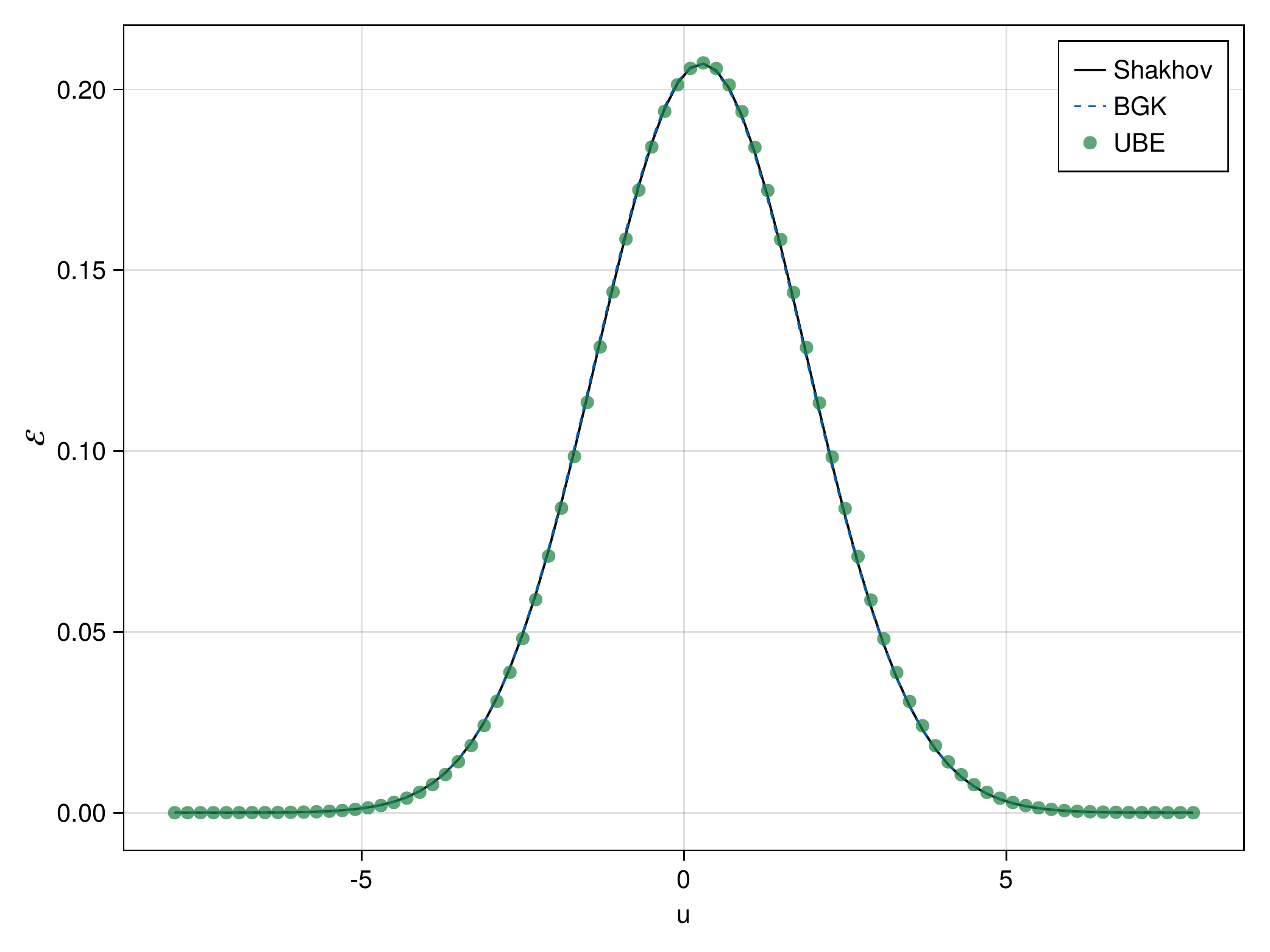}
	}
	\caption{Equilibrium functions at different time instants in the homogeneous relaxation problem (Shakhov model as reference solution).}
    \label{fig:shakhov equilibrium}
\end{figure}

\begin{figure}[htb!]
	\centering
	\subfigure[$t=0.5$]{
		\includegraphics[width=0.31\textwidth]{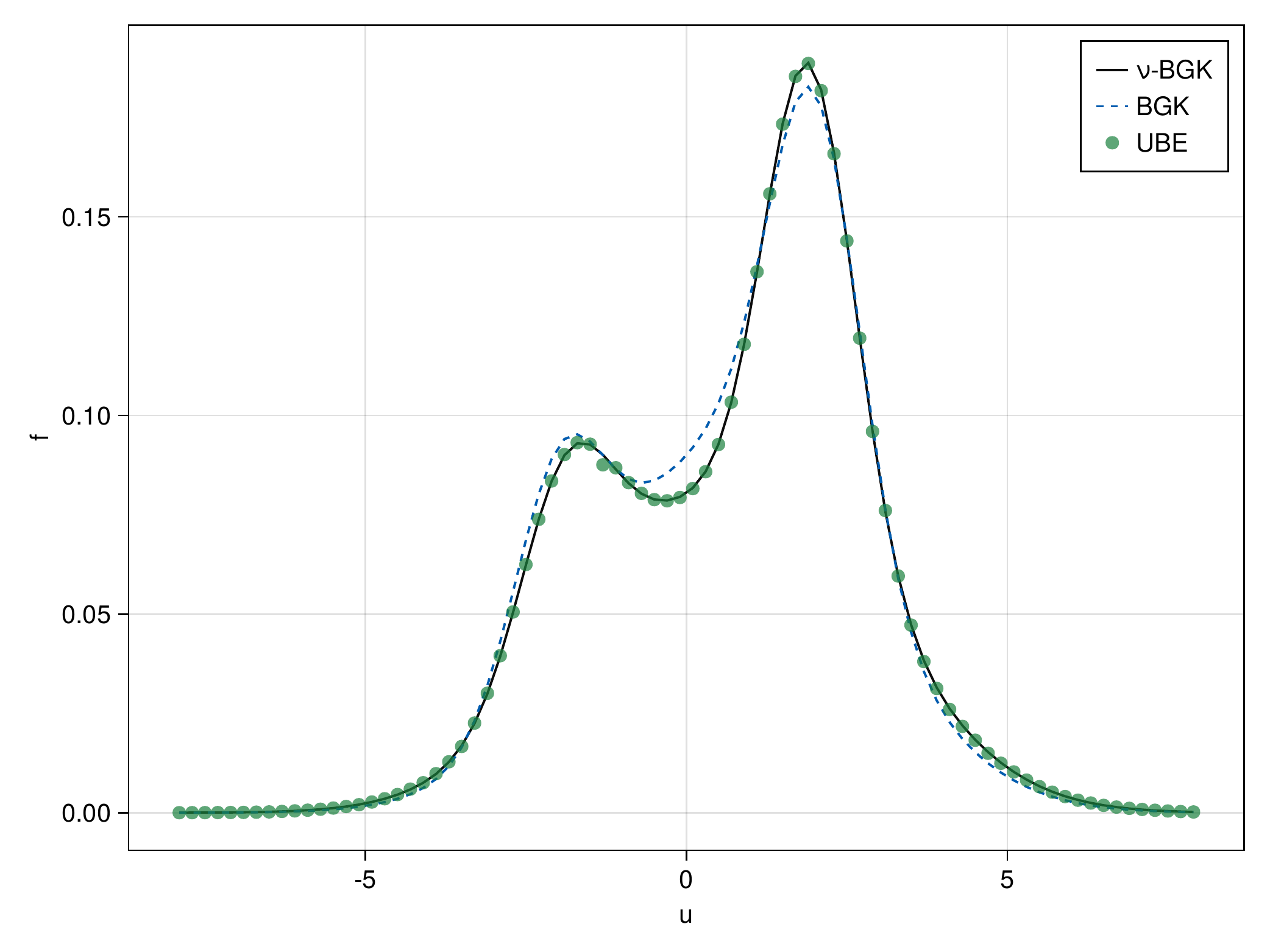}
	}
	\subfigure[$t=1$]{
		\includegraphics[width=0.31\textwidth]{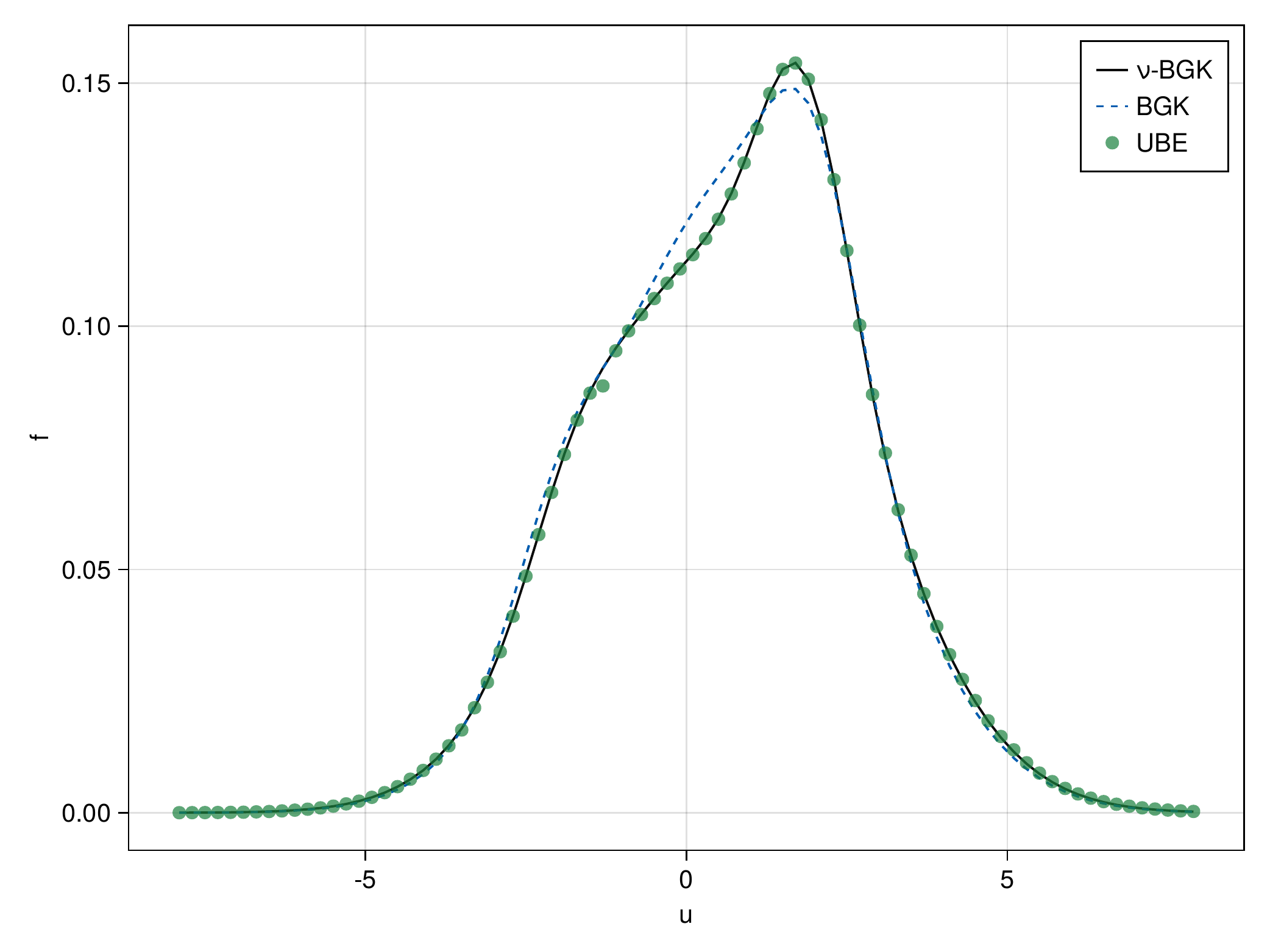}
	}
	\subfigure[$t=2$]{
		\includegraphics[width=0.31\textwidth]{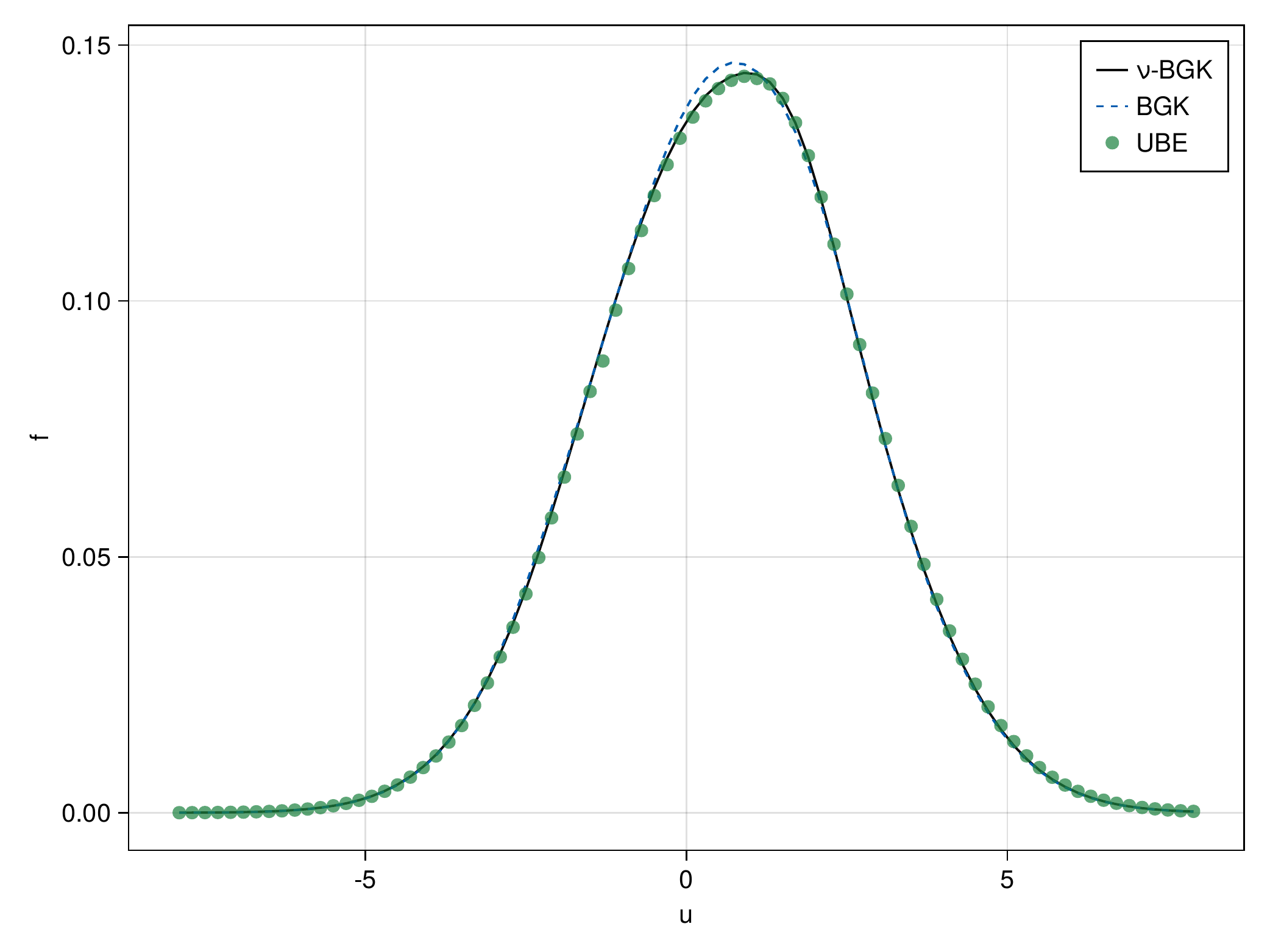}
	}
	\caption{Particle distribution functions at different time instants in the homogeneous relaxation problem ($\nu$-BGK model as reference solution).}
    \label{fig:nubgk solution}
\end{figure}

\begin{figure}[htb!]
	\centering
	\subfigure[$t=0$]{
		\includegraphics[width=0.31\textwidth]{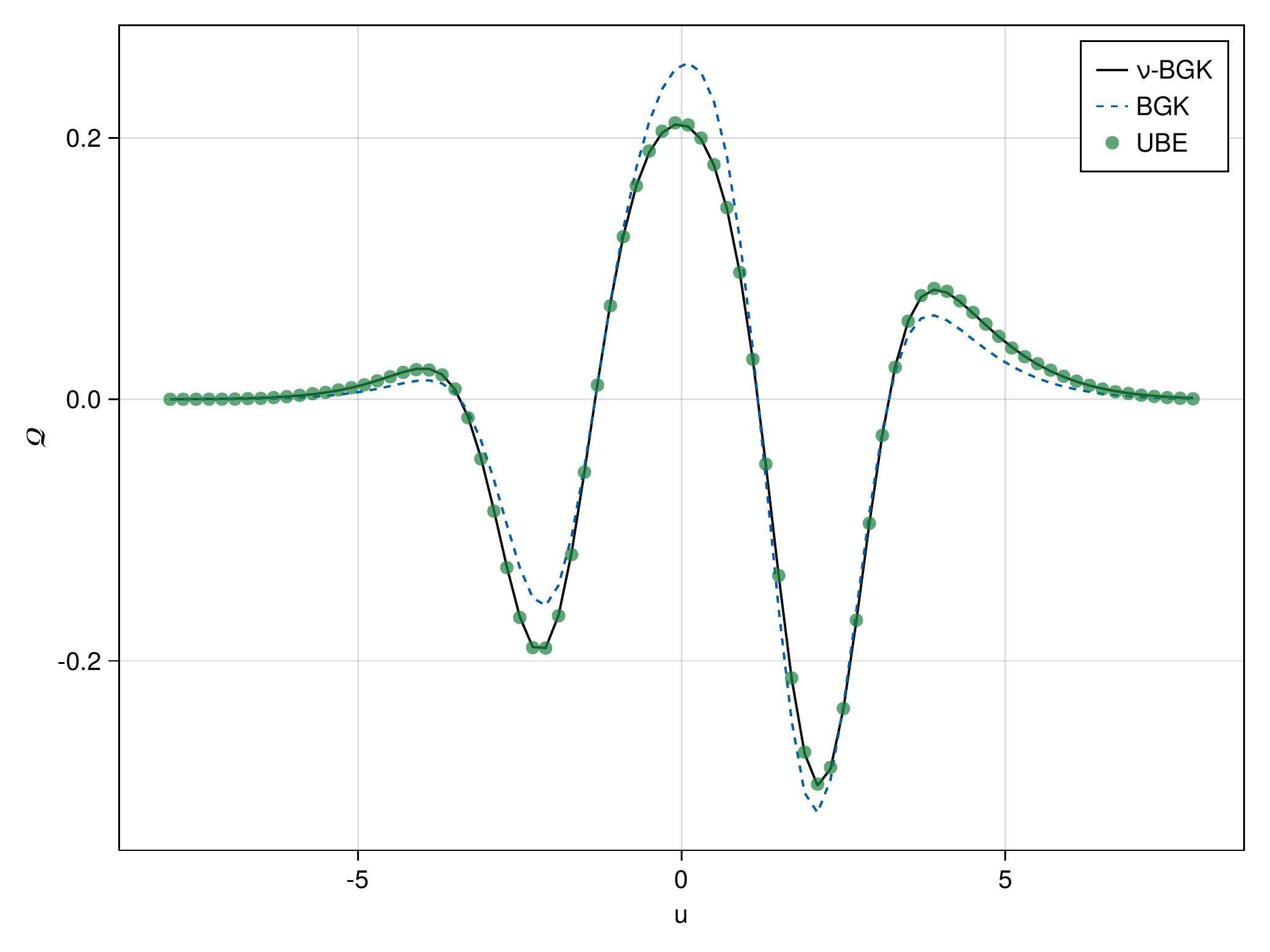}
	}
	\subfigure[$t=0.5$]{
		\includegraphics[width=0.31\textwidth]{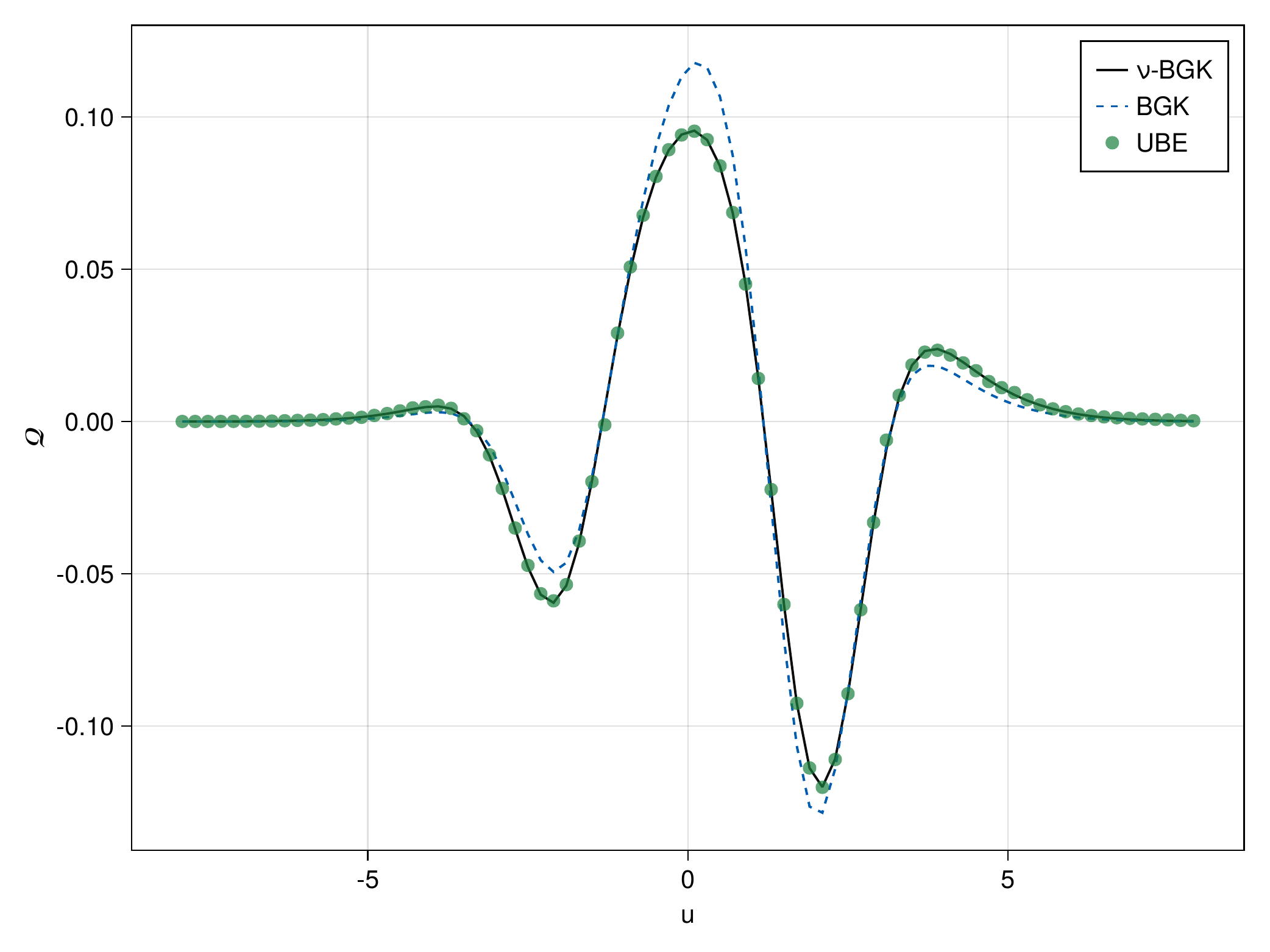}
	}
	\subfigure[$t=1$]{
		\includegraphics[width=0.31\textwidth]{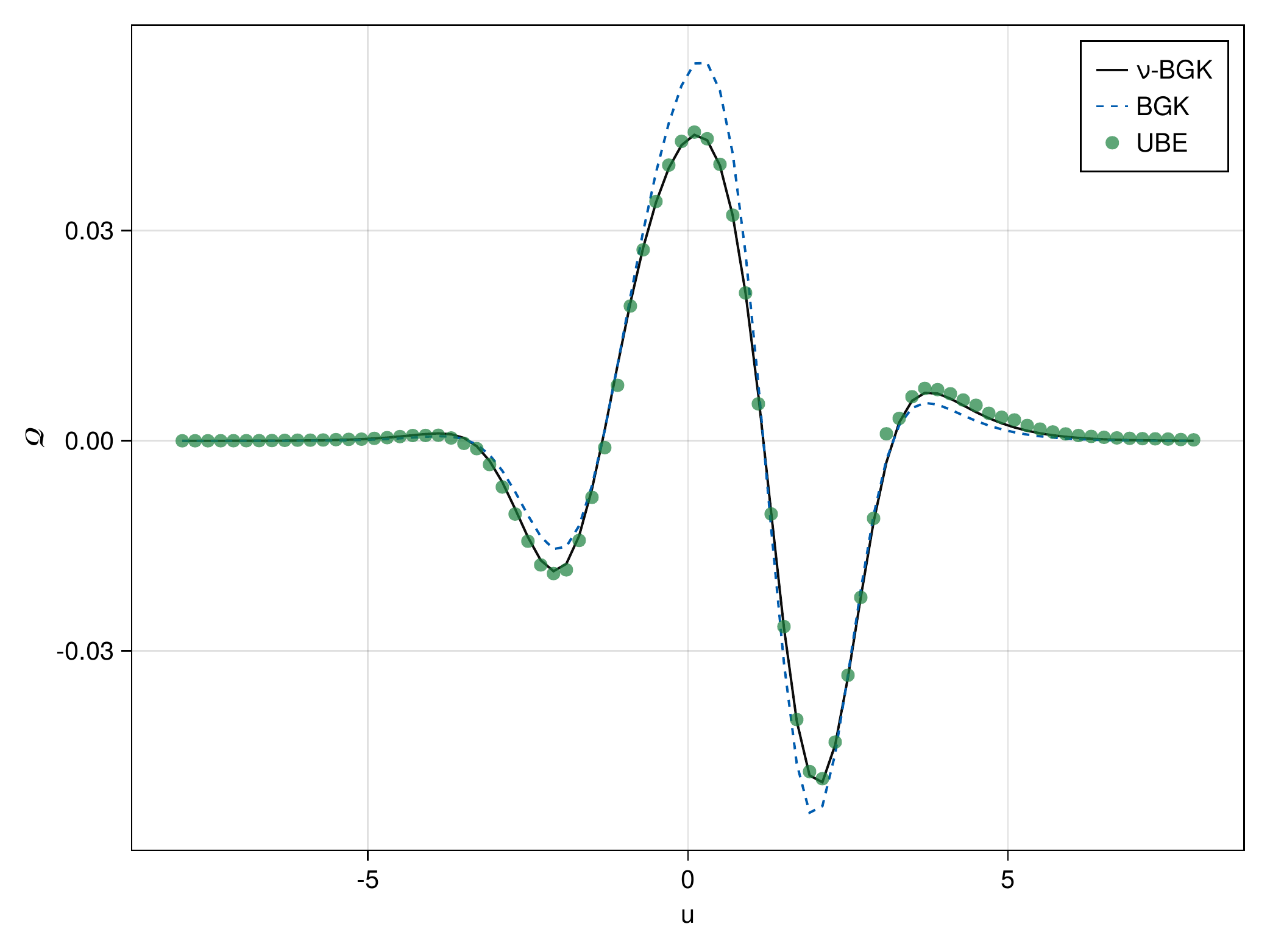}
	}
	\caption{Collision terms at different time instants in the homogeneous relaxation problem ($\nu$-BGK model as reference solution).}
    \label{fig:nubgk collision}
\end{figure}

\begin{figure}[htb!]
	\centering
	\subfigure[$t=0$]{
		\includegraphics[width=0.31\textwidth]{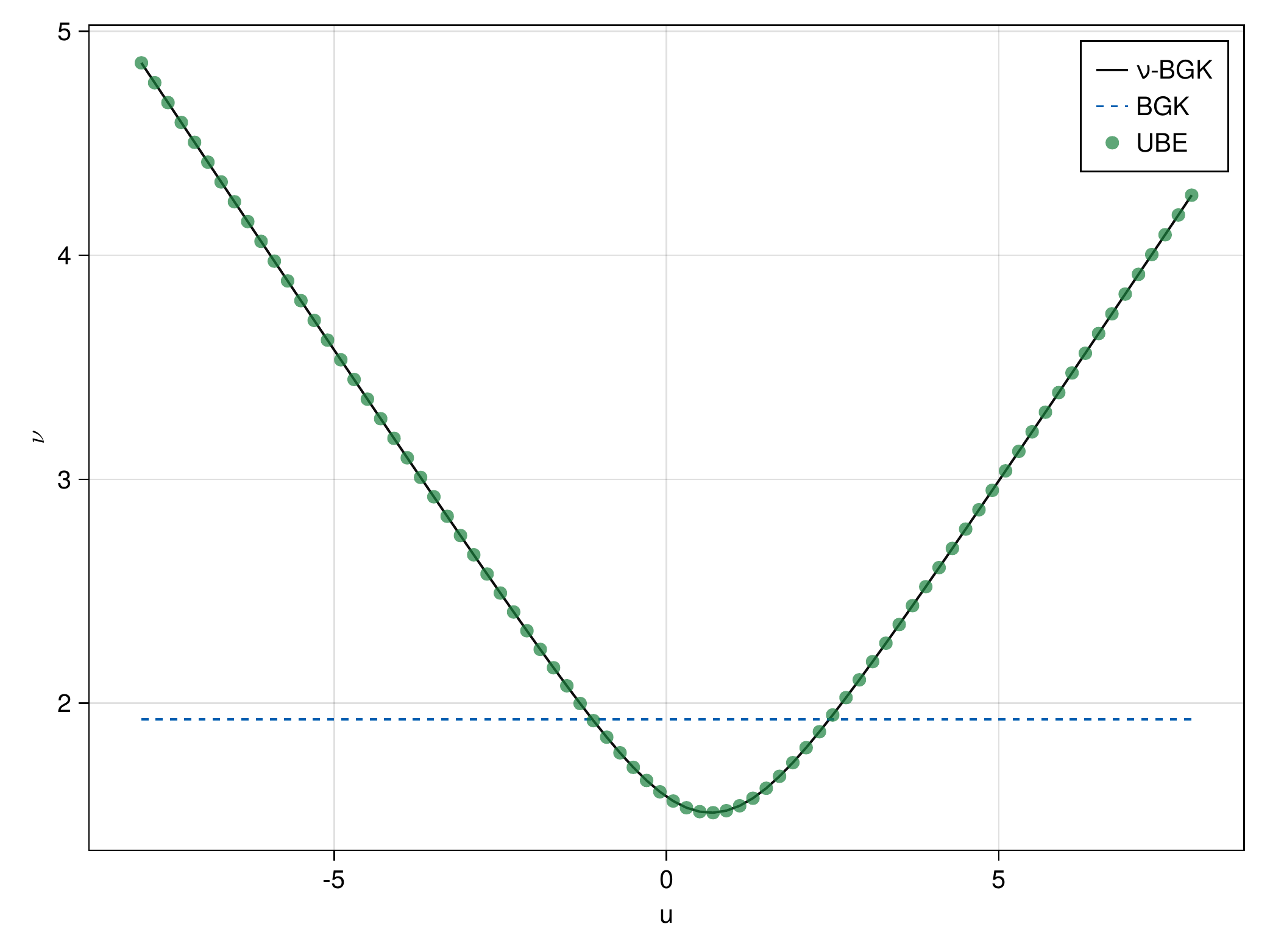}
	}
	\subfigure[$t=0.5$]{
		\includegraphics[width=0.31\textwidth]{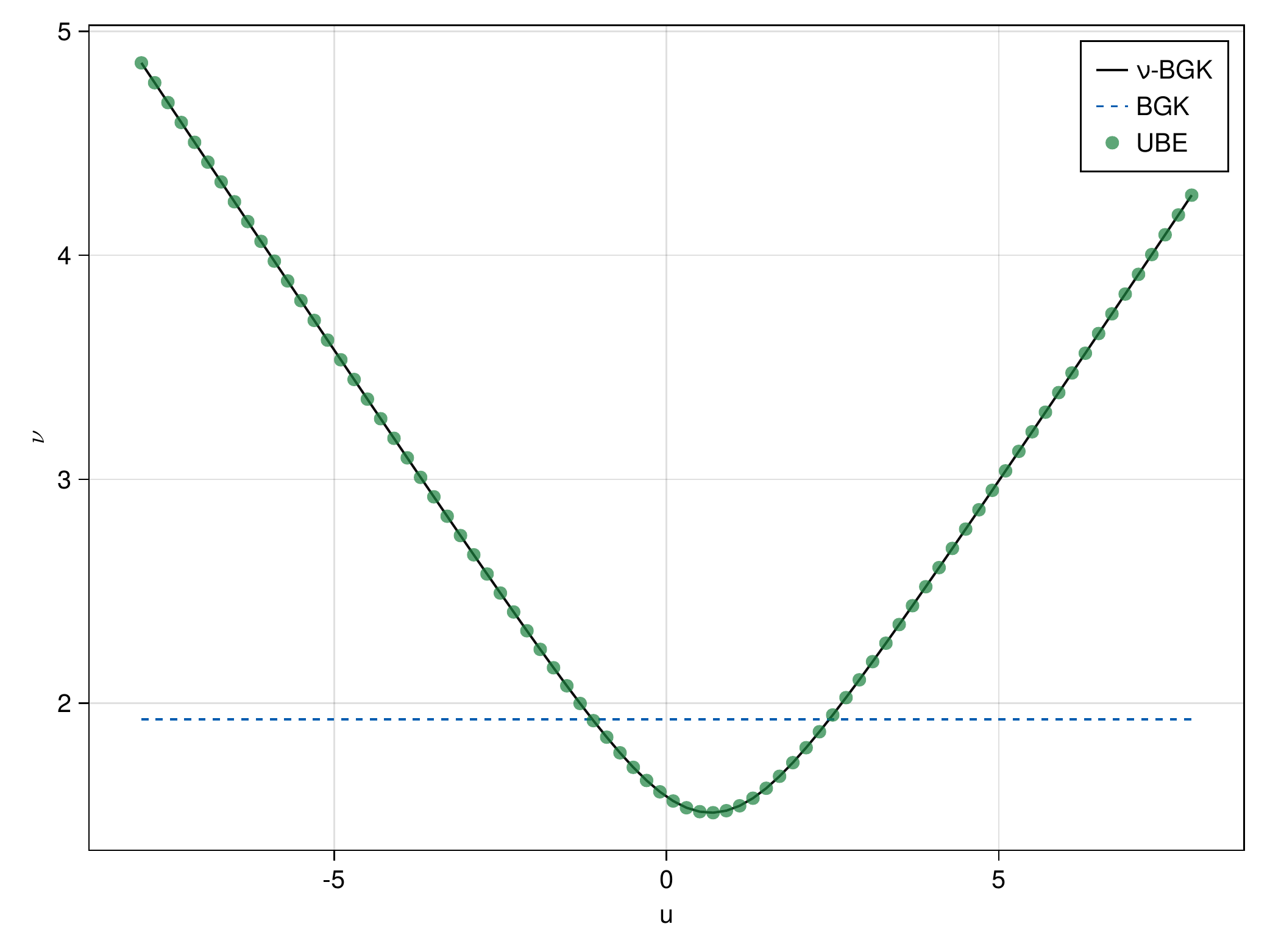}
	}
	\subfigure[$t=1$]{
		\includegraphics[width=0.31\textwidth]{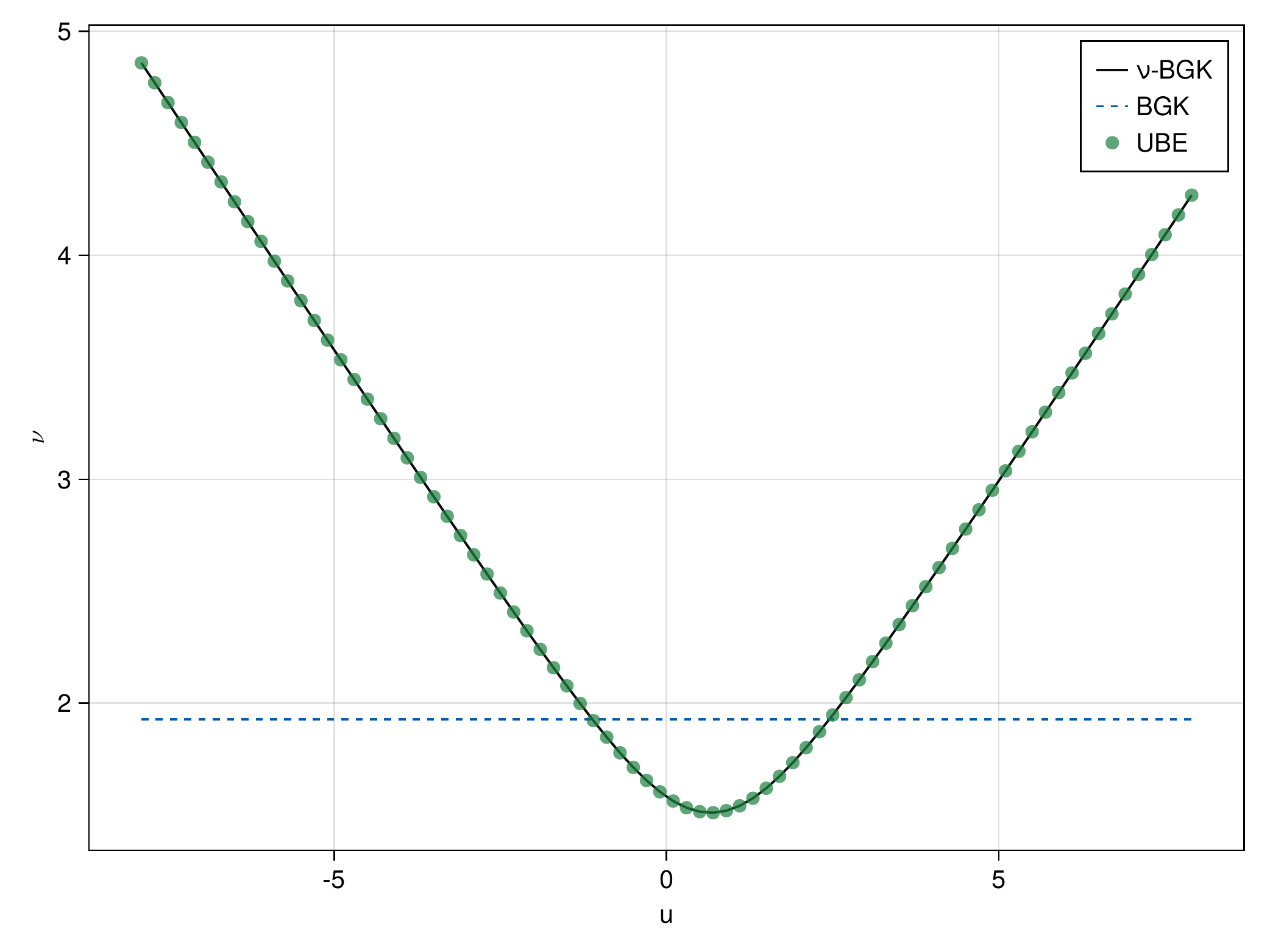}
	}
	\caption{Relaxation frequencies at different time instants in the homogeneous relaxation problem ($\nu$-BGK model as reference solution).}
    \label{fig:nubgk frequency}
\end{figure}

\begin{figure}[htb!]
	\centering
	\subfigure[Density]{
		\includegraphics[width=0.4\textwidth]{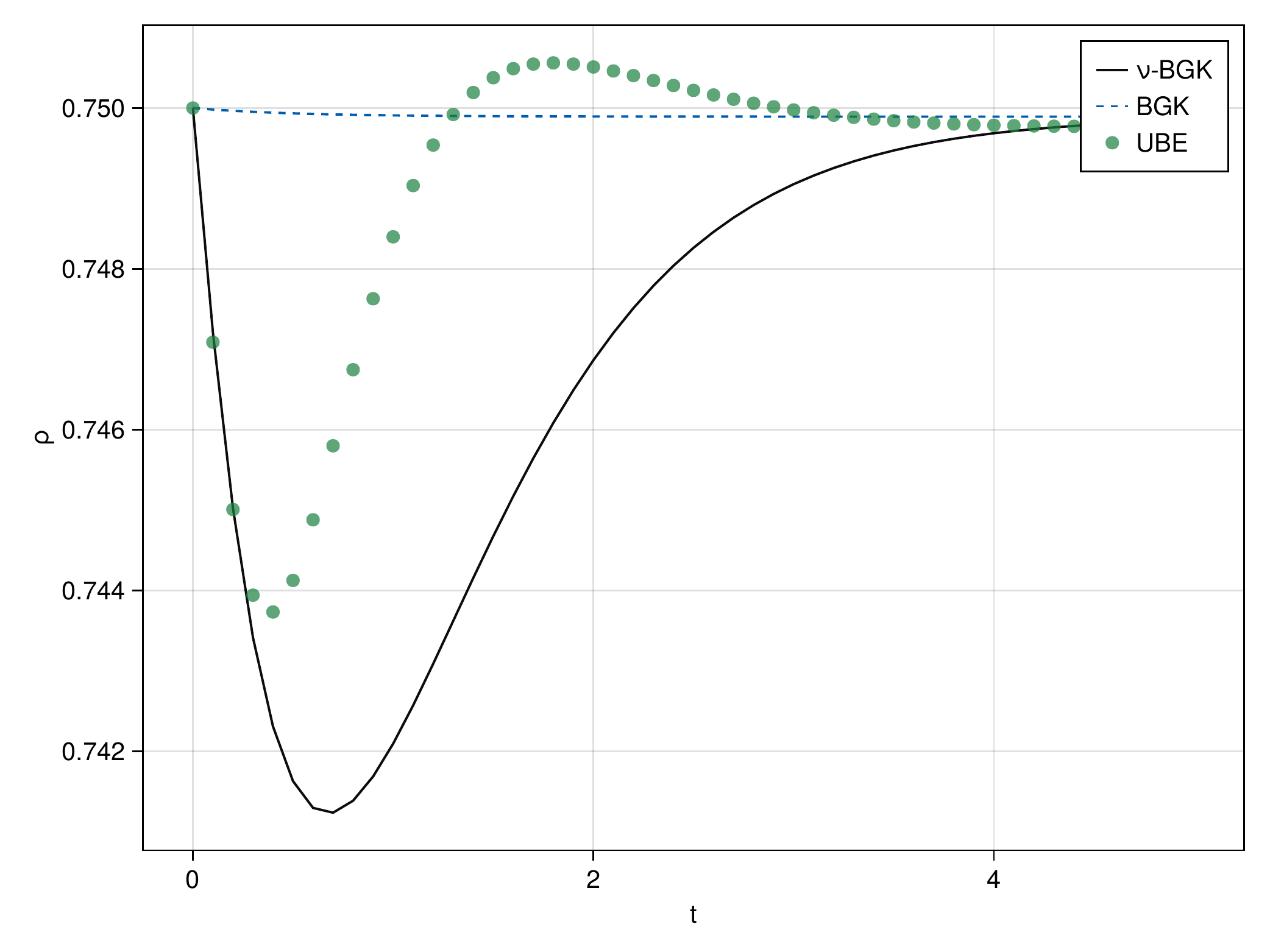}
	}
	\subfigure[Energy]{
		\includegraphics[width=0.4\textwidth]{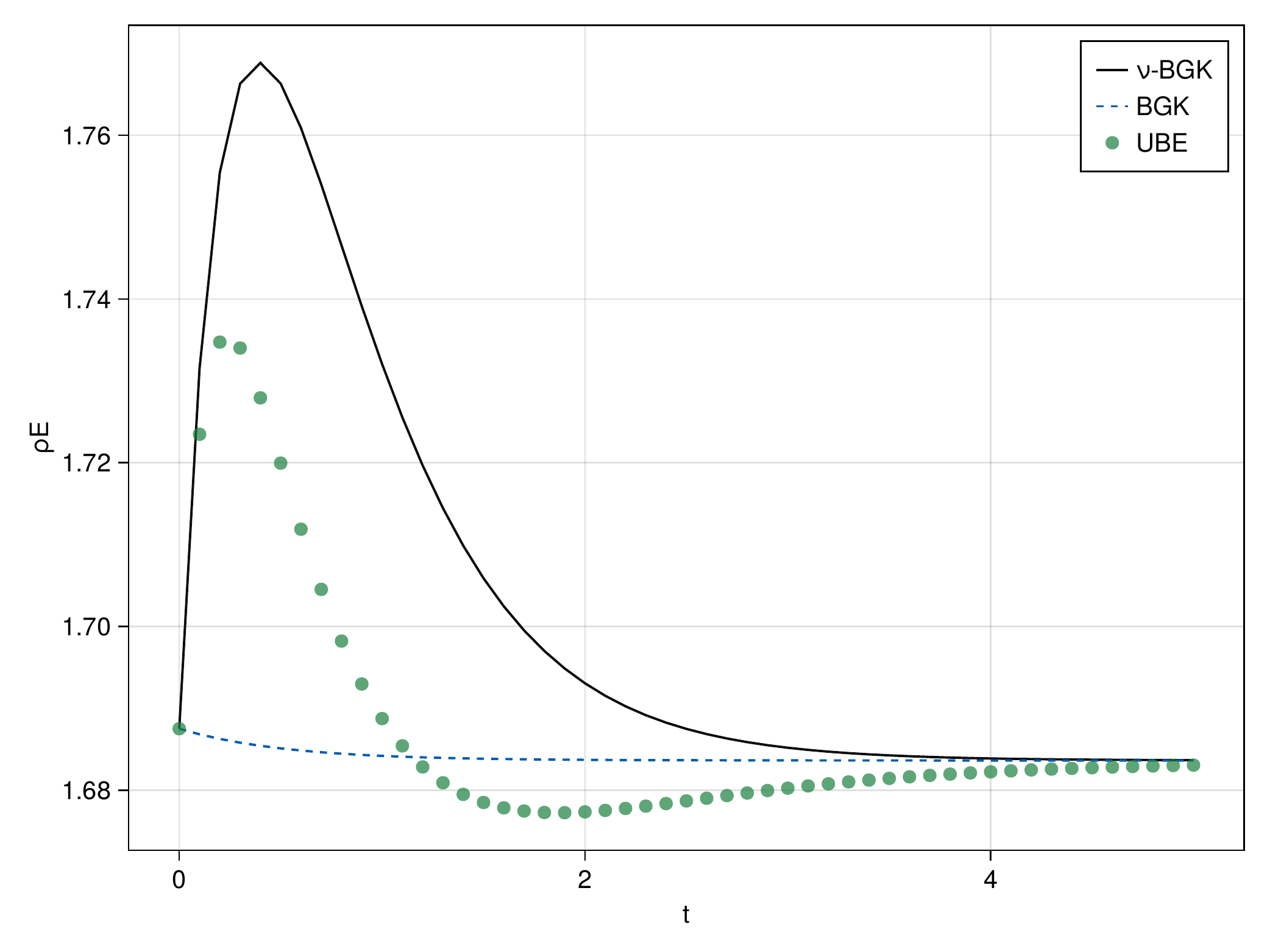}
	}
	\caption{Evolution of total density and energy with time in the homogeneous relaxation problem ($\nu$-BGK model as reference solution).}
    \label{fig:nubgk macro}
\end{figure}

% full Boltzmann
\begin{figure}[htb!]
	\centering
	\subfigure[$t=0.2$]{
		\includegraphics[width=0.31\textwidth]{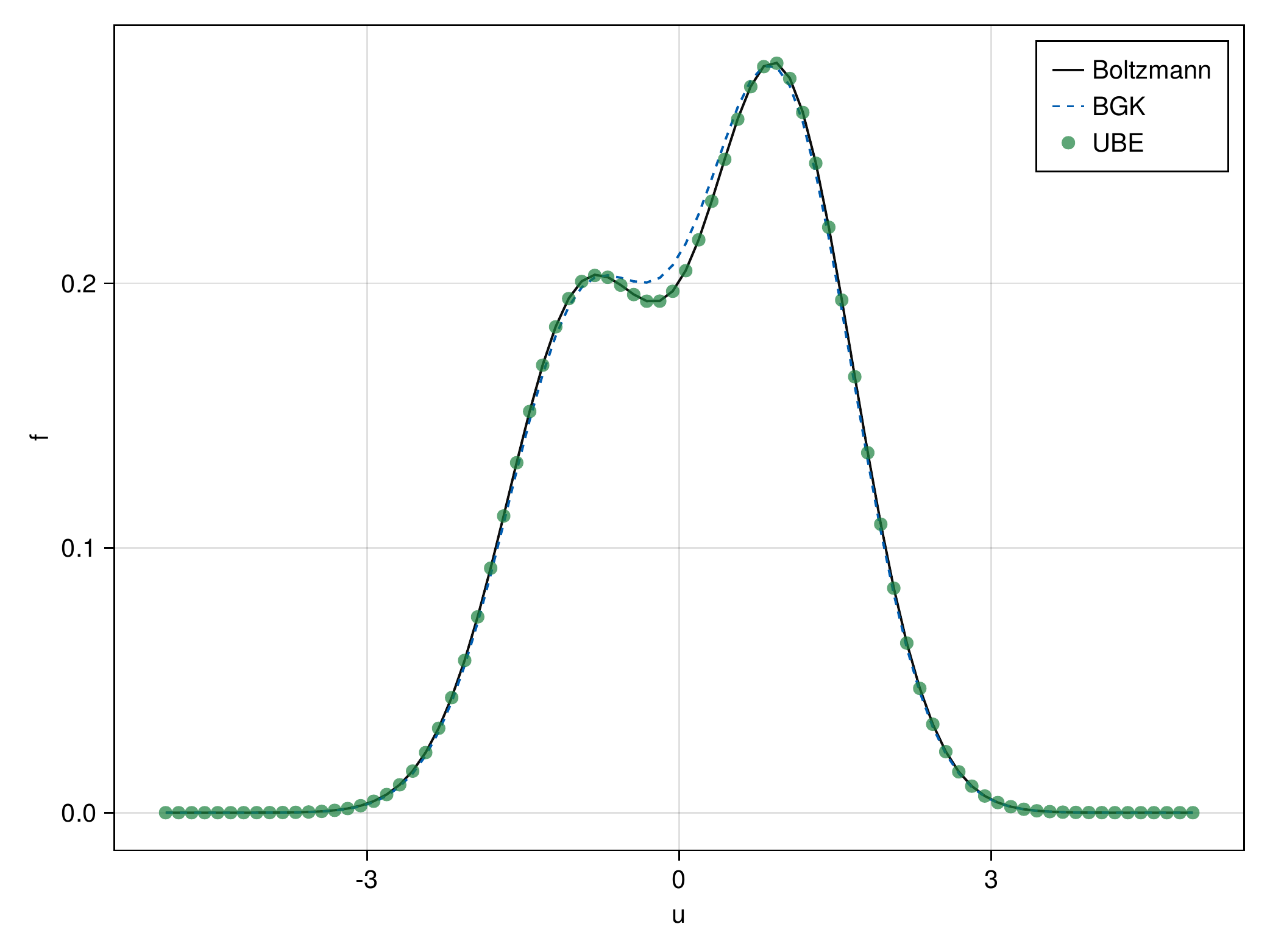}
	}
	\subfigure[$t=0.5$]{
		\includegraphics[width=0.31\textwidth]{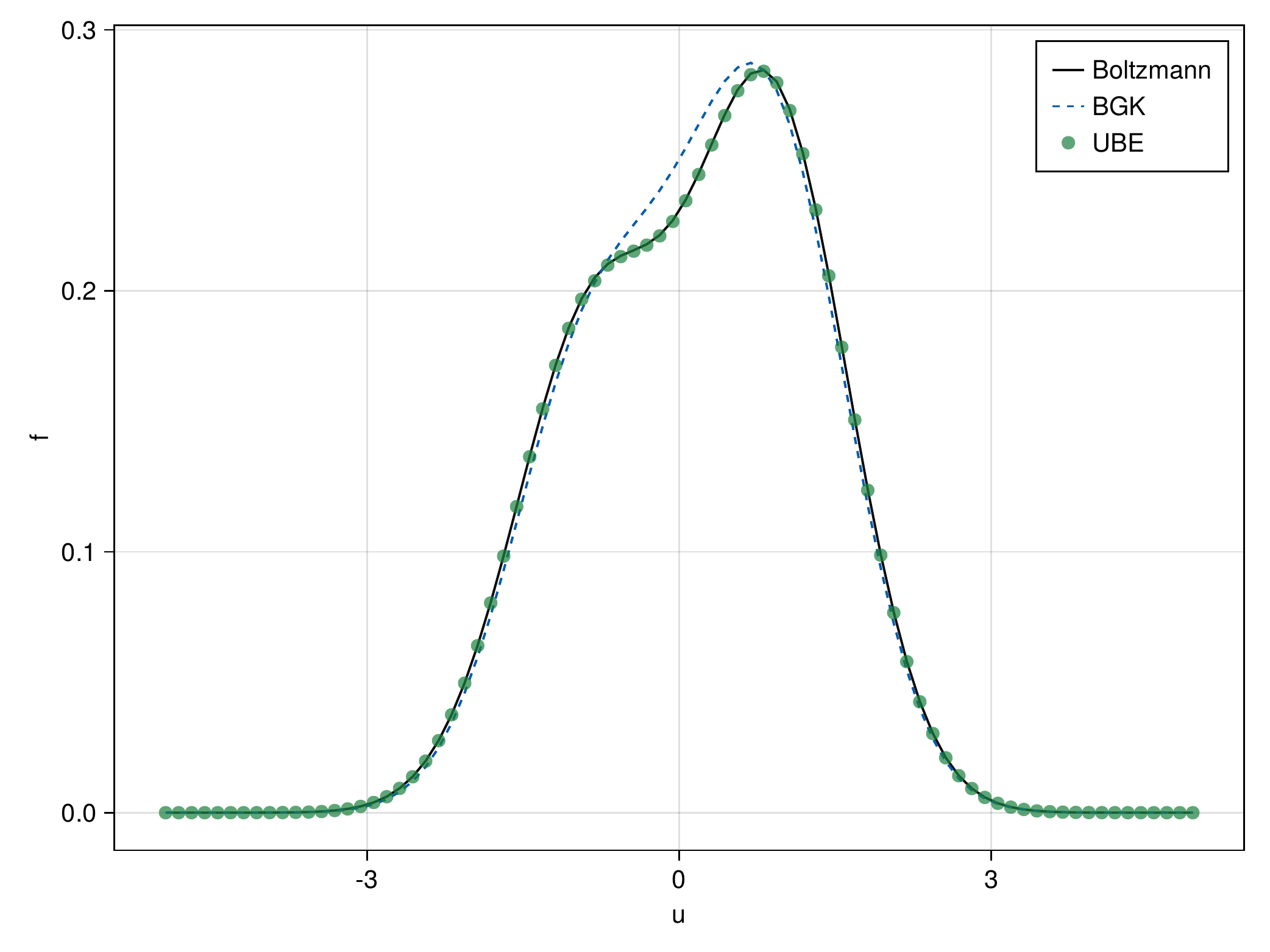}
	}
	\subfigure[$t=1$]{
		\includegraphics[width=0.31\textwidth]{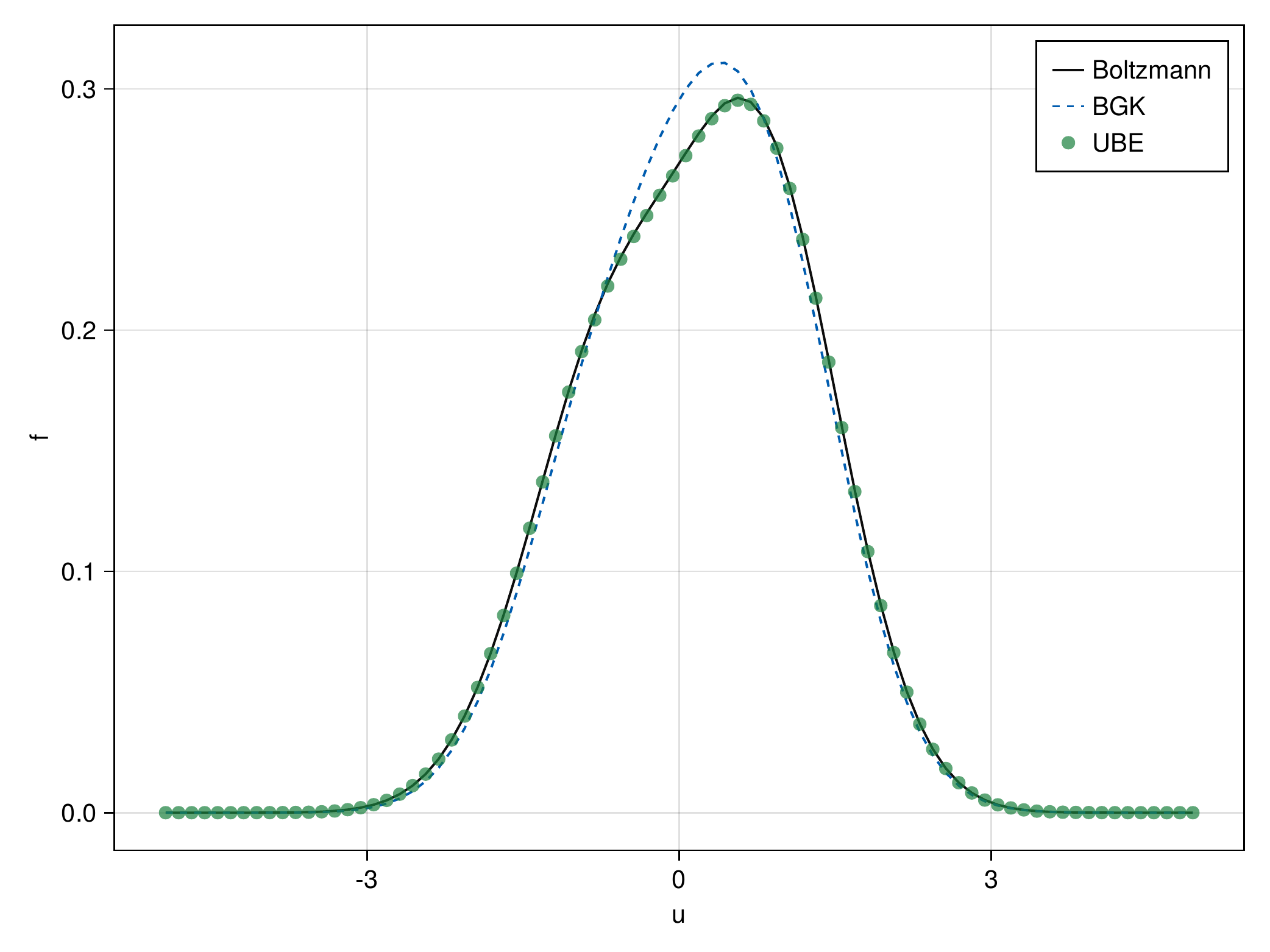}
	}
	\subfigure[$t=2$]{
		\includegraphics[width=0.31\textwidth]{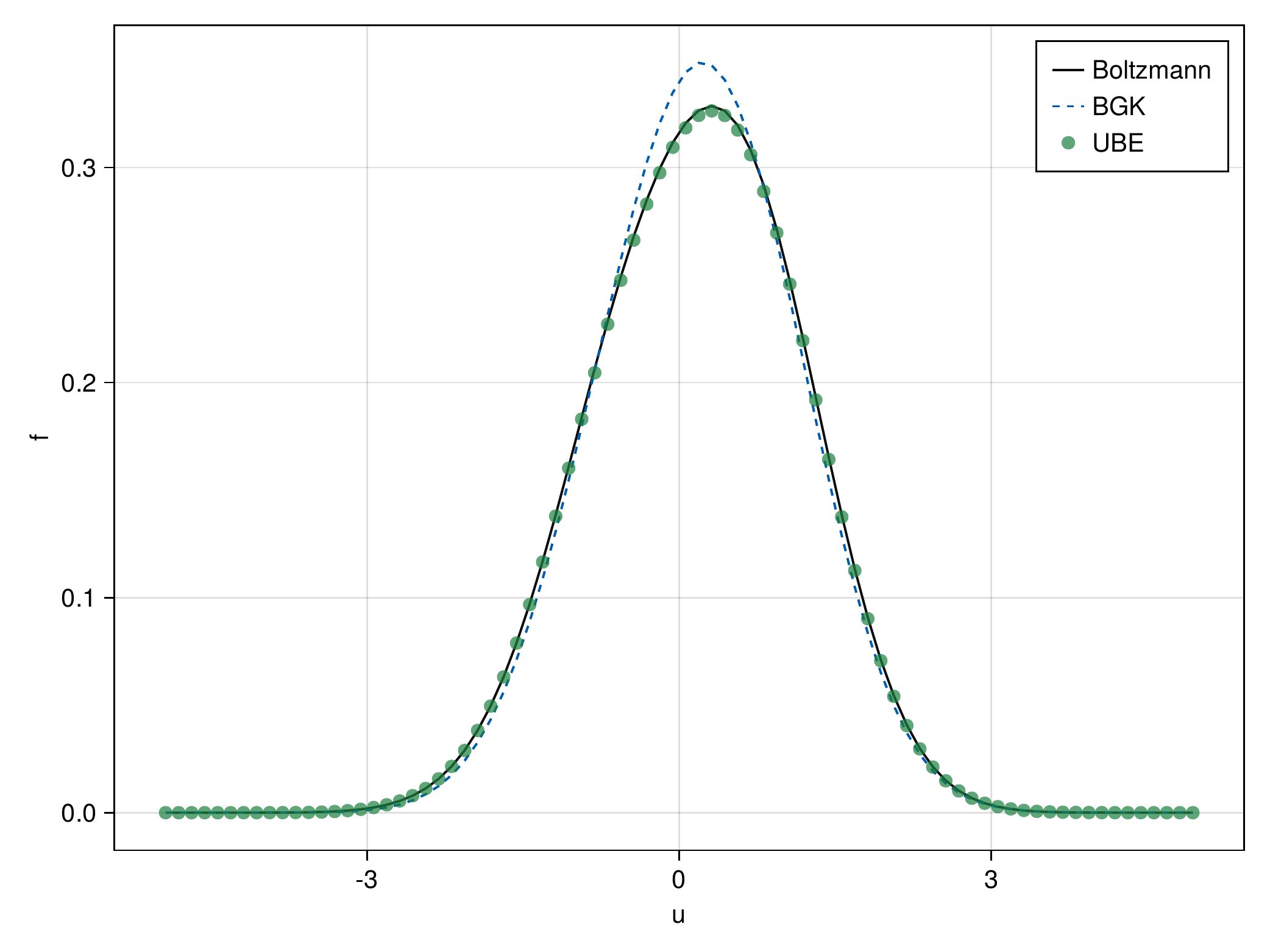}
	}
	\subfigure[$t=4$]{
		\includegraphics[width=0.31\textwidth]{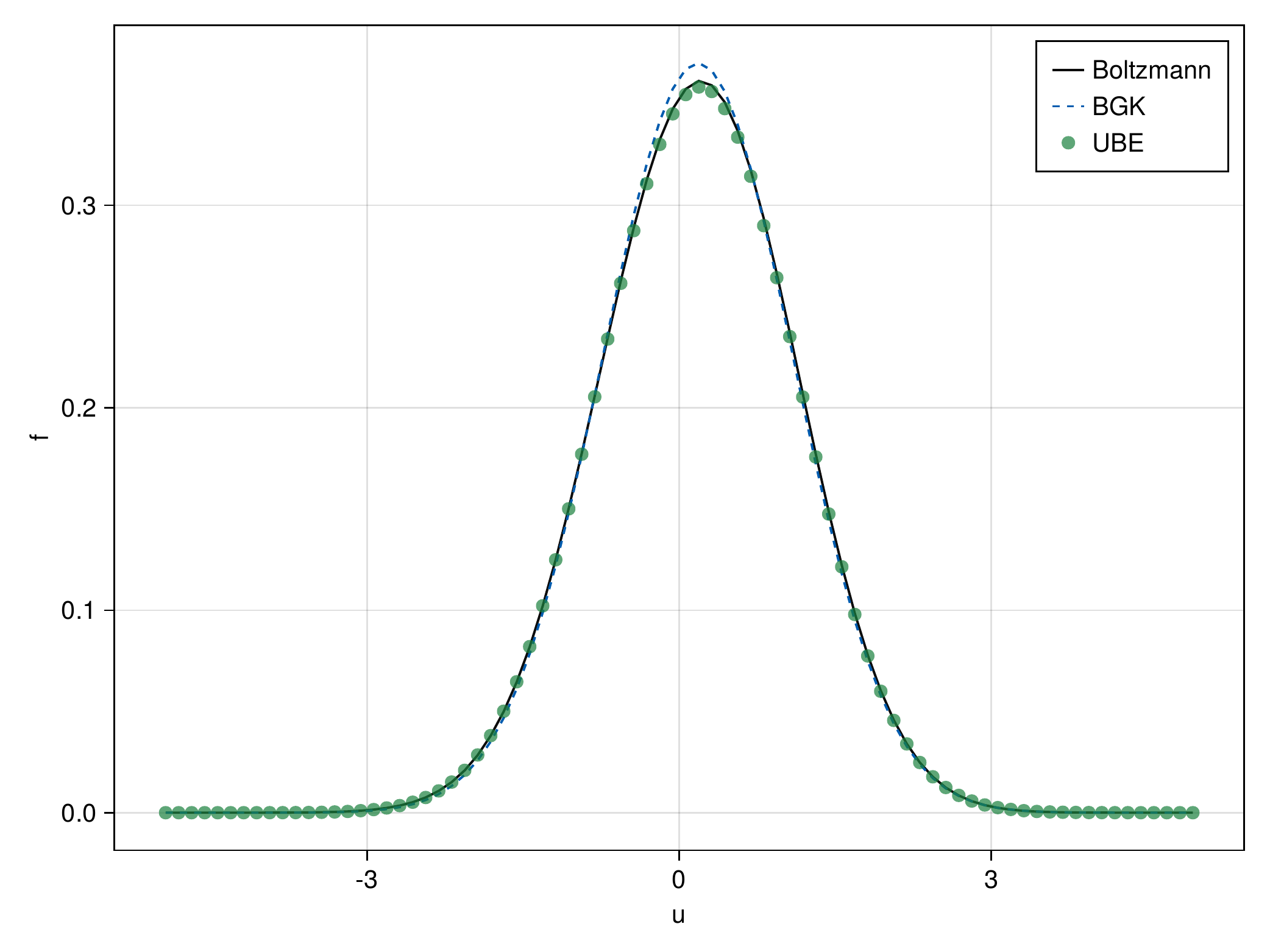}
	}
	\subfigure[$t=6$]{
		\includegraphics[width=0.31\textwidth]{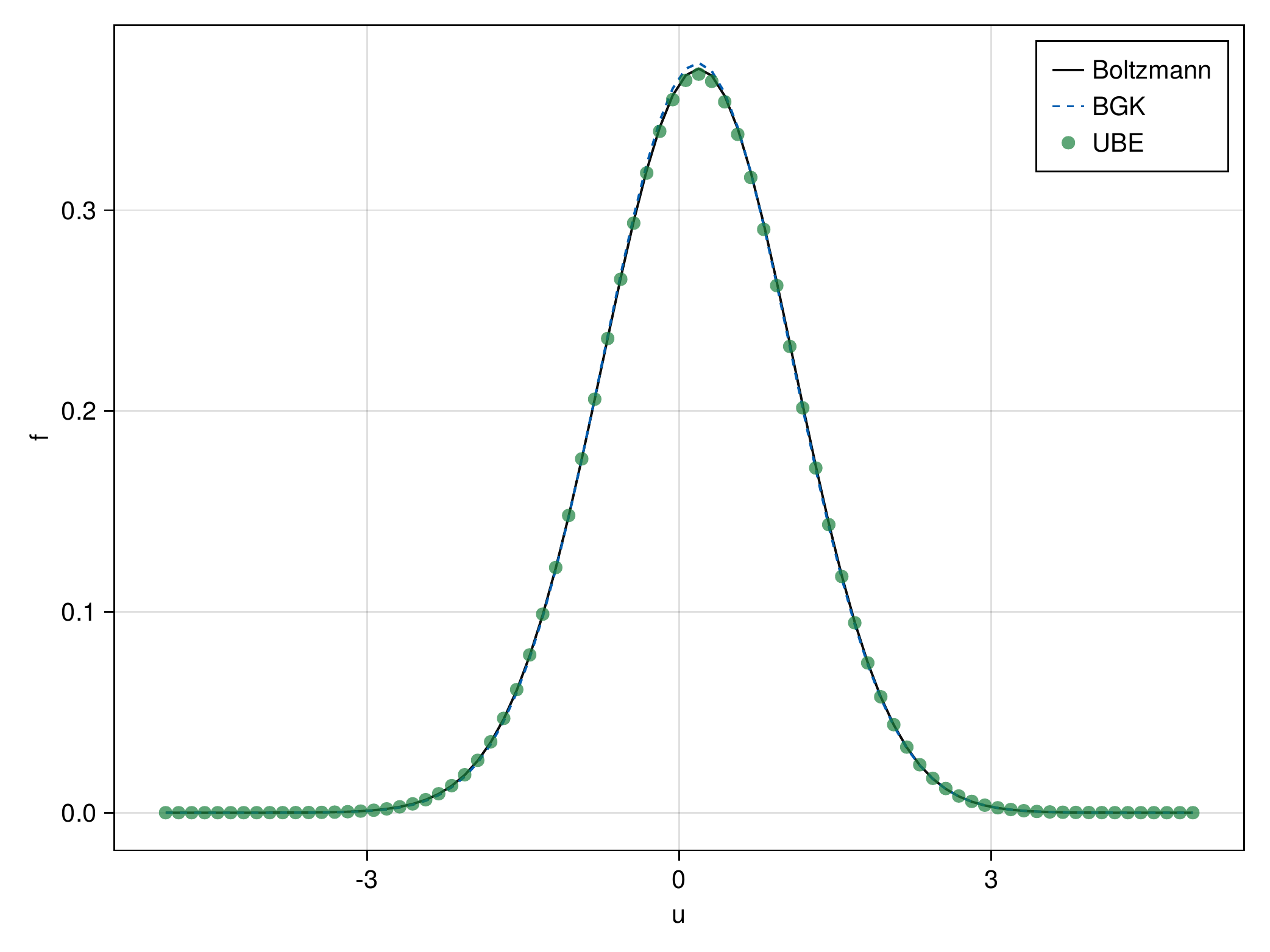}
	}
	\caption{Particle distribution functions at different time instants in the homogeneous relaxation problem (full Boltzmann model as reference solution).}
    \label{fig:boltzmann solution}
\end{figure}

\begin{figure}[htb!]
	\centering
	\subfigure[$t=0$]{
		\includegraphics[width=0.31\textwidth]{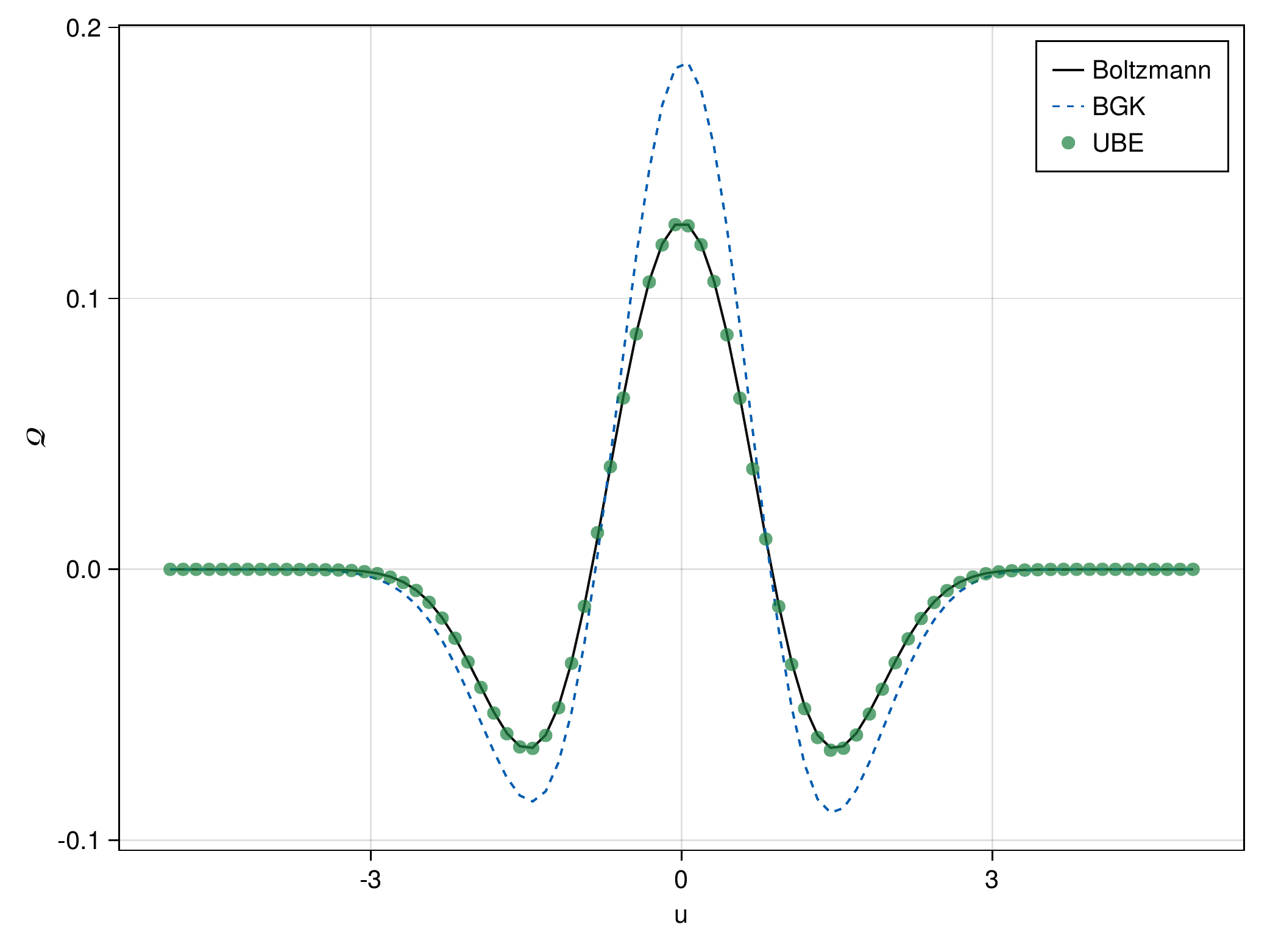}
	}
	\subfigure[$t=1$]{
		\includegraphics[width=0.31\textwidth]{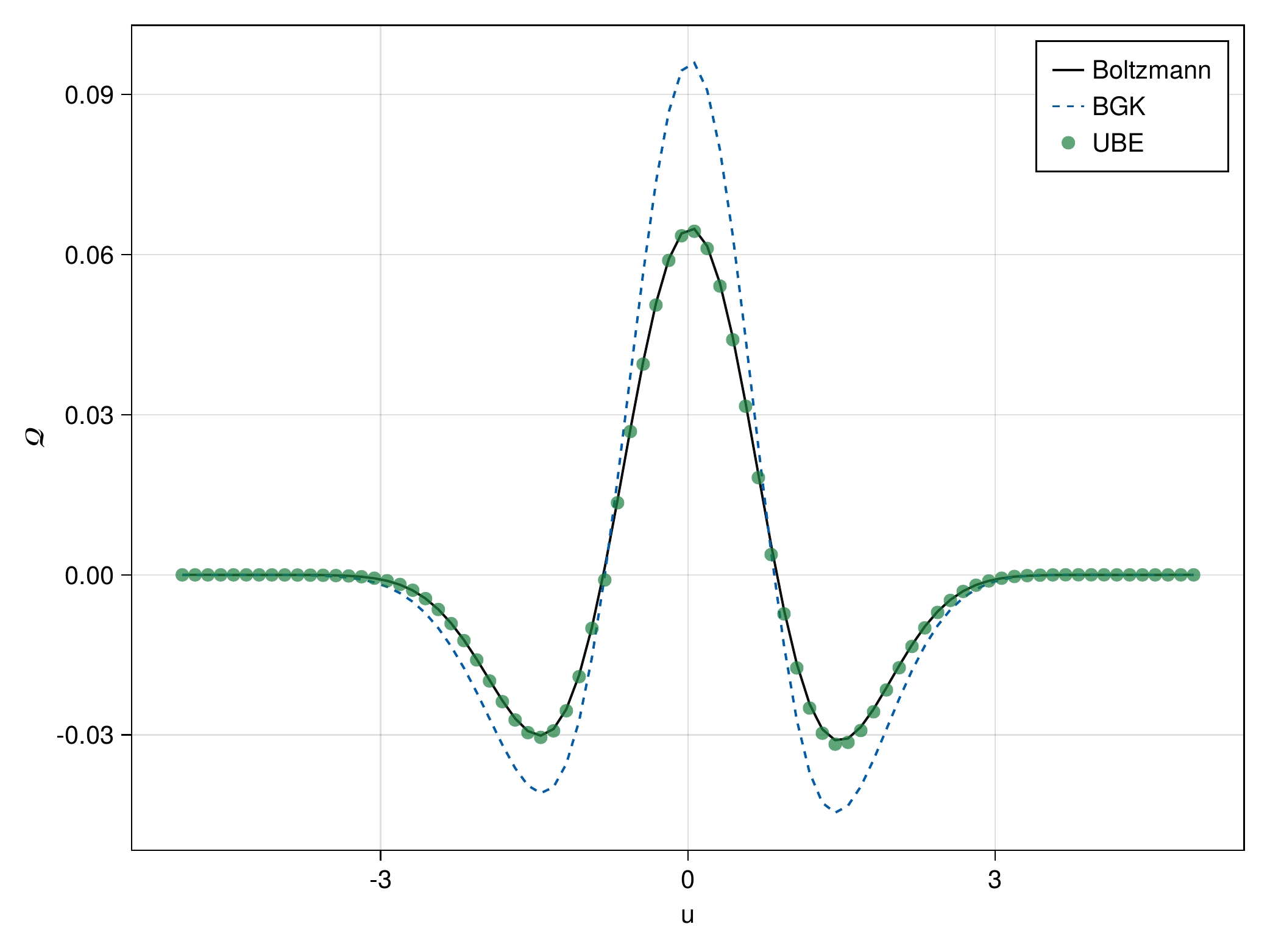}
	}
	\subfigure[$t=2$]{
		\includegraphics[width=0.31\textwidth]{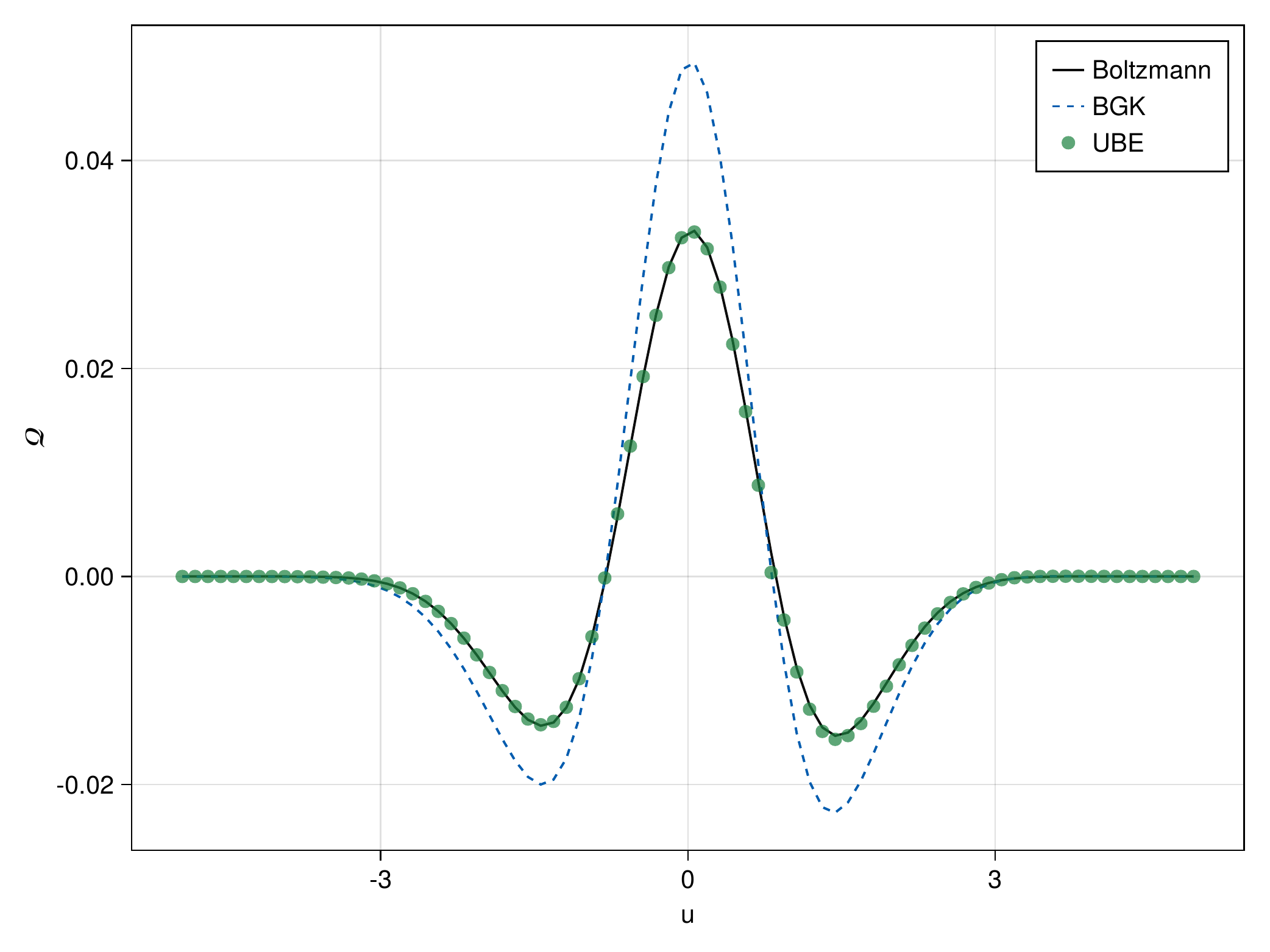}
	}
	\caption{Collision terms at different time instants in the homogeneous relaxation problem (full Boltzmann model as reference solution).}
    \label{fig:boltzmann collision}
\end{figure}

\begin{figure}[htb!]
	\centering
	\subfigure[$t=0$]{
		\includegraphics[width=0.31\textwidth]{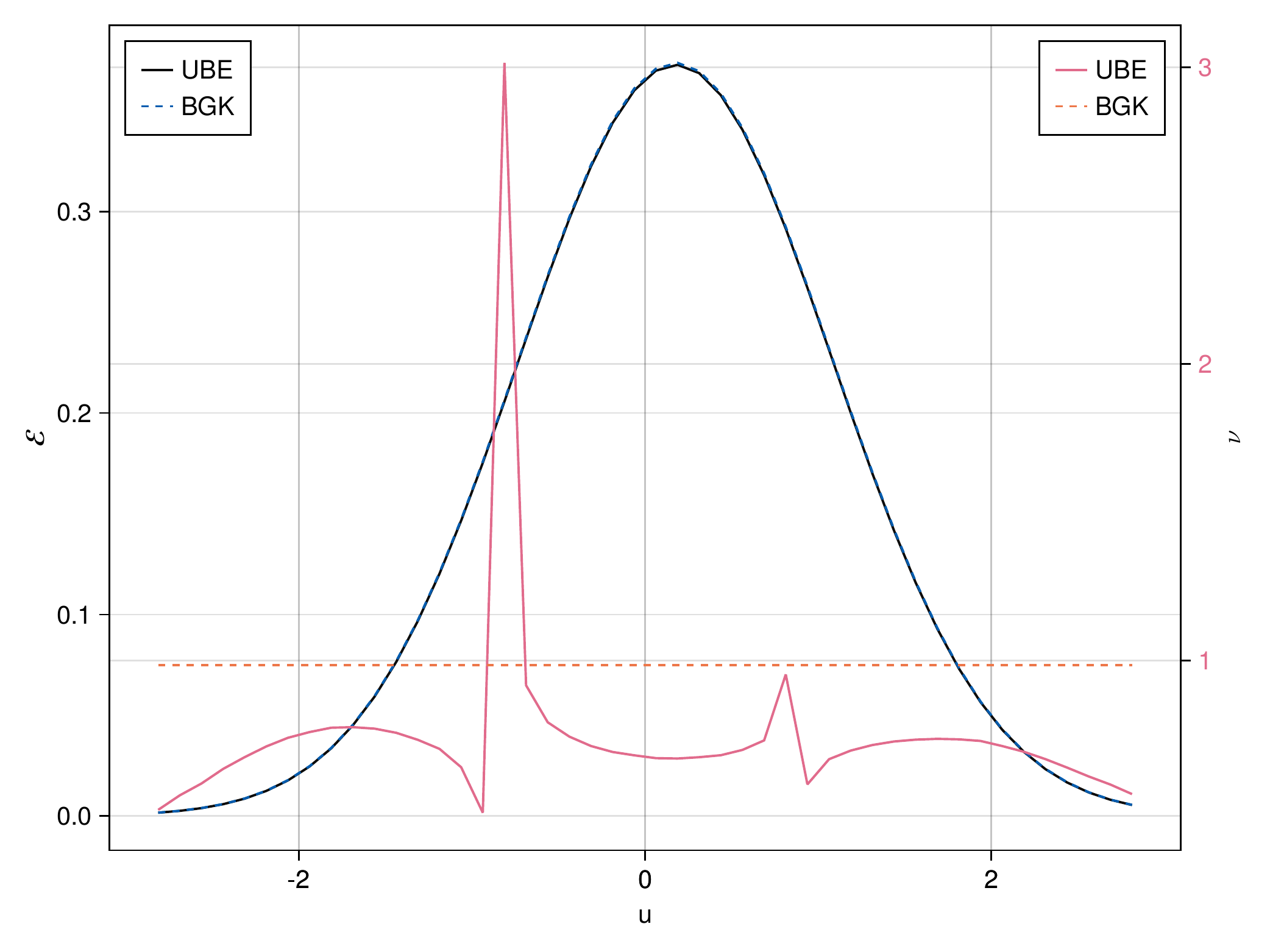}
	}
	\subfigure[$t=1$]{
		\includegraphics[width=0.31\textwidth]{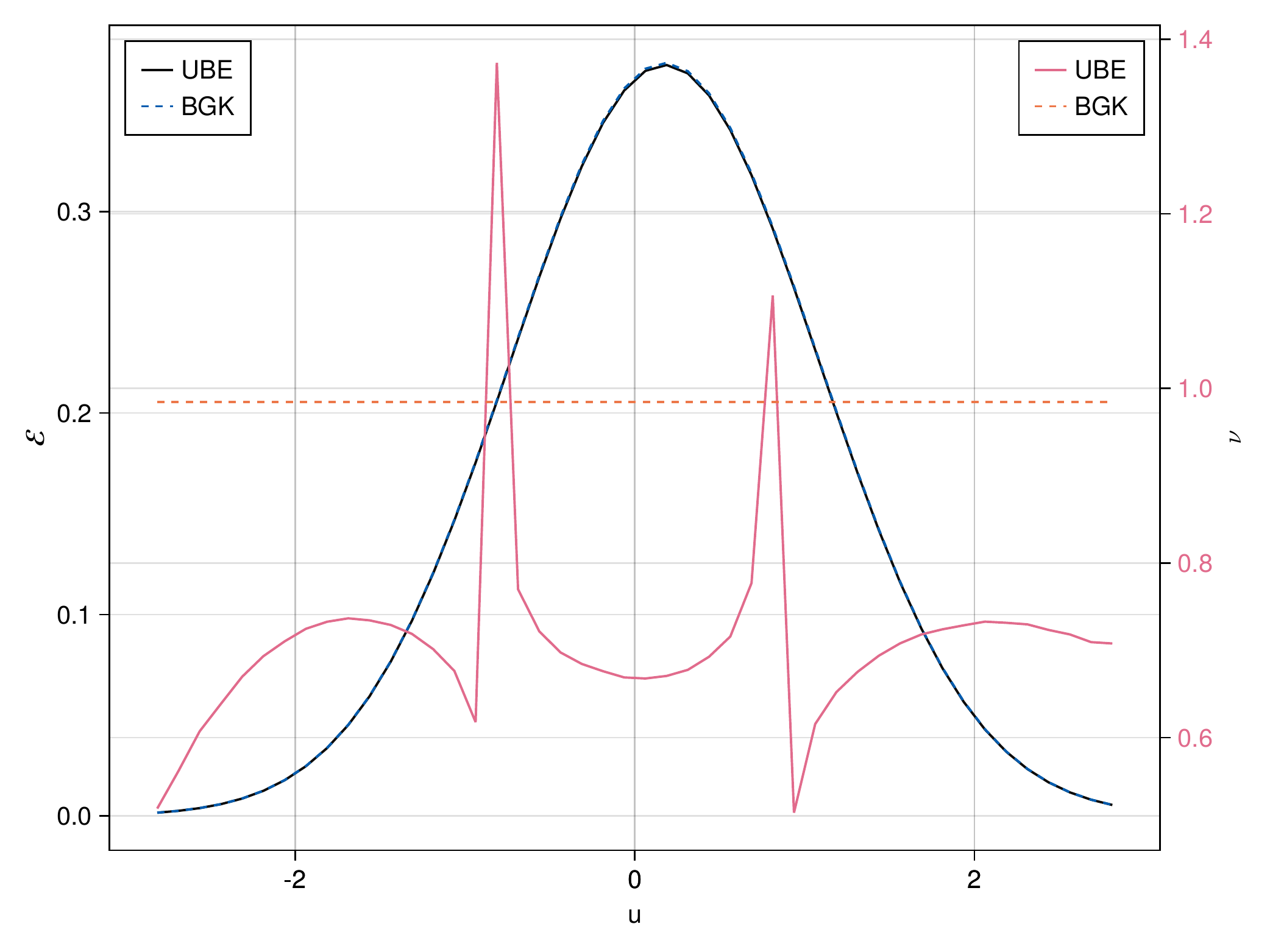}
	}
	\subfigure[$t=2$]{
		\includegraphics[width=0.31\textwidth]{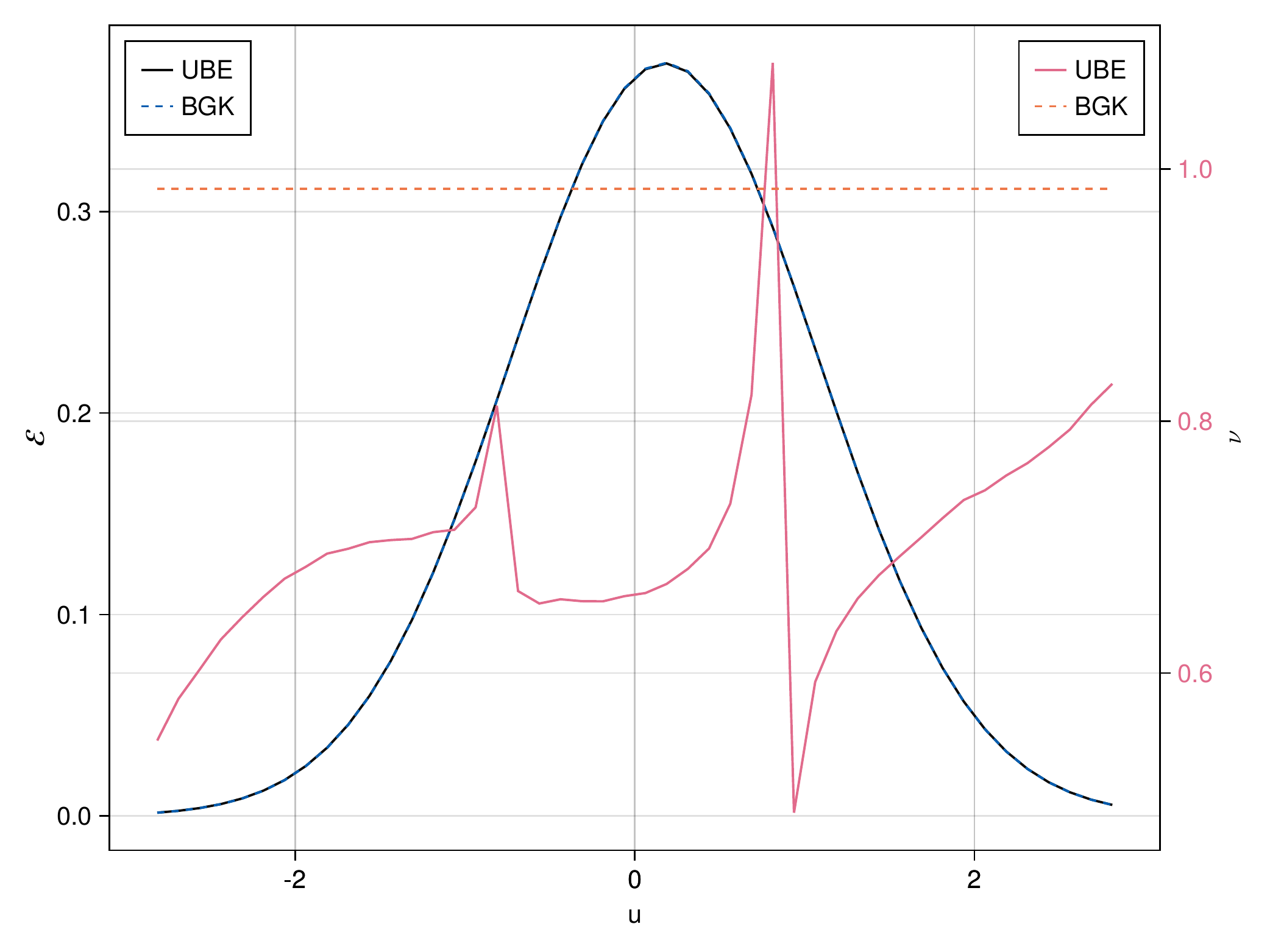}
	}
	\caption{Equilibrium (dark colors) and relaxation frequencies (light colors) at different time instants in the homogeneous relaxation problem (full Boltzmann model as reference solution).}
    \label{fig:boltzmann enu}
\end{figure}

\begin{figure}[htb!]
	\centering
	\includegraphics[width=0.4\textwidth]{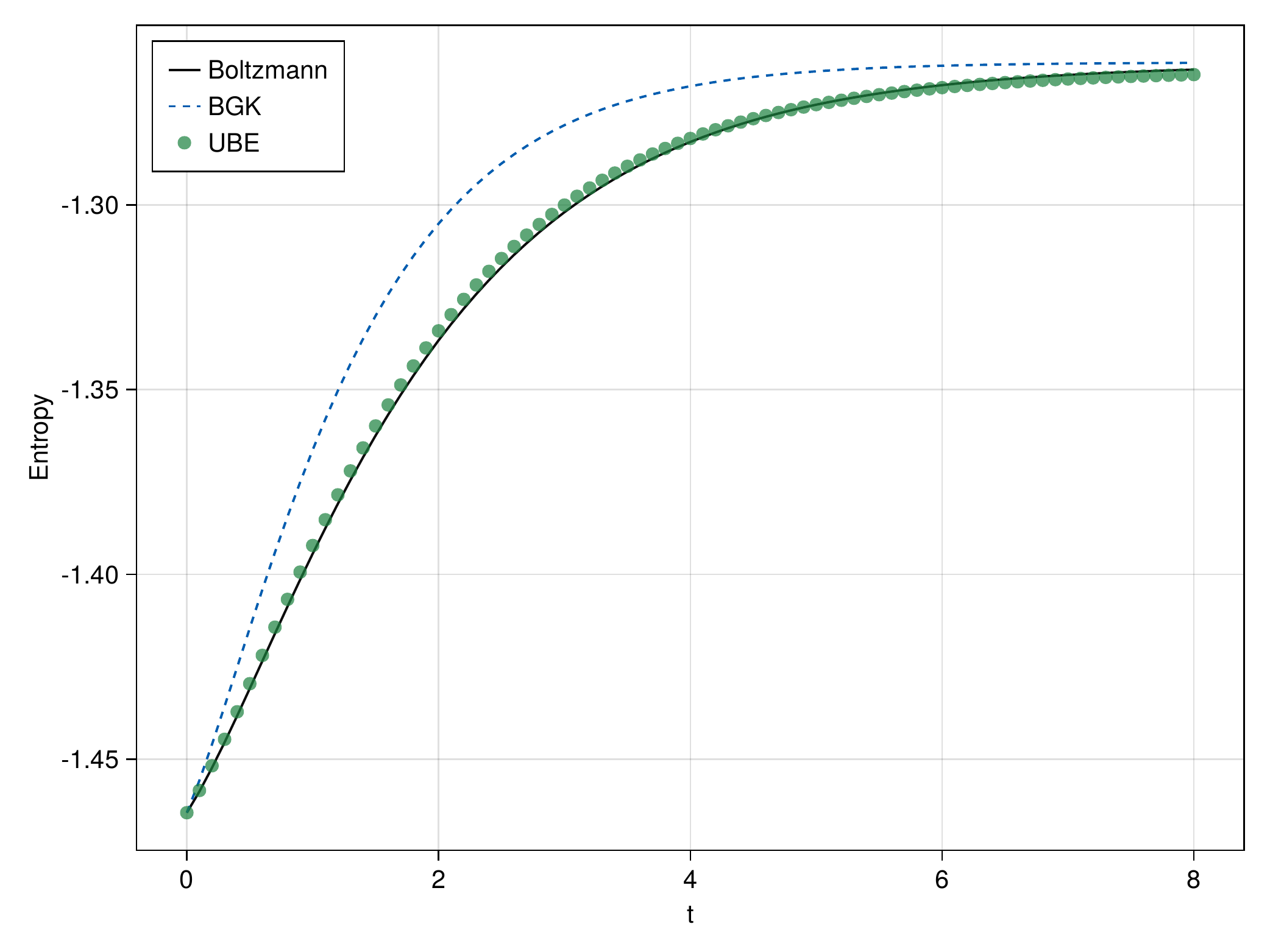}
	\caption{Evolution of entropy with time in the homogeneous relaxation problem (full Boltzmann model as reference solution).}
    \label{fig:boltzmann entropy}
\end{figure}

% shock
\begin{figure}[htb!]
	\centering
	\subfigure[Density]{
		\includegraphics[width=0.31\textwidth]{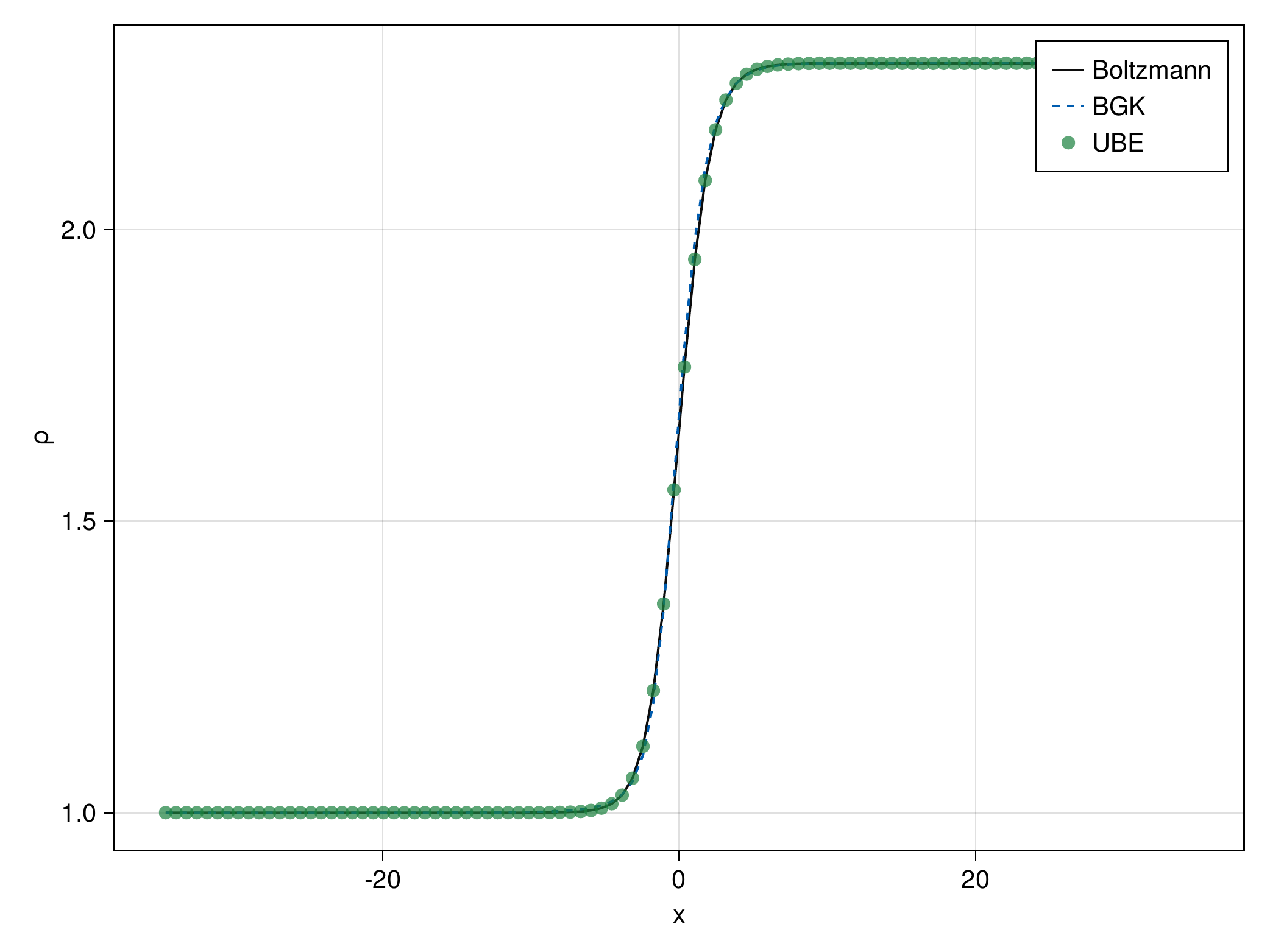}
	}
	\subfigure[Velocity]{
		\includegraphics[width=0.31\textwidth]{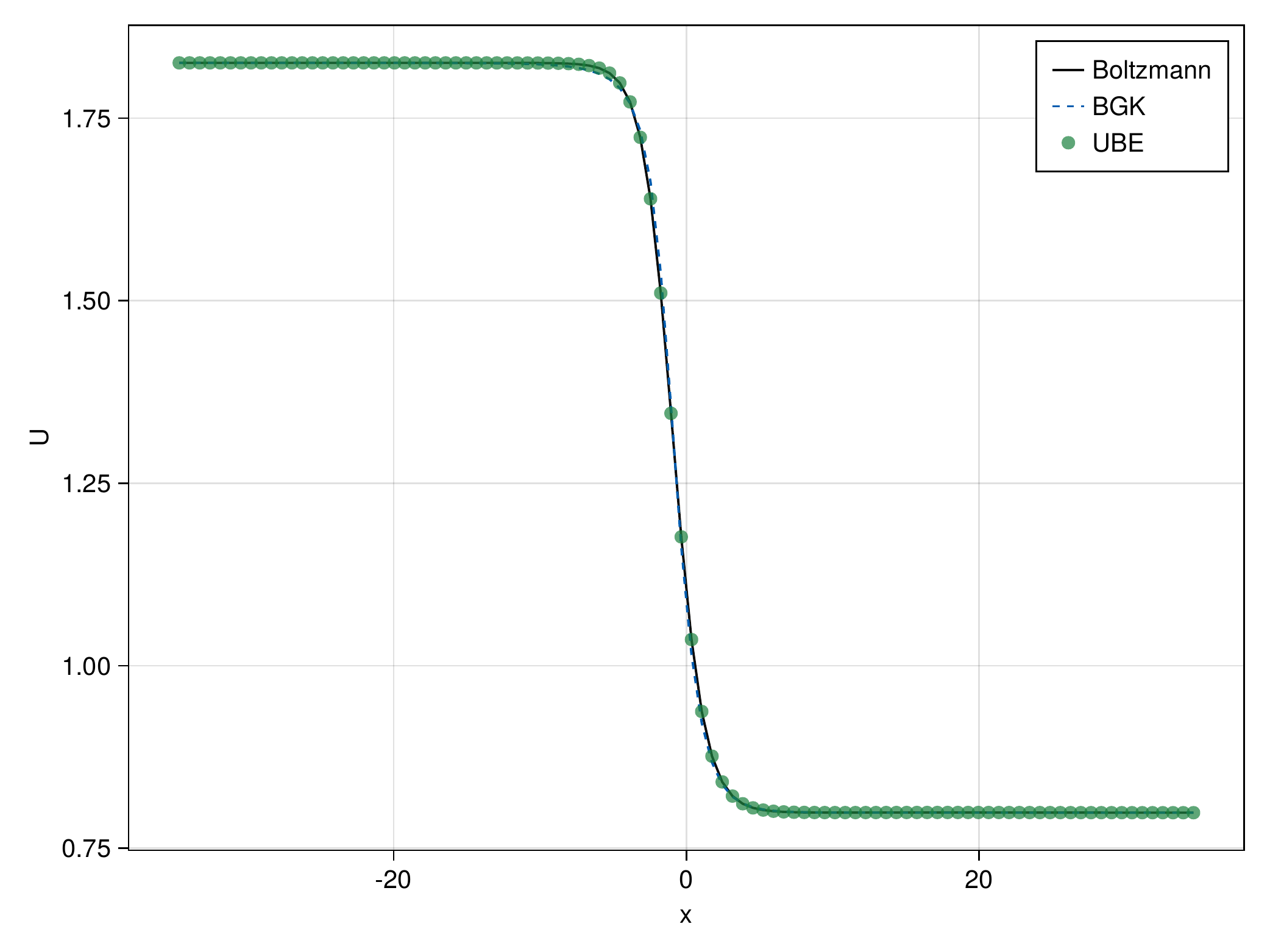}
	}
	\subfigure[Temperature]{
		\includegraphics[width=0.31\textwidth]{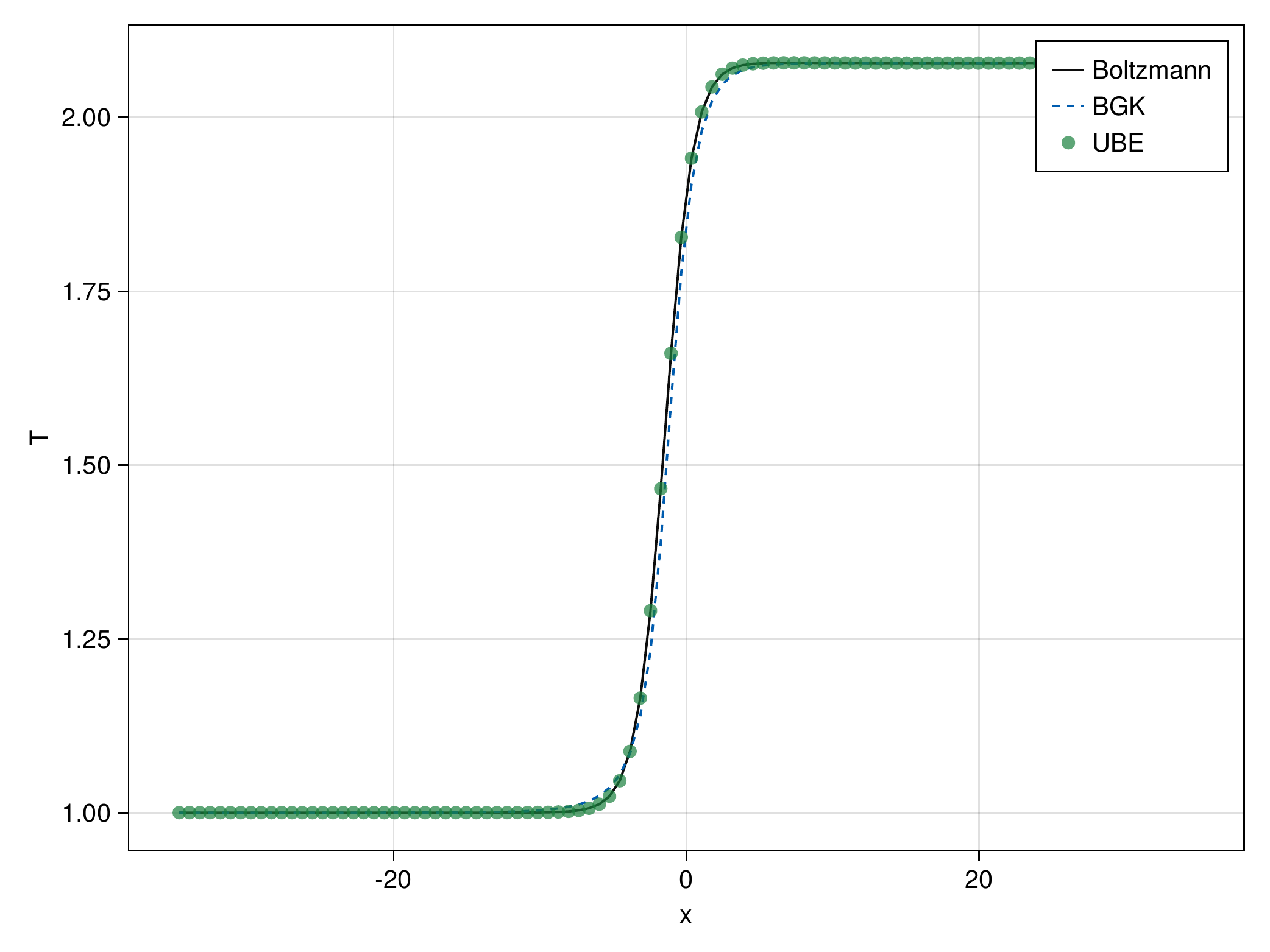}
	}
	\subfigure[Viscous stress]{
		\includegraphics[width=0.31\textwidth]{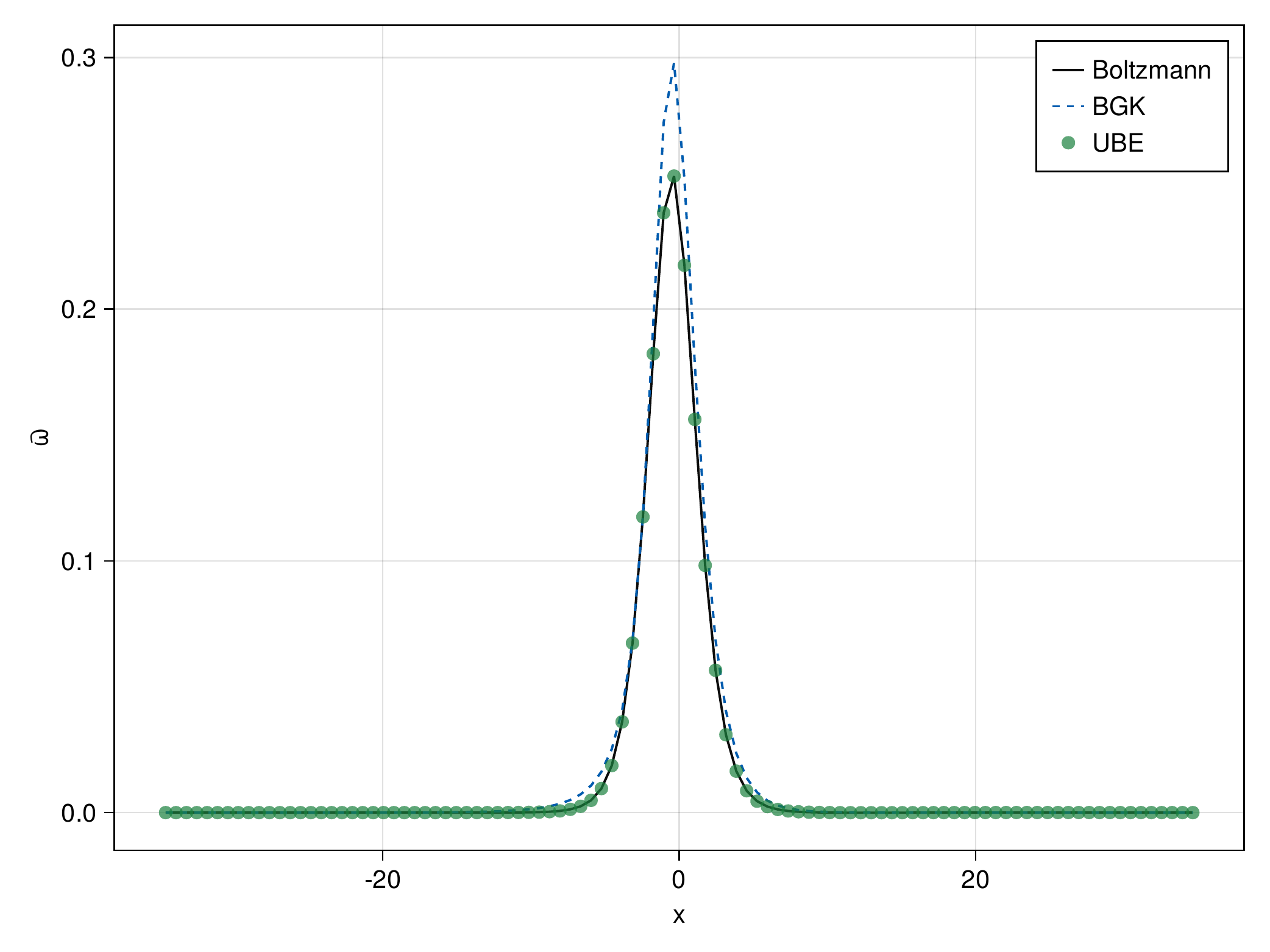}
	}
	\subfigure[Heat flux]{
		\includegraphics[width=0.31\textwidth]{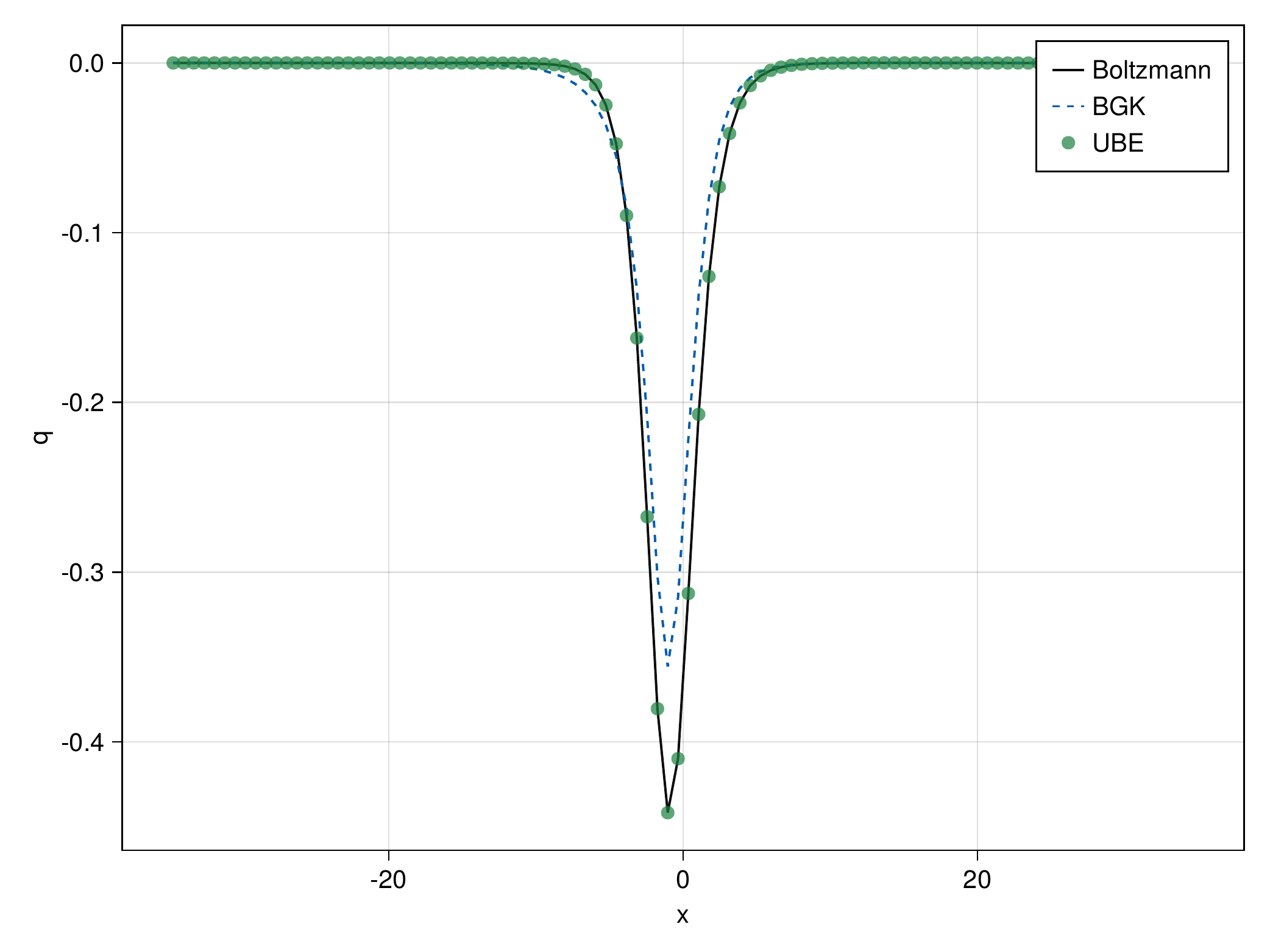}
	}
	\subfigure[Entropy]{
		\includegraphics[width=0.31\textwidth]{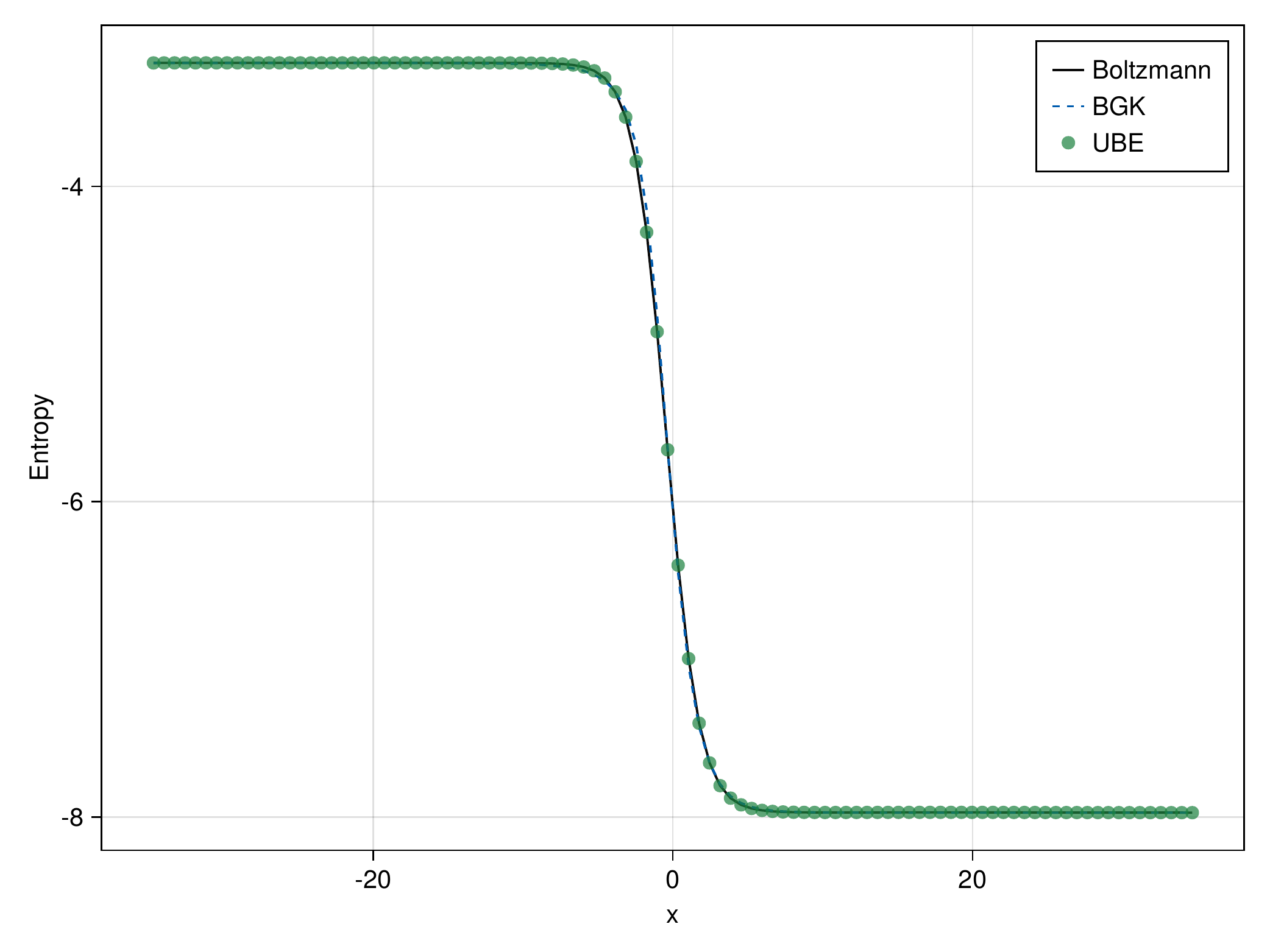}
	}
	\caption{Profiles of macroscopic variables at $\rm Ma=2$ in the normal shock structure.}
    \label{fig:shock macro ma2}
\end{figure}

\begin{figure}[htb!]
	\centering
	\subfigure[Density]{
		\includegraphics[width=0.31\textwidth]{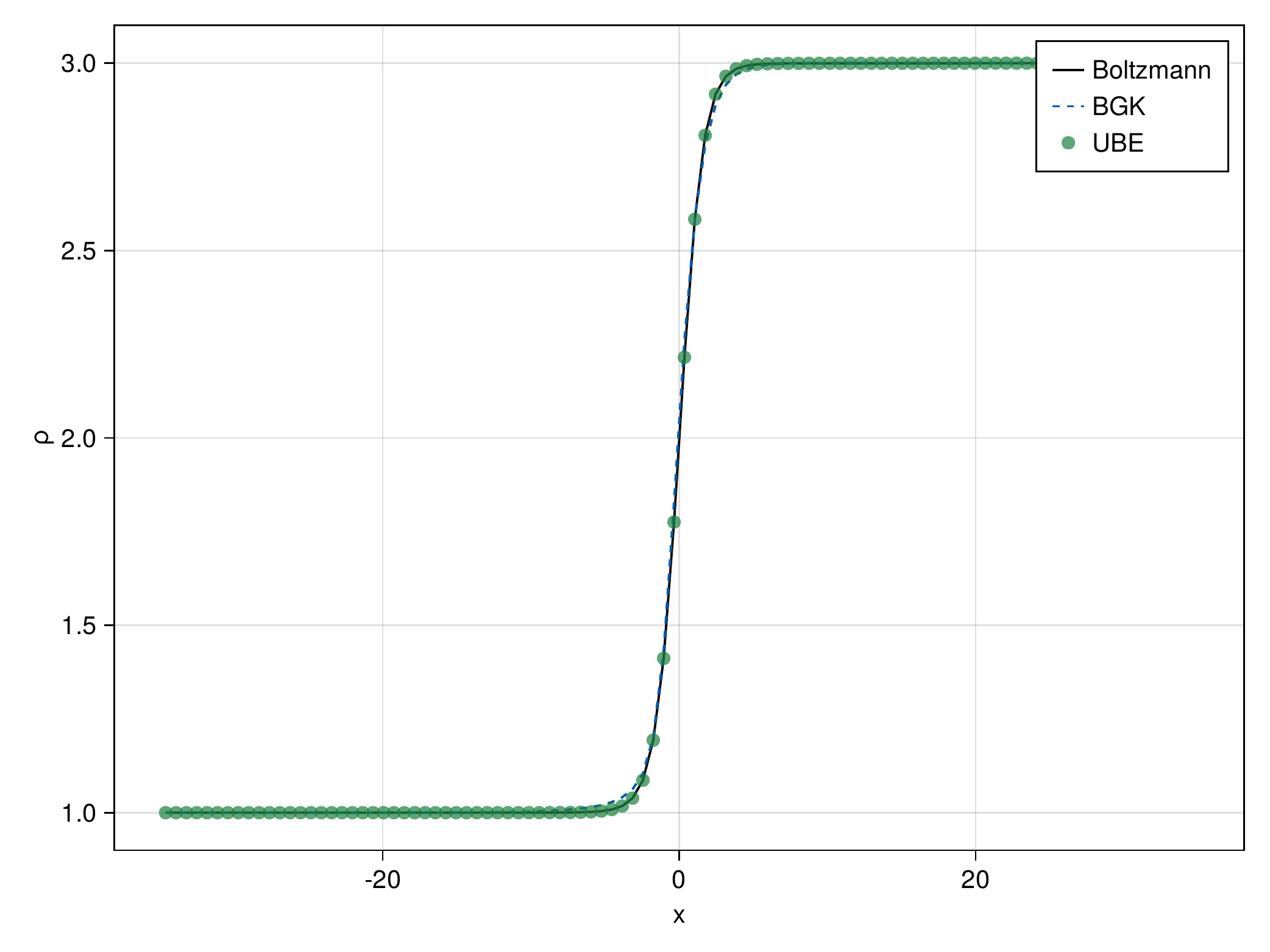}
	}
	\subfigure[Velocity]{
		\includegraphics[width=0.31\textwidth]{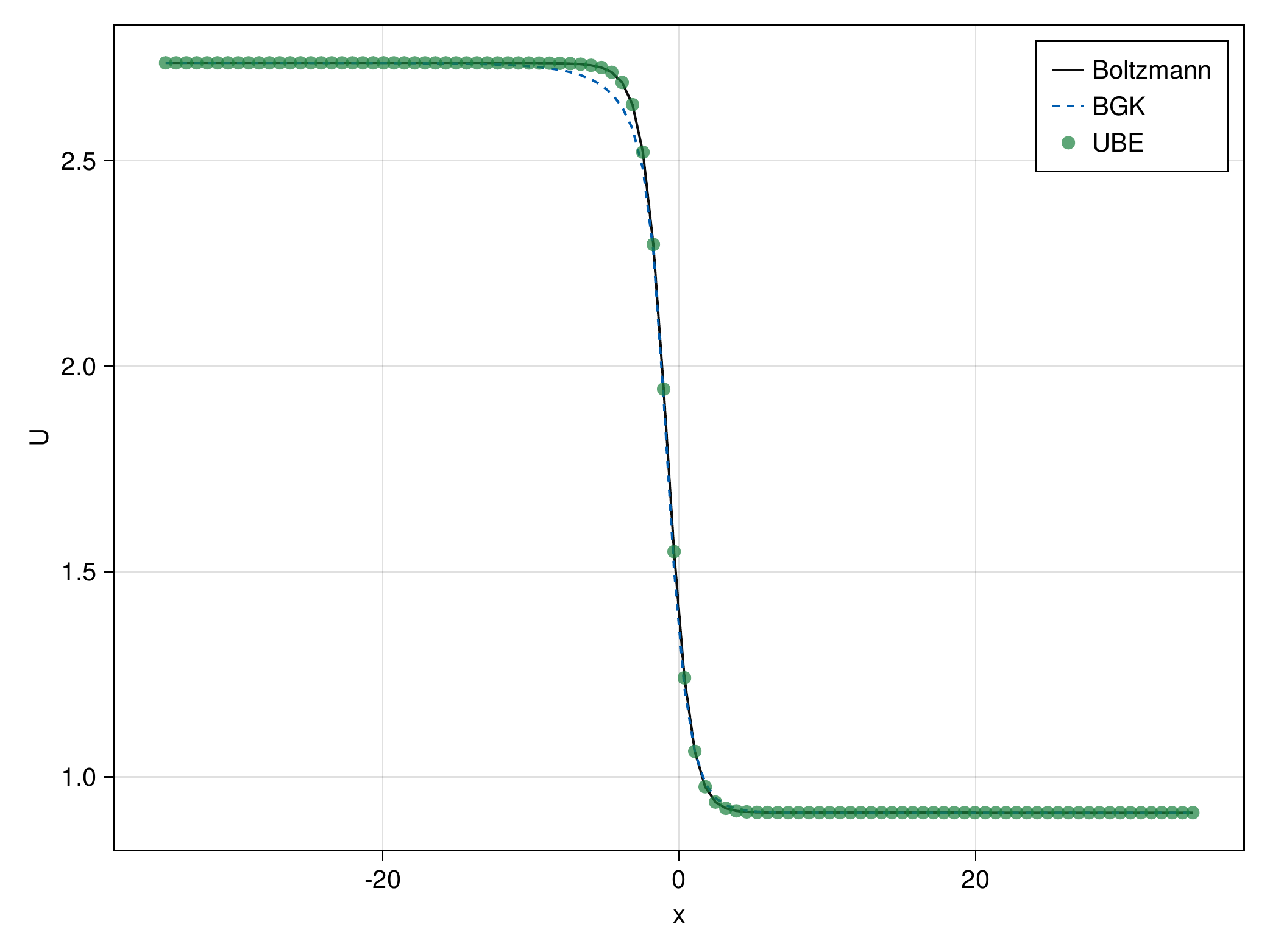}
	}
	\subfigure[Temperature]{
		\includegraphics[width=0.31\textwidth]{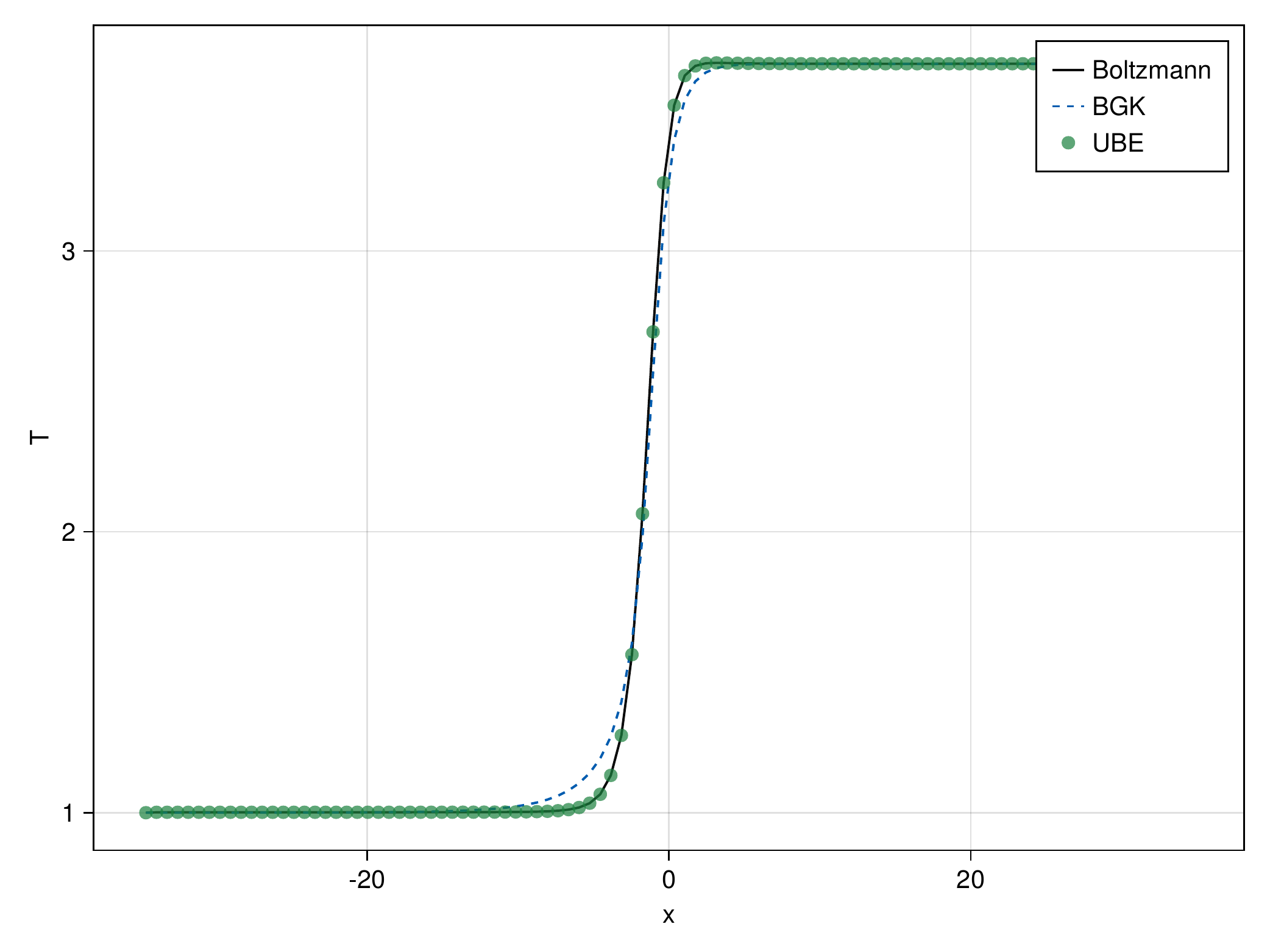}
	}
	\subfigure[Viscous stress]{
		\includegraphics[width=0.31\textwidth]{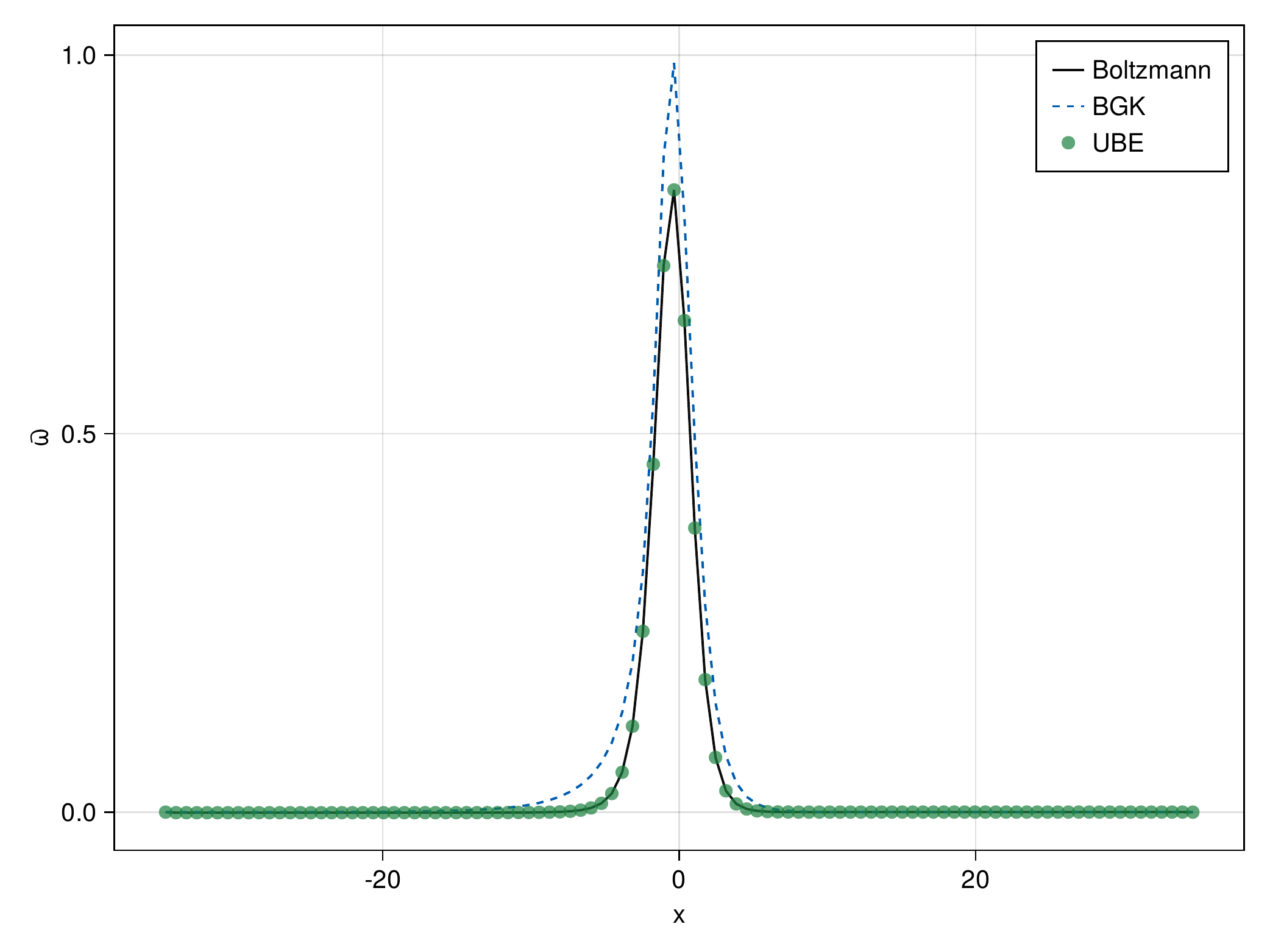}
	}
	\subfigure[Heat flux]{
		\includegraphics[width=0.31\textwidth]{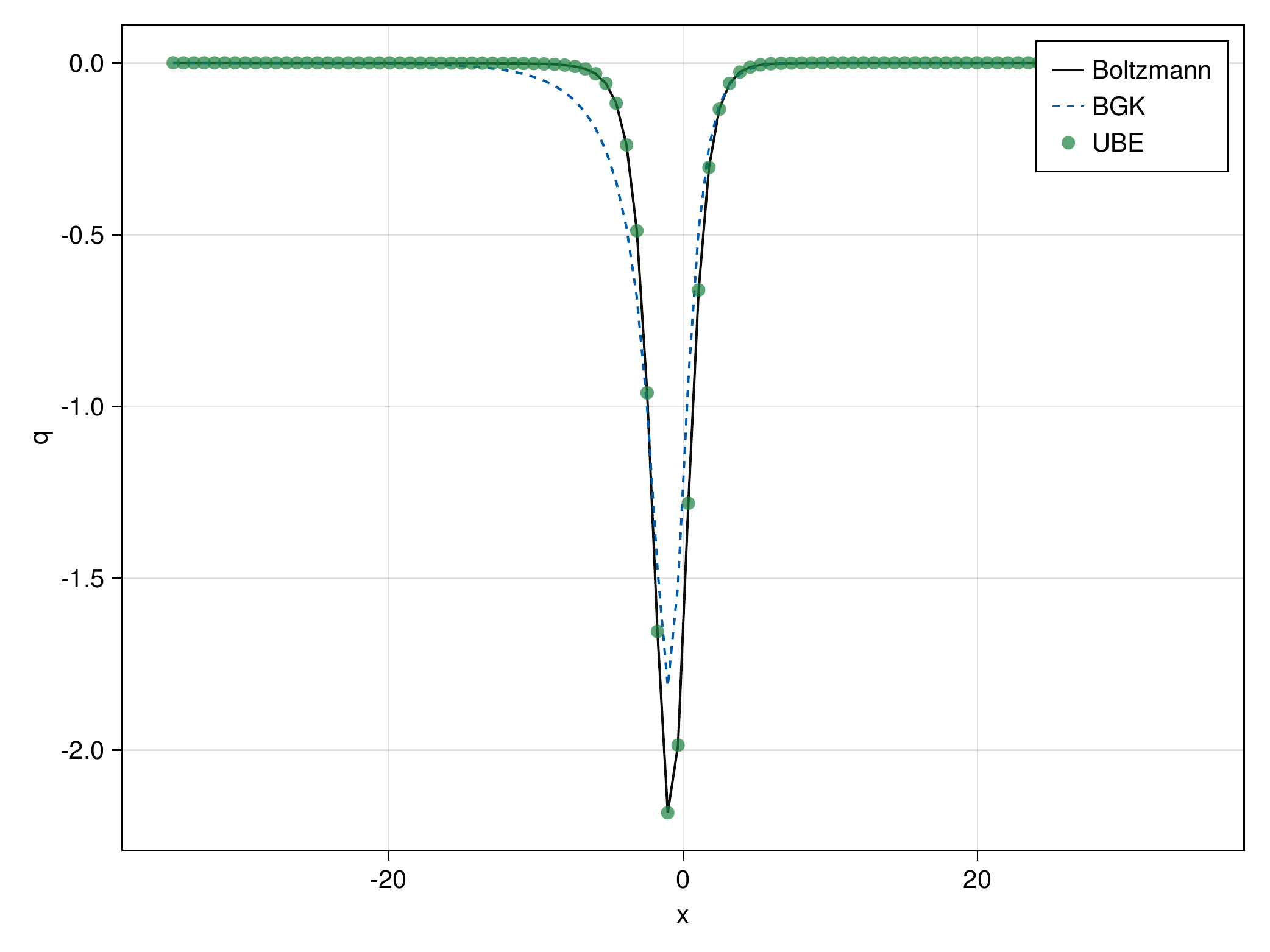}
	}
	\subfigure[Entropy]{
		\includegraphics[width=0.31\textwidth]{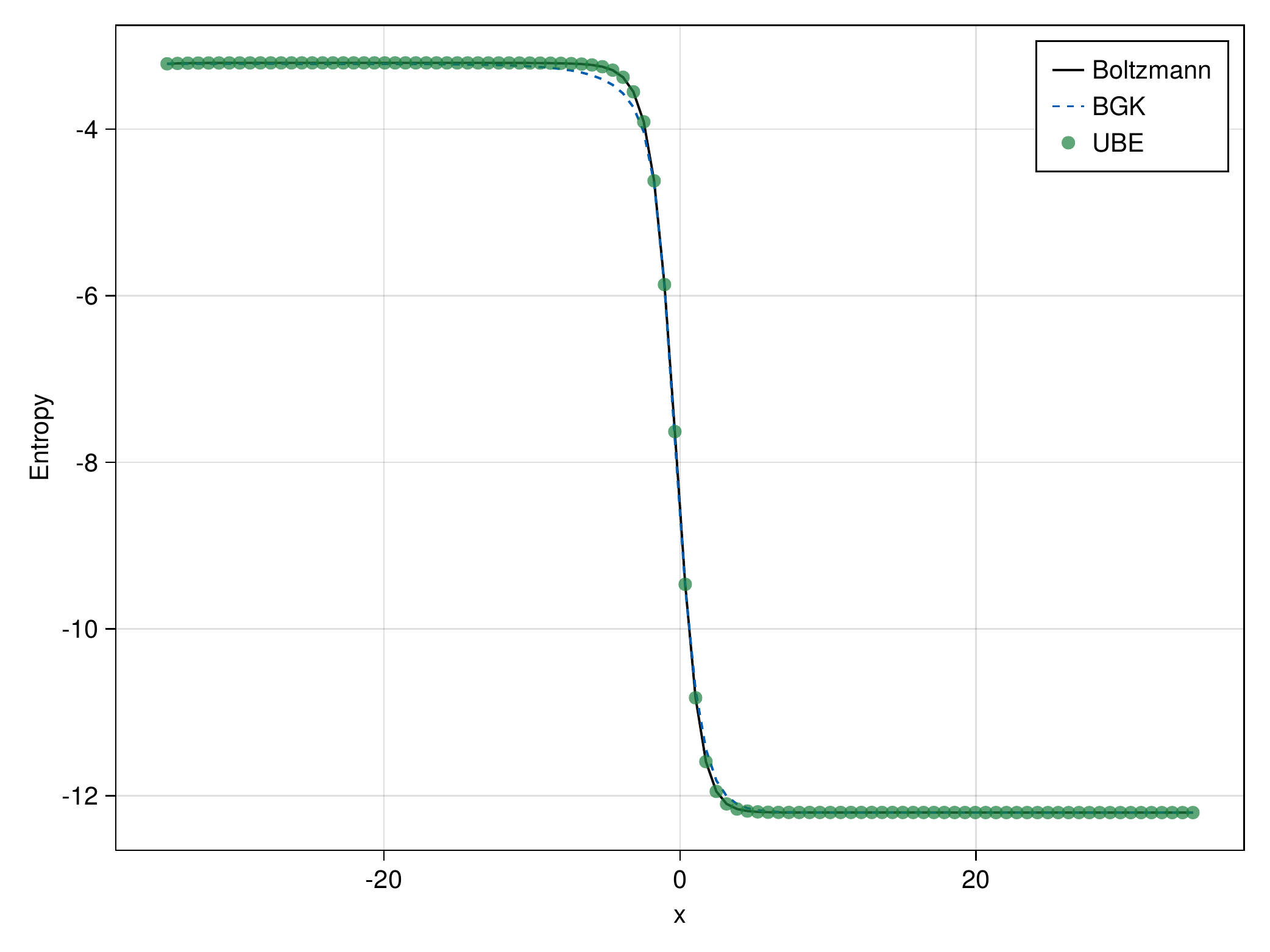}
	}
	\caption{Profiles of macroscopic variables at $\rm Ma=3$ in the normal shock structure.}
    \label{fig:shock macro ma3}
\end{figure}

\begin{figure}[htb!]
	\centering
	\subfigure[UBE]{
		\includegraphics[width=0.35\textwidth]{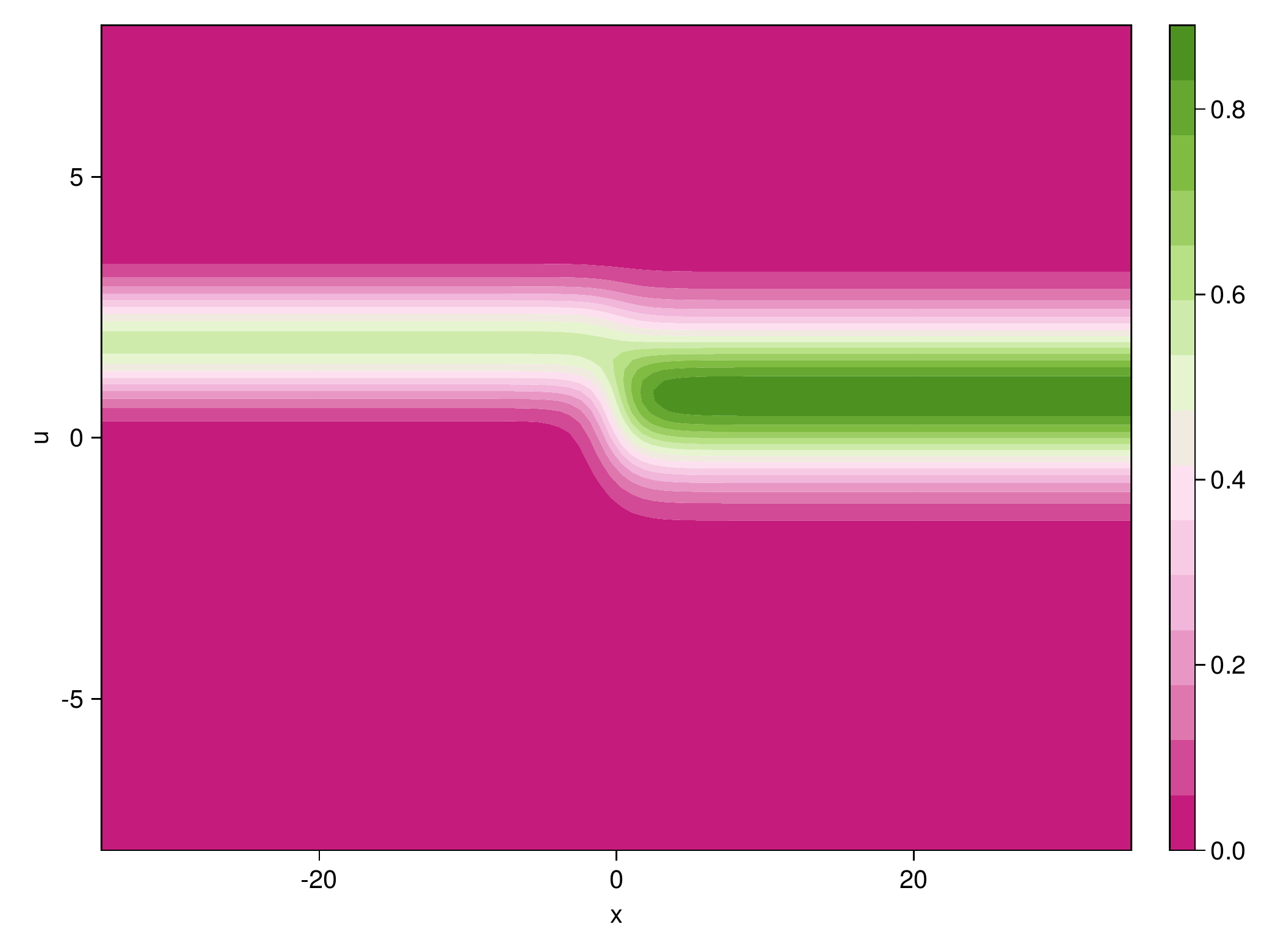}
	}
	\subfigure[BGK]{
		\includegraphics[width=0.35\textwidth]{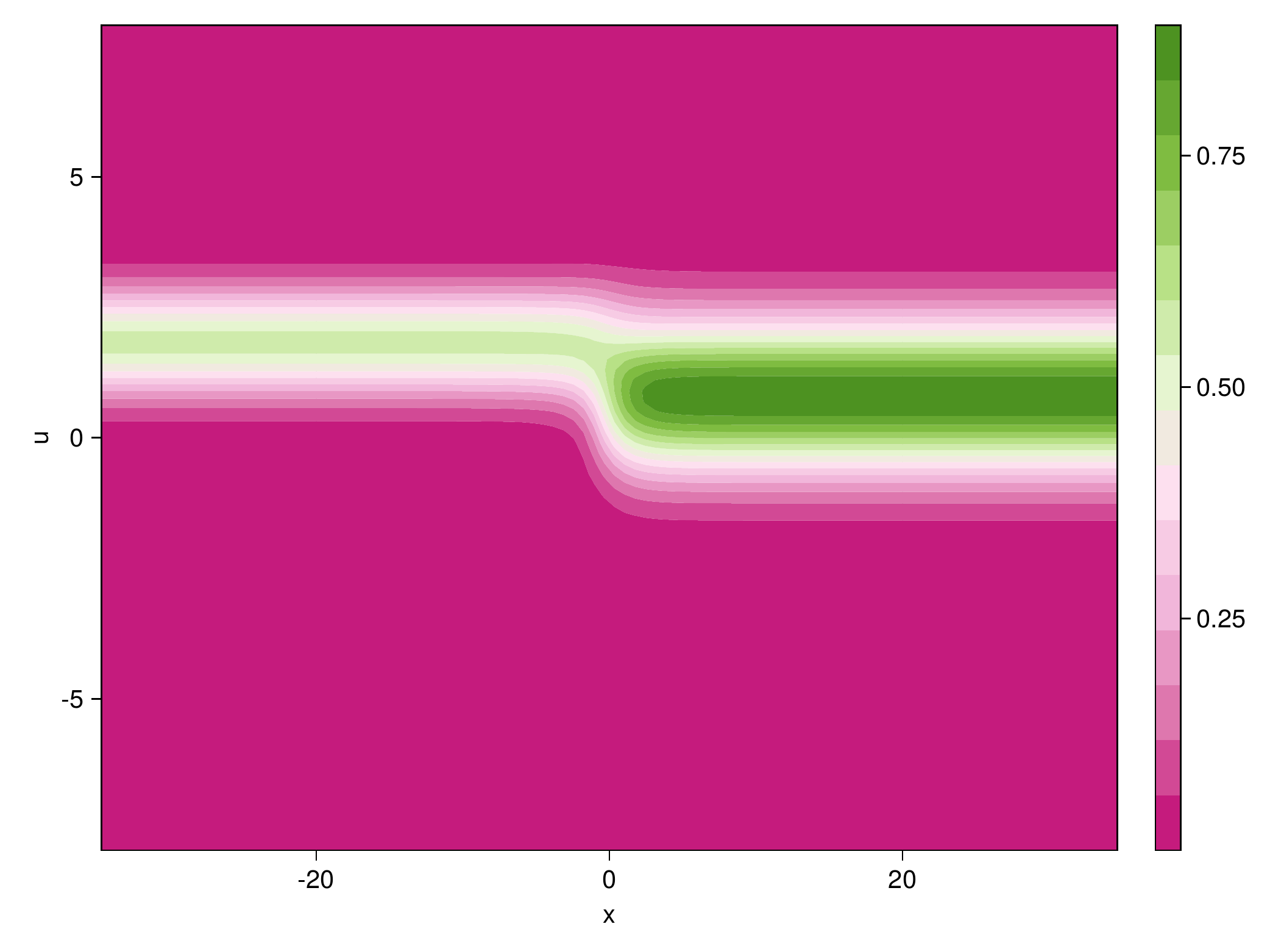}
	}
	\subfigure[UBE]{
		\includegraphics[width=0.35\textwidth]{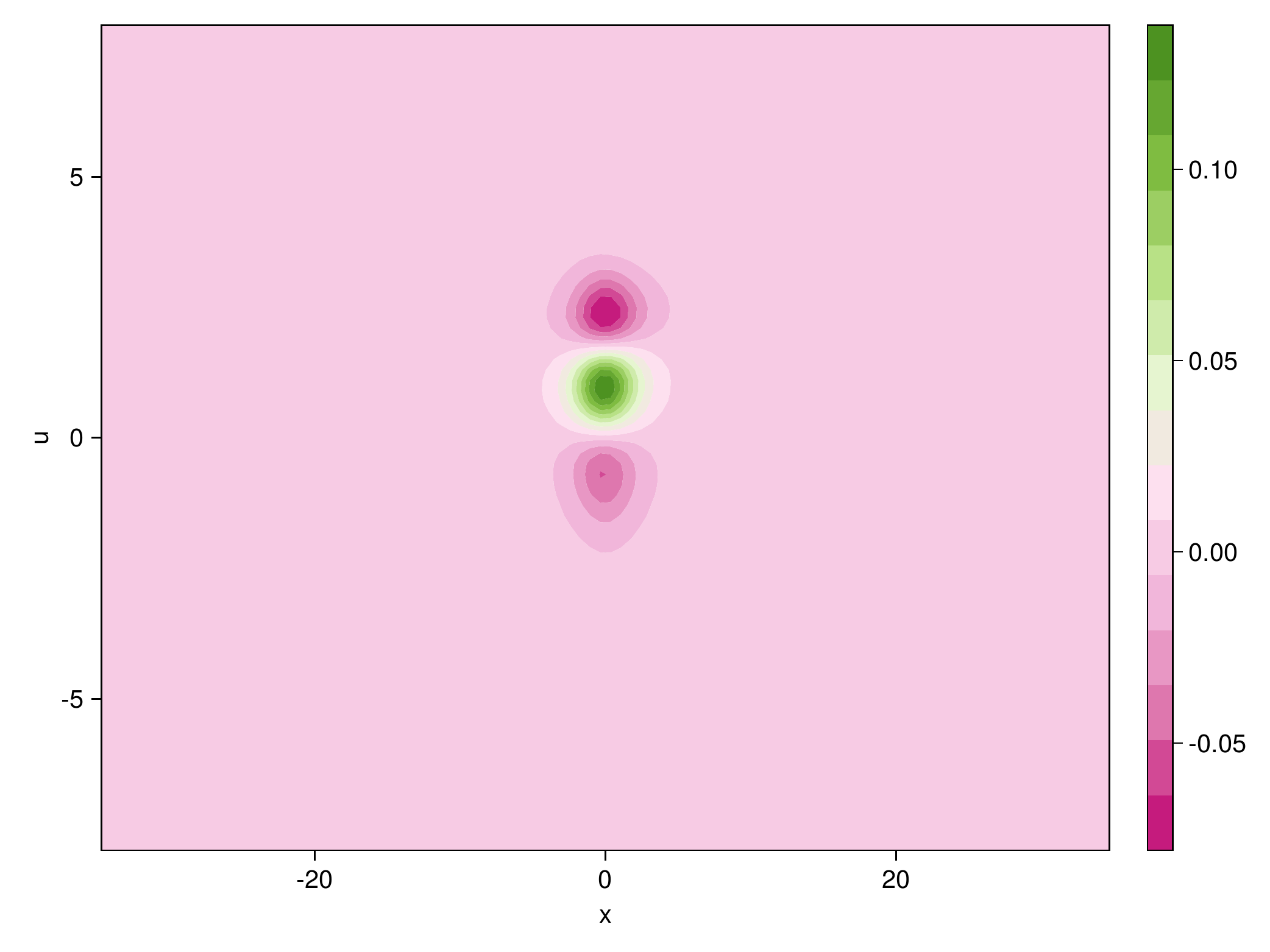}
	}
	\subfigure[BGK]{
		\includegraphics[width=0.35\textwidth]{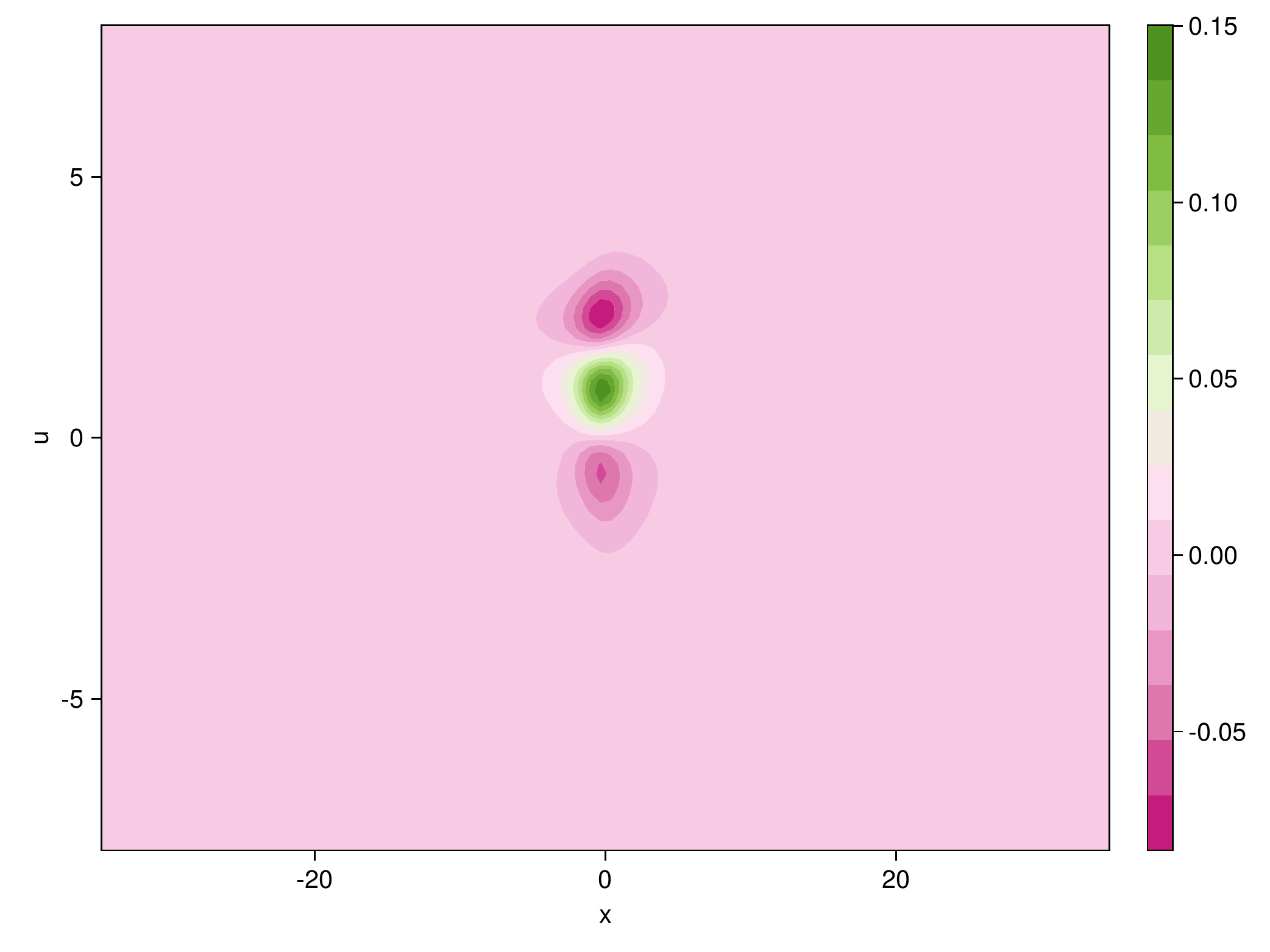}
	}
	\caption{Contours of particle distribution function (first row) and collision term (second row) at $\rm Ma=2$ in the normal shock structure.}
    \label{fig:shock pdf ma2}
\end{figure}

\begin{figure}[htb!]
	\centering
	\subfigure[UBE]{
		\includegraphics[width=0.35\textwidth]{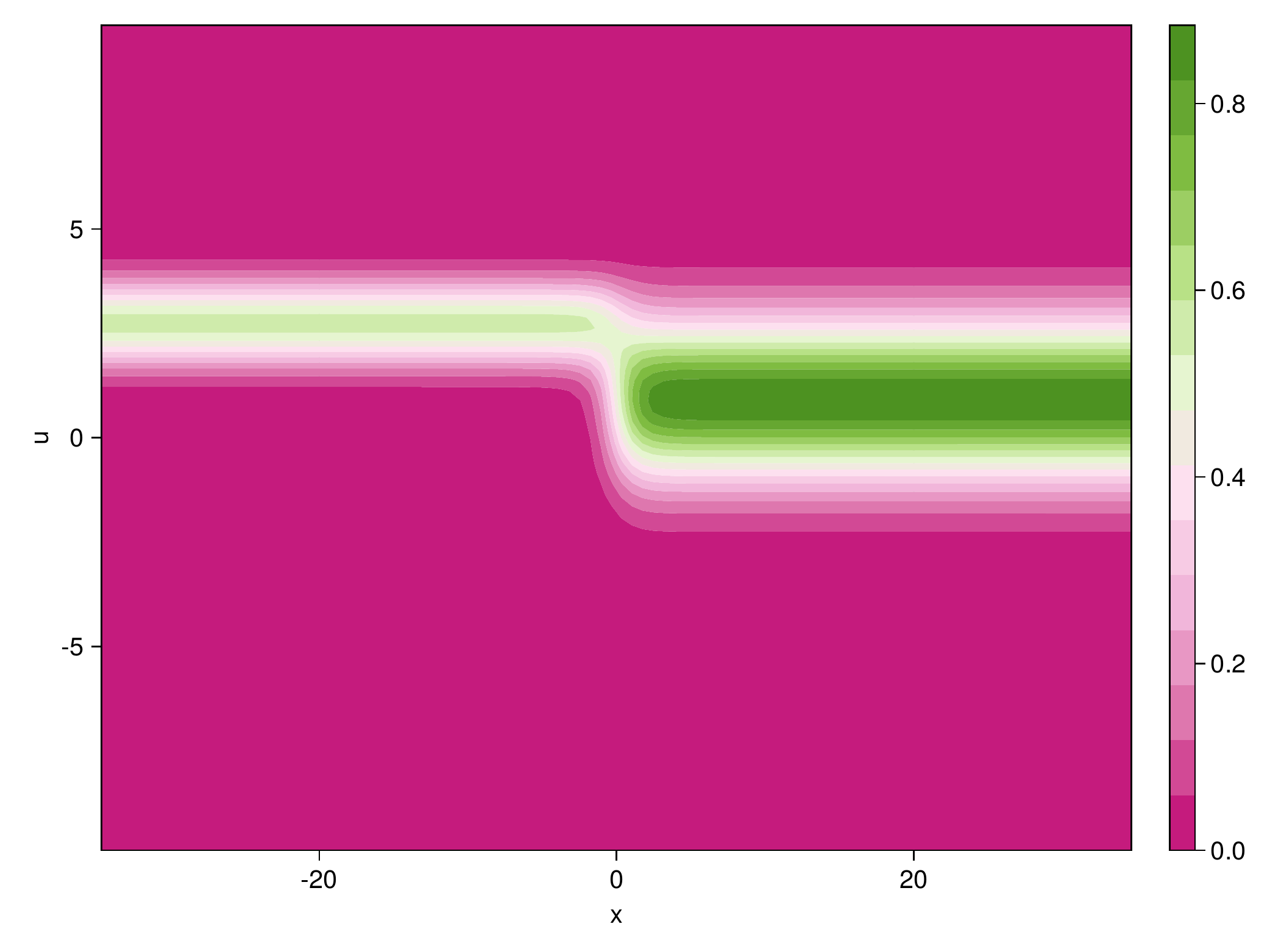}
	}
	\subfigure[BGK]{
		\includegraphics[width=0.35\textwidth]{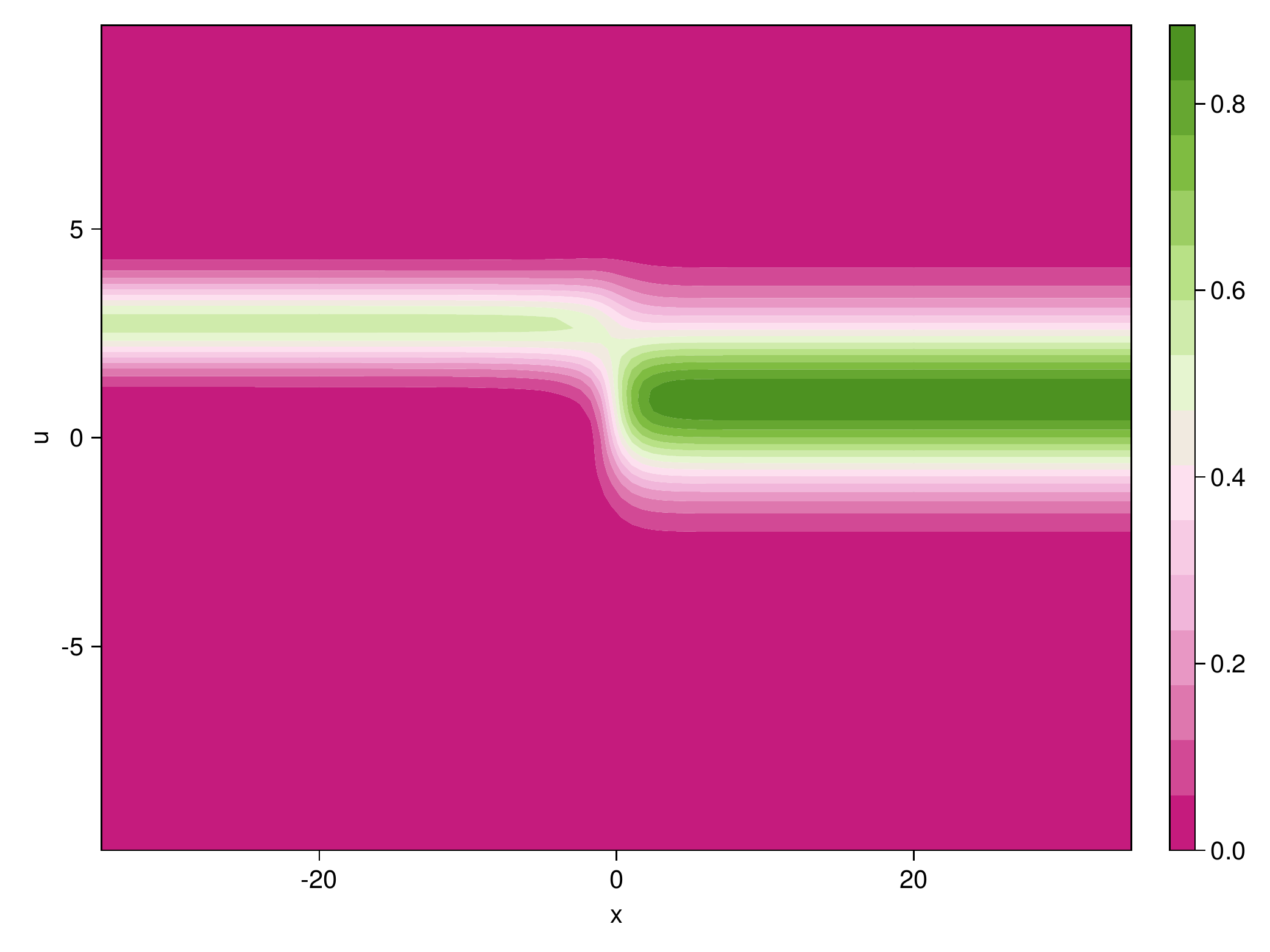}
	}
	\subfigure[UBE]{
		\includegraphics[width=0.35\textwidth]{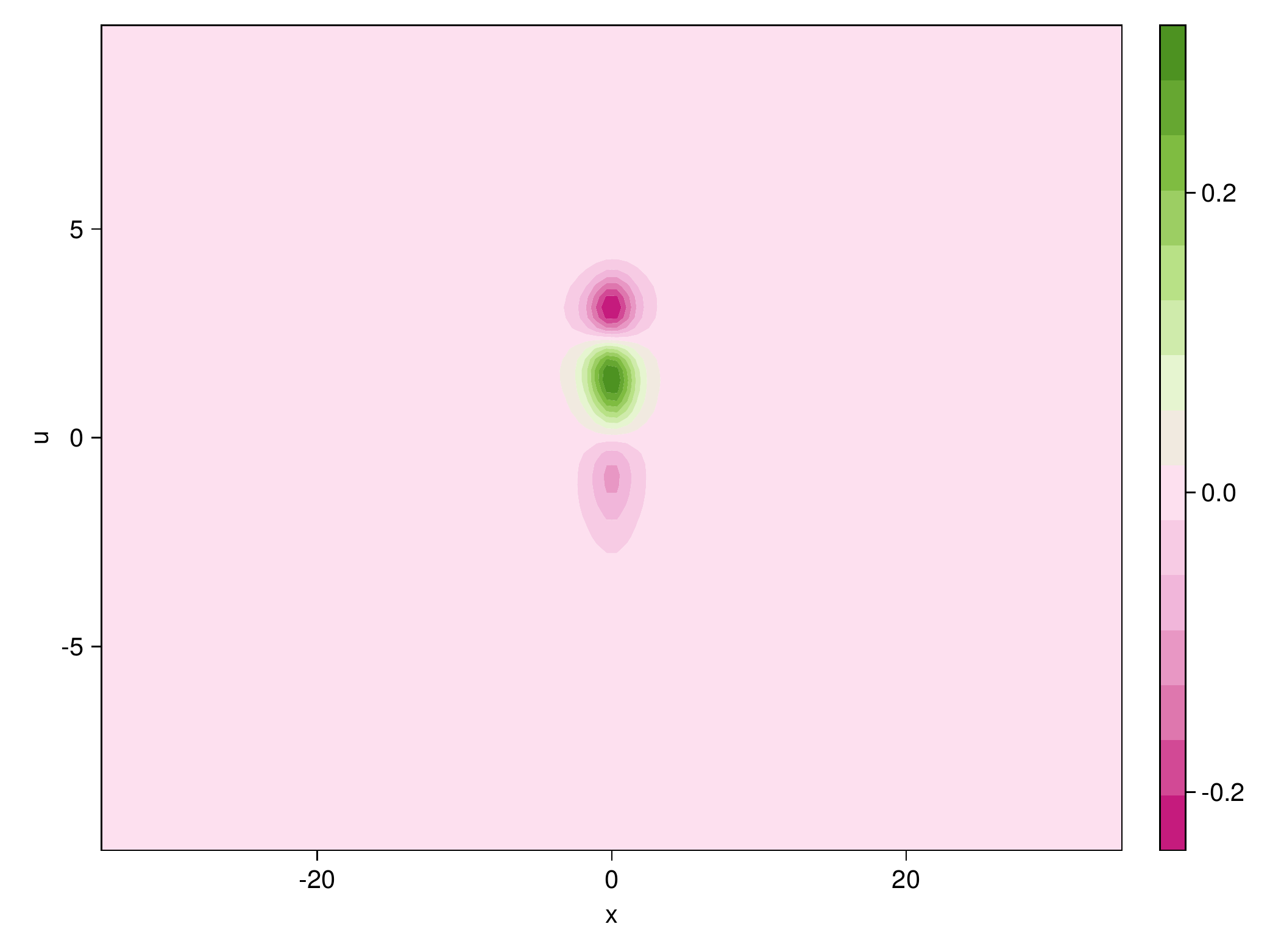}
	}
	\subfigure[BGK]{
		\includegraphics[width=0.35\textwidth]{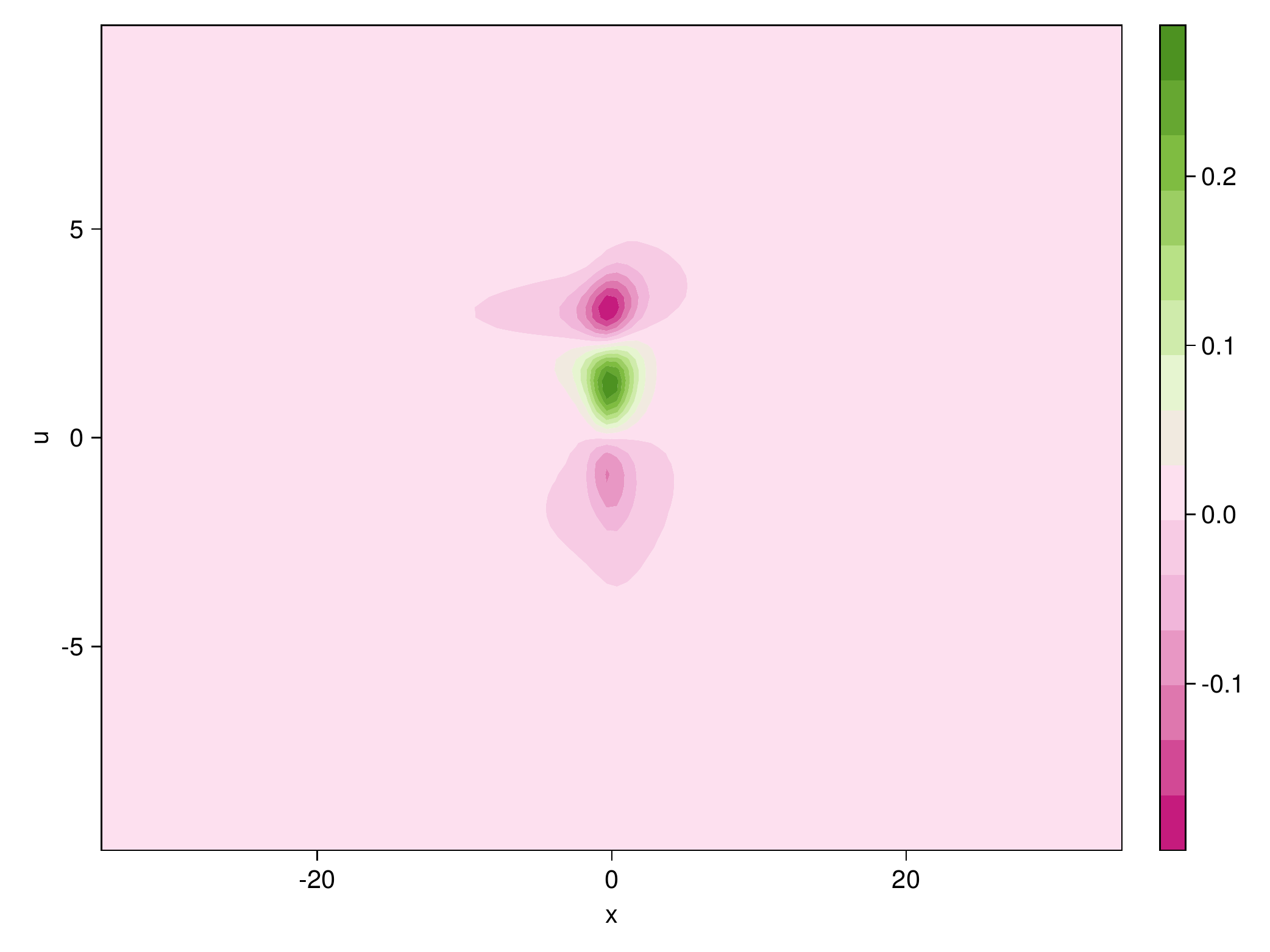}
	}
	\caption{Contours of particle distribution function (first row) and collision term (second row) at $\rm Ma=3$ in the normal shock structure.}
    \label{fig:shock pdf ma3}
\end{figure}

% cavity
\begin{figure}[htb!]
	\centering
	\subfigure[Density with streamlines]{
		\includegraphics[width=0.4\textwidth]{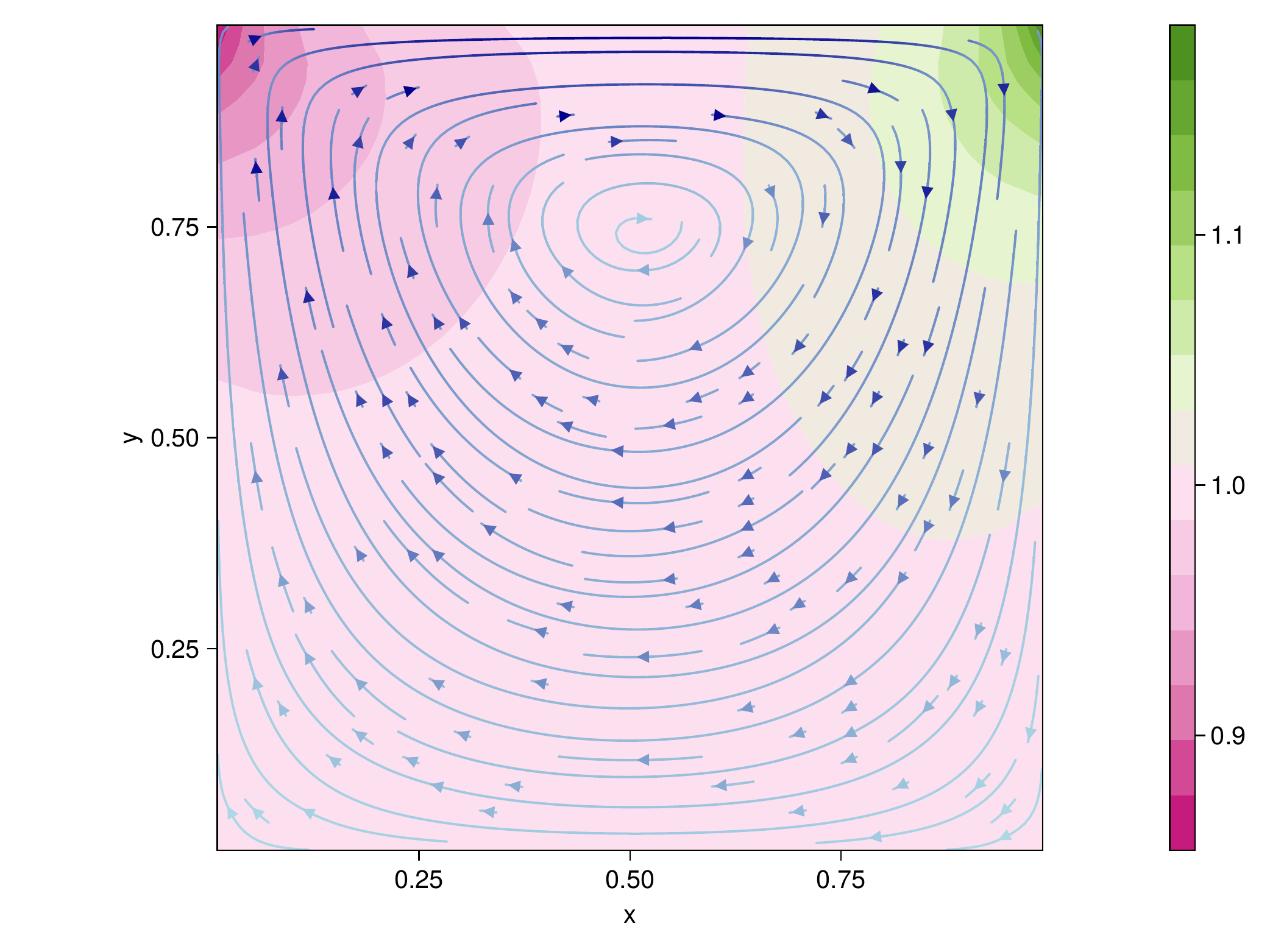}
	}
	\subfigure[Temperature with heat flux vectors]{
		\includegraphics[width=0.4\textwidth]{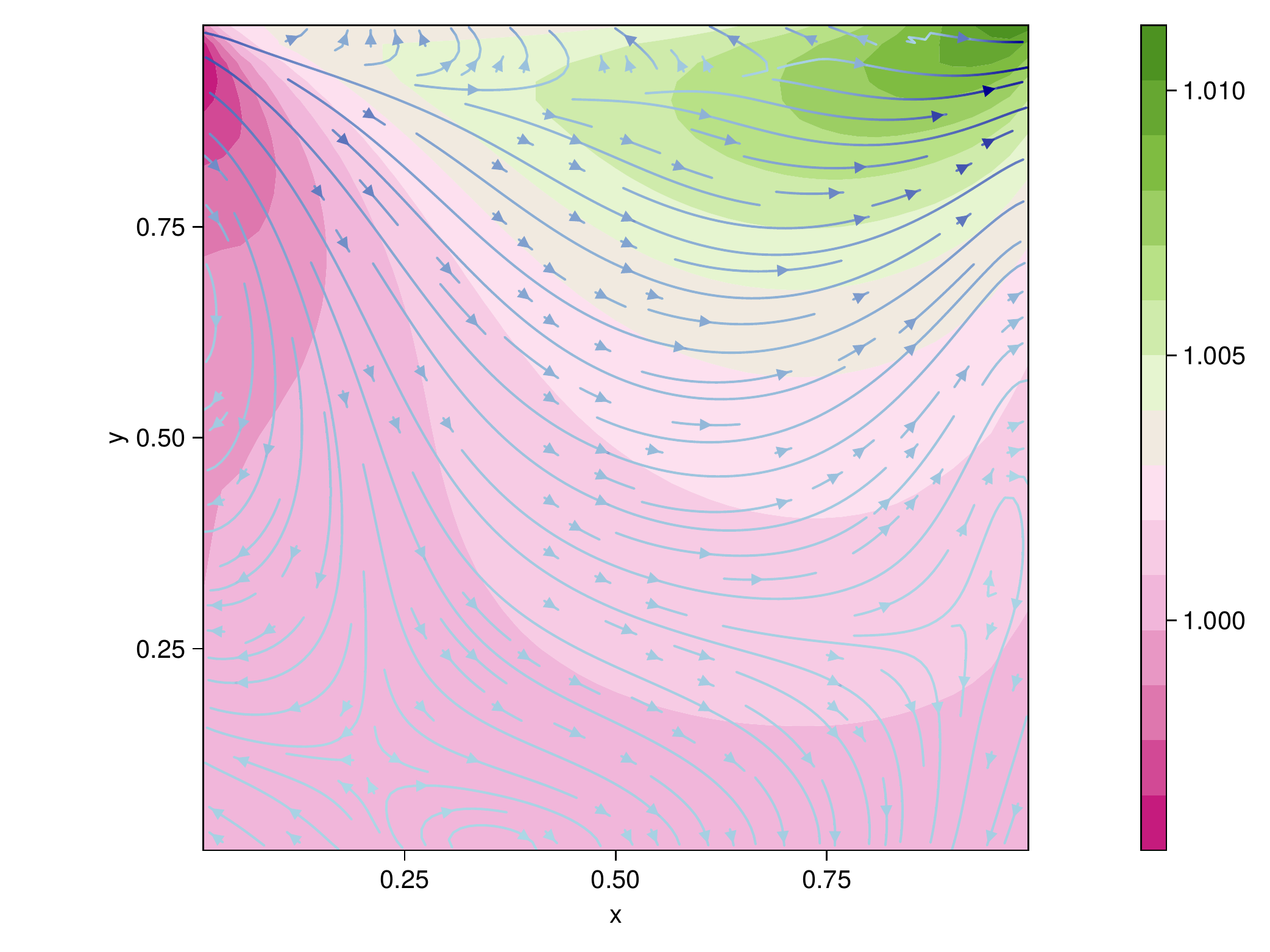}
	}
	\caption{Contours of density with streamlines and temperature with heat flux vectors inside the cavity.}
    \label{fig:cavity contour}
\end{figure}

\begin{figure}[htb!]
	\centering
	\subfigure[$U$-velocity along vertical central line]{
		\includegraphics[width=0.4\textwidth]{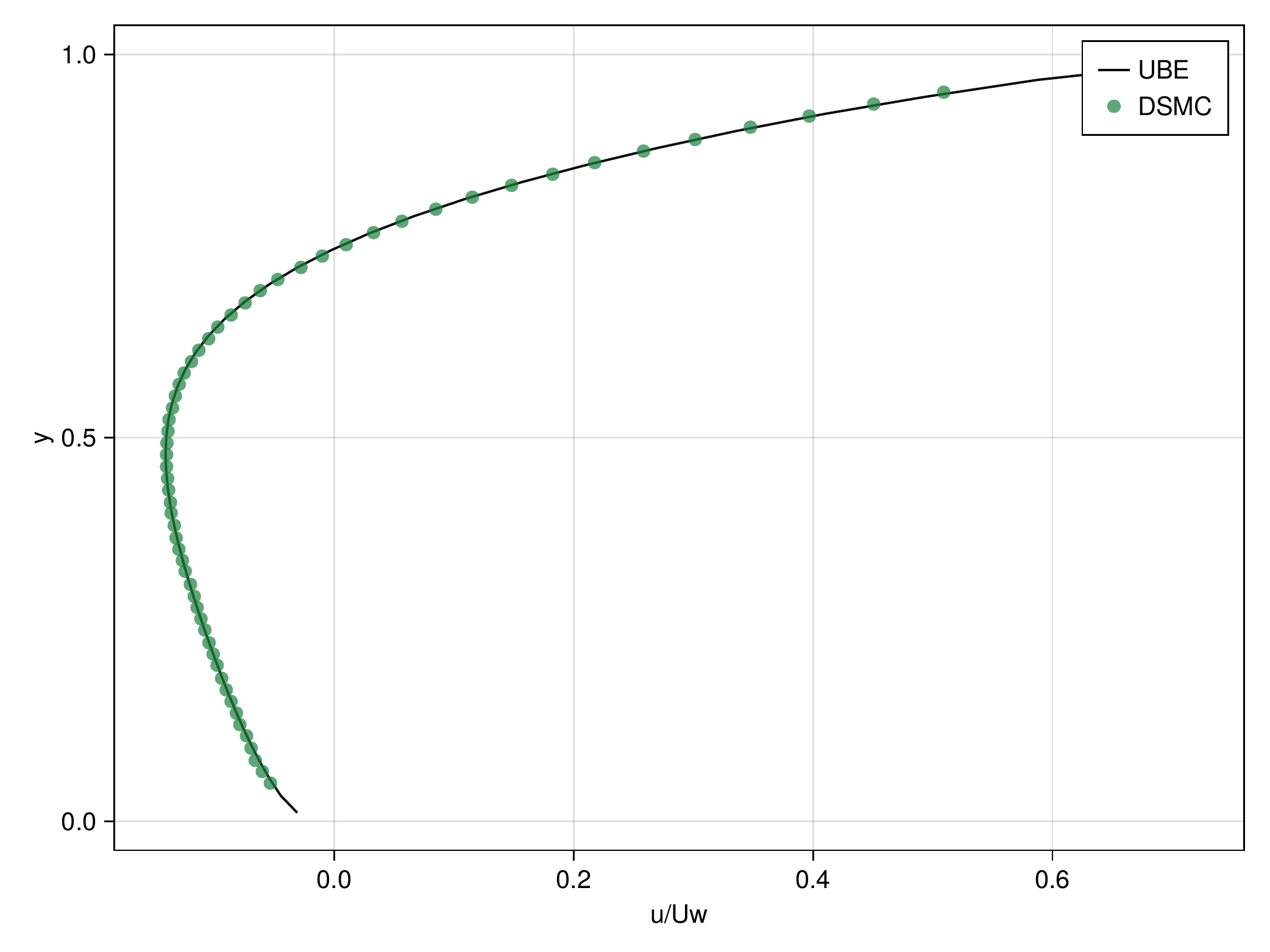}
	}
	\subfigure[$V$-velocity along horizontal central line]{
		\includegraphics[width=0.4\textwidth]{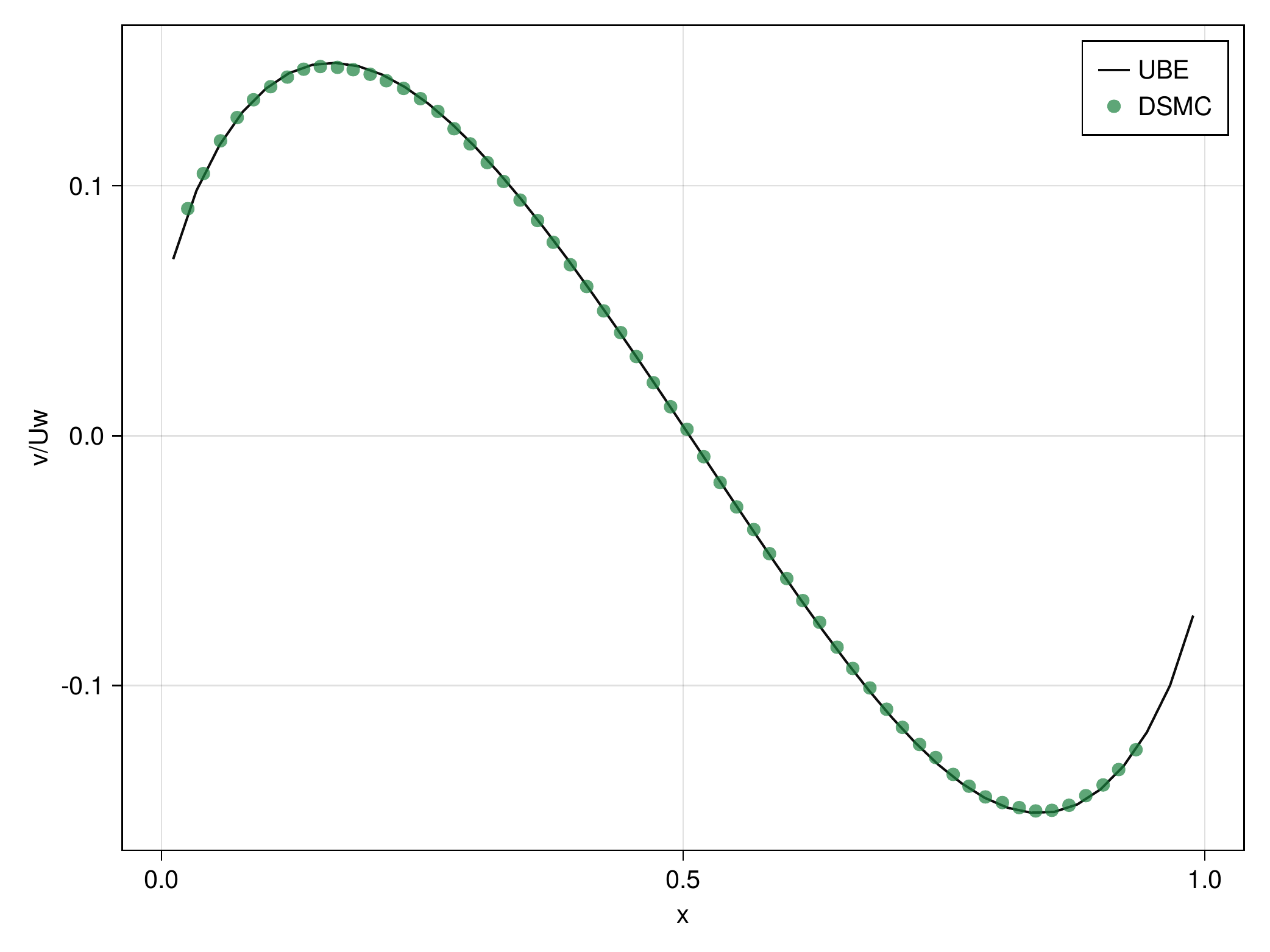}
	}
	\caption{Velocity profiles along vertical and horizontal central lines of the cavity. The results are normalized with the wall speed $U_w$.}
    \label{fig:cavity line}
\end{figure}

\begin{figure}[htb!]
	\centering
	\subfigure[UBE]{
		\includegraphics[width=0.47\textwidth]{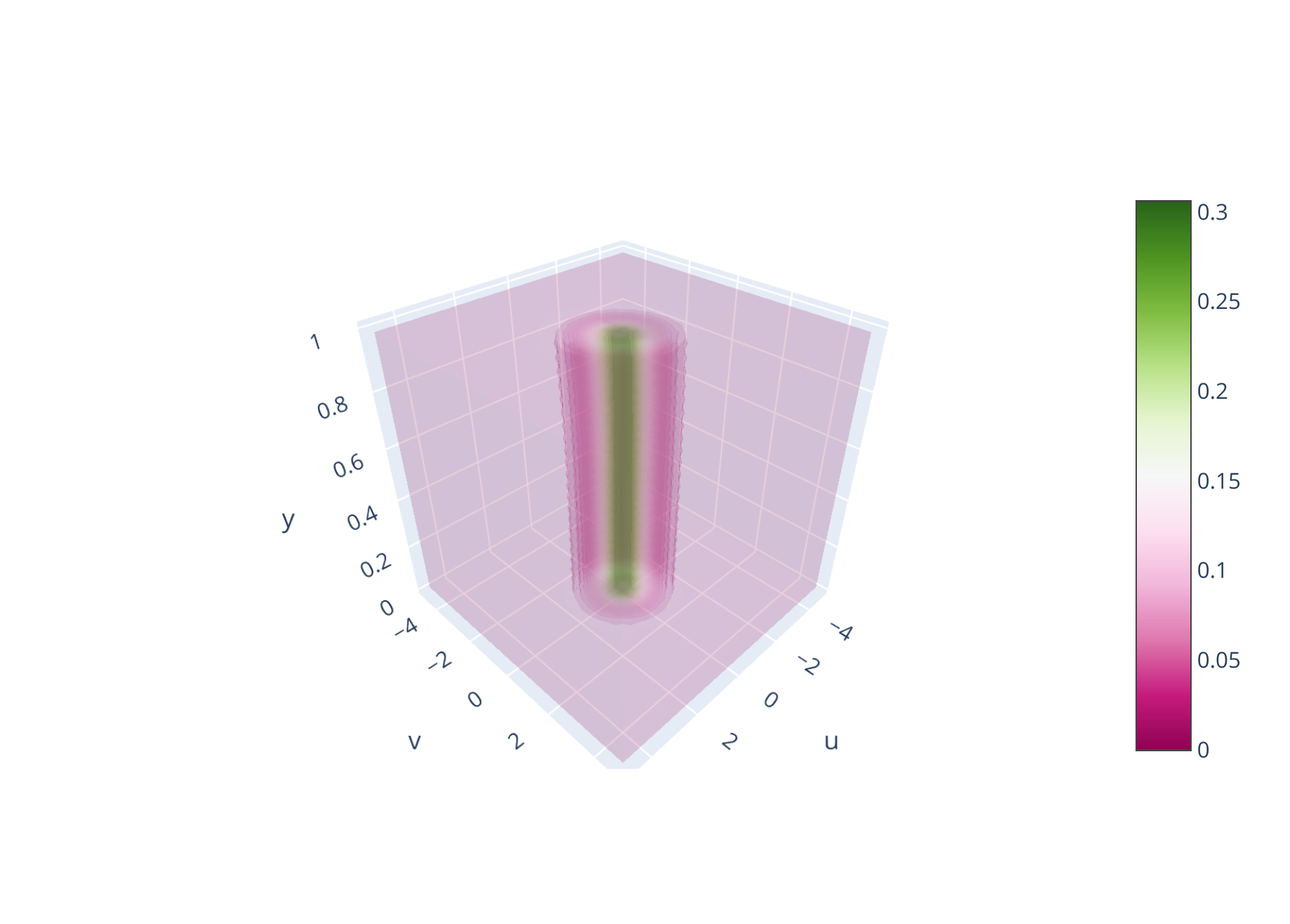}
	}
	\subfigure[BGK]{
		\includegraphics[width=0.47\textwidth]{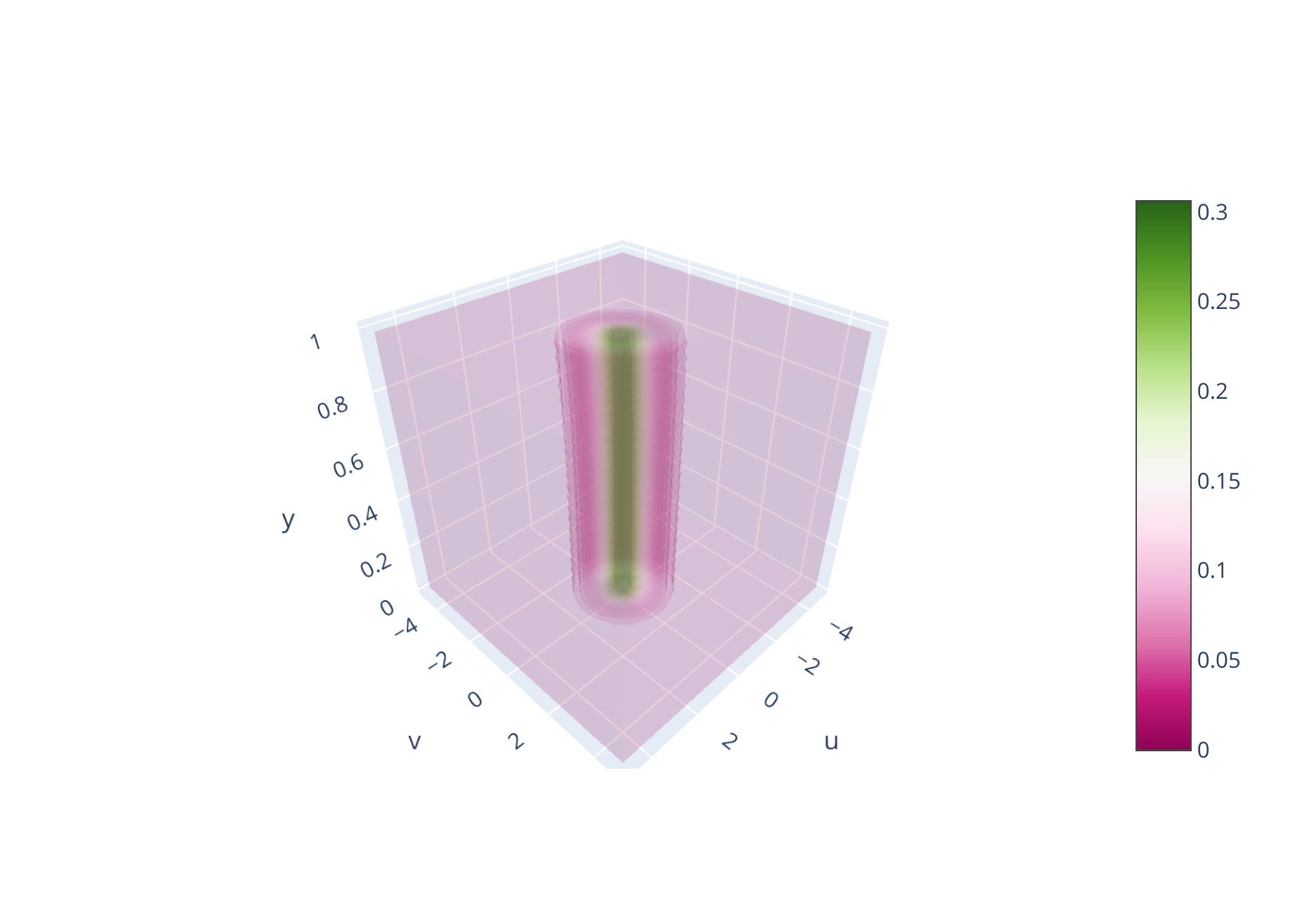}
	}
	\subfigure[UBE]{
		\includegraphics[width=0.47\textwidth]{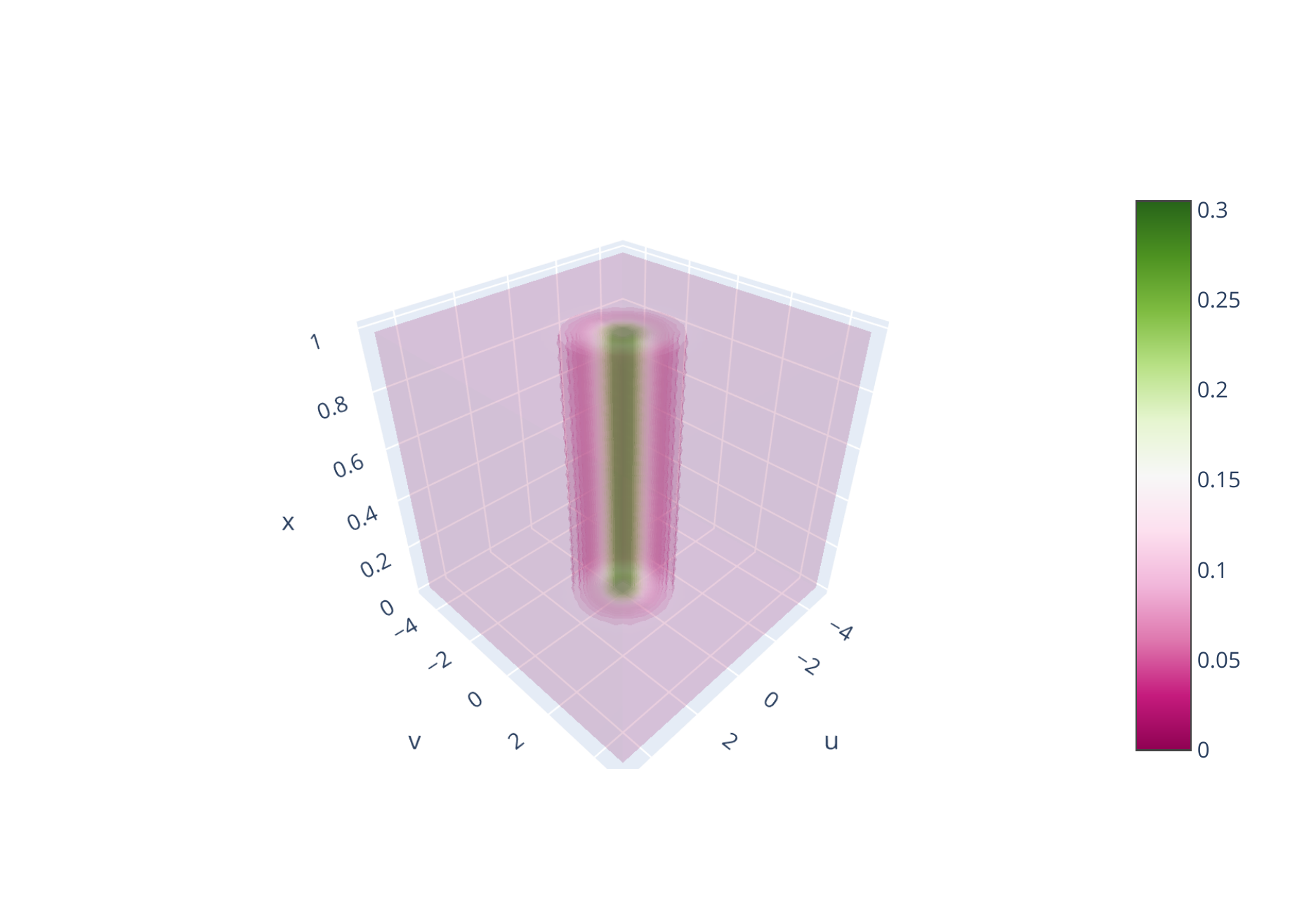}
	}
	\subfigure[BGK]{
		\includegraphics[width=0.47\textwidth]{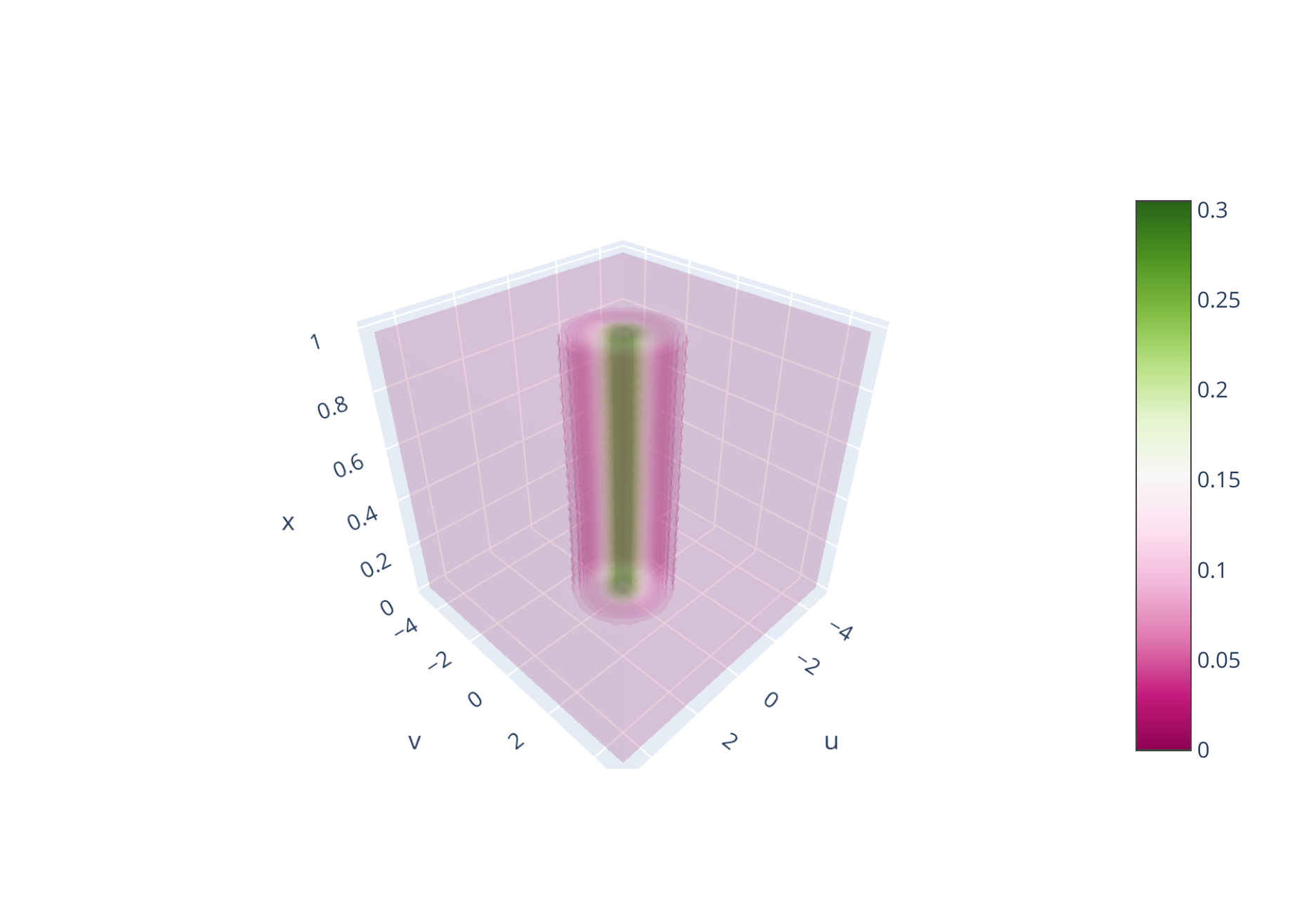}
	}
	\caption{Particle distribution functions along vertical (first row) and horizontal (second row) central lines of the cavity.}
    \label{fig:cavity pdf}
\end{figure}

\begin{figure}[htb!]
	\centering
	\subfigure[UBE]{
		\includegraphics[width=0.47\textwidth]{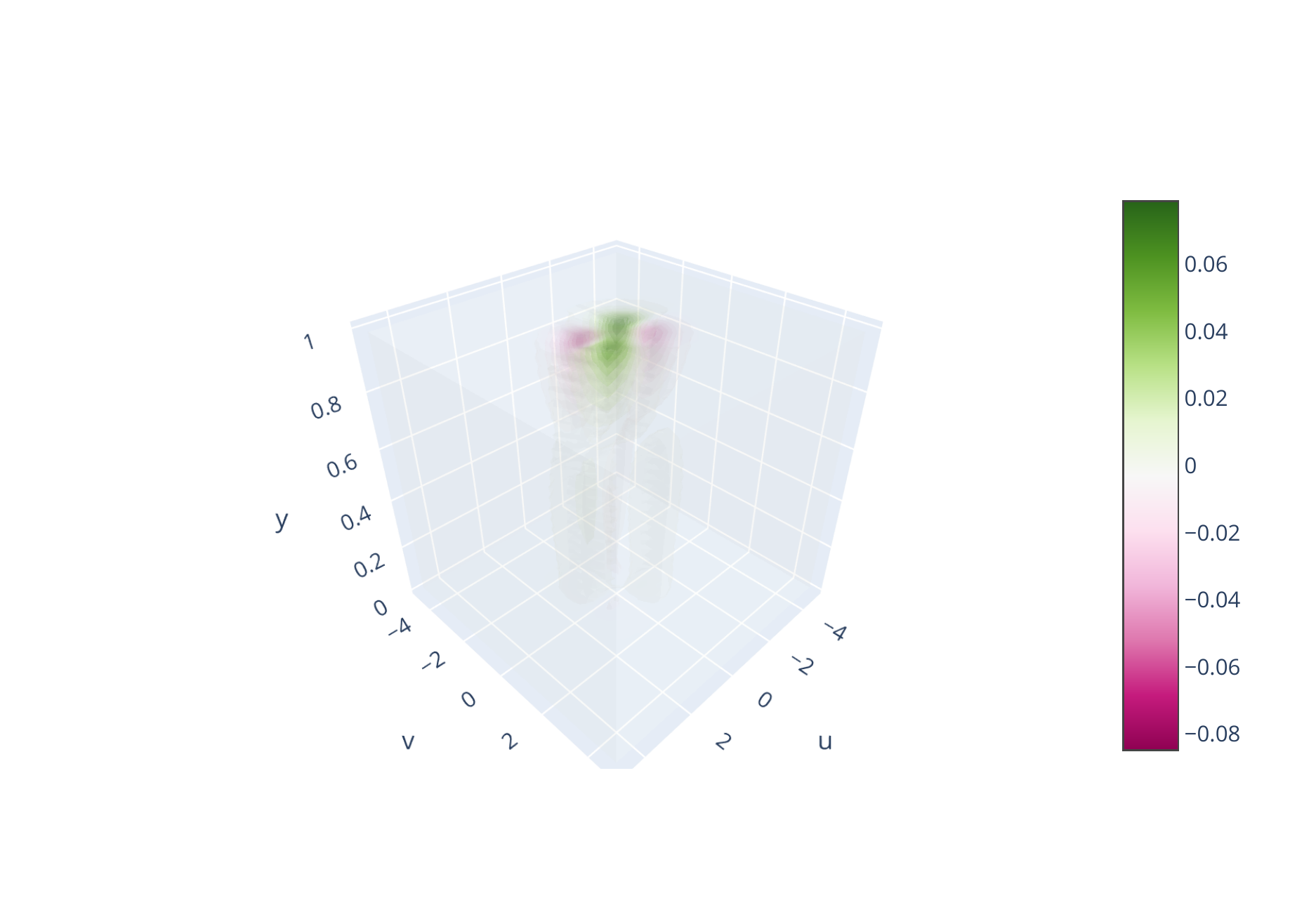}
	}
	\subfigure[BGK]{
		\includegraphics[width=0.47\textwidth]{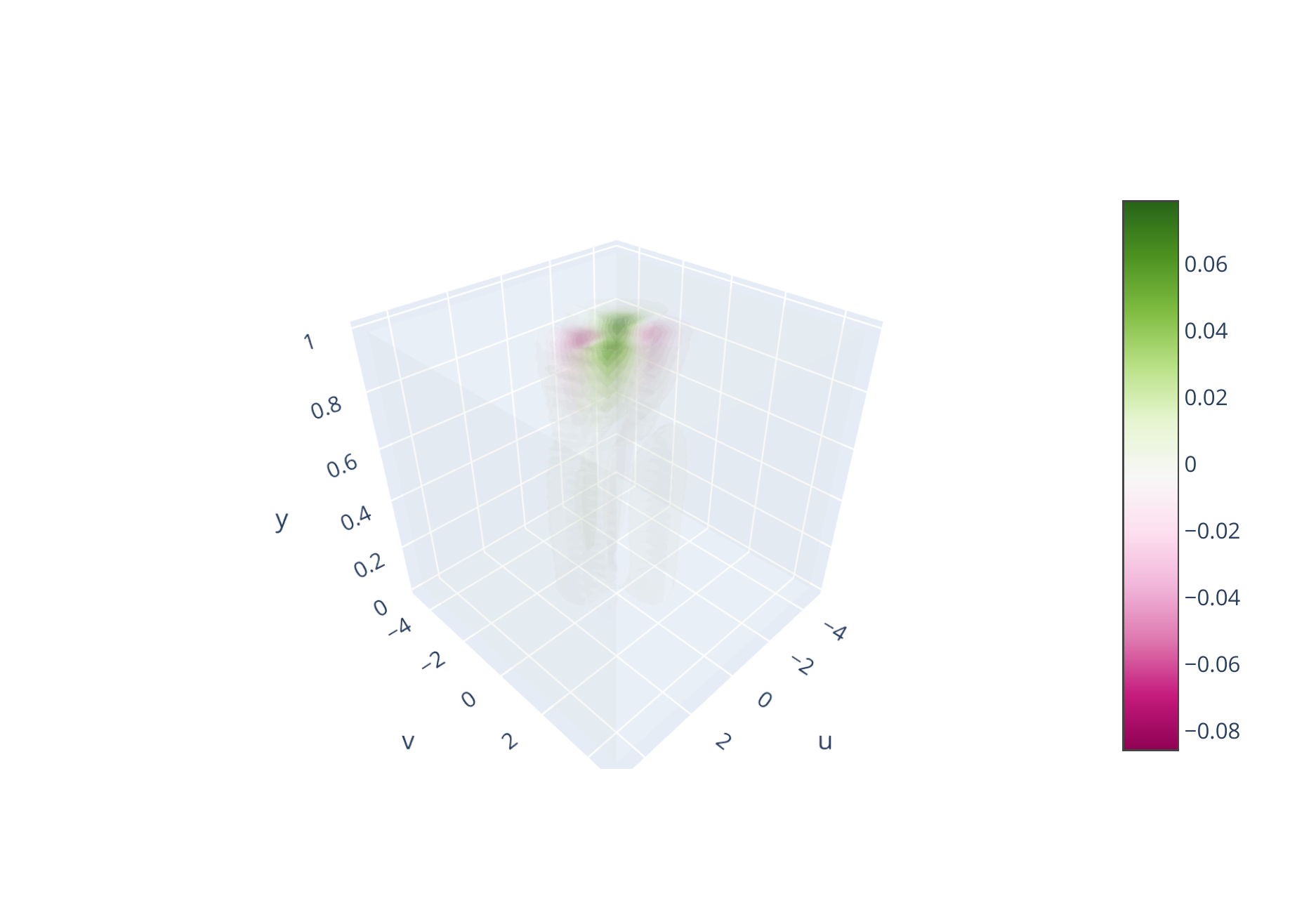}
	}
	\subfigure[UBE]{
		\includegraphics[width=0.47\textwidth]{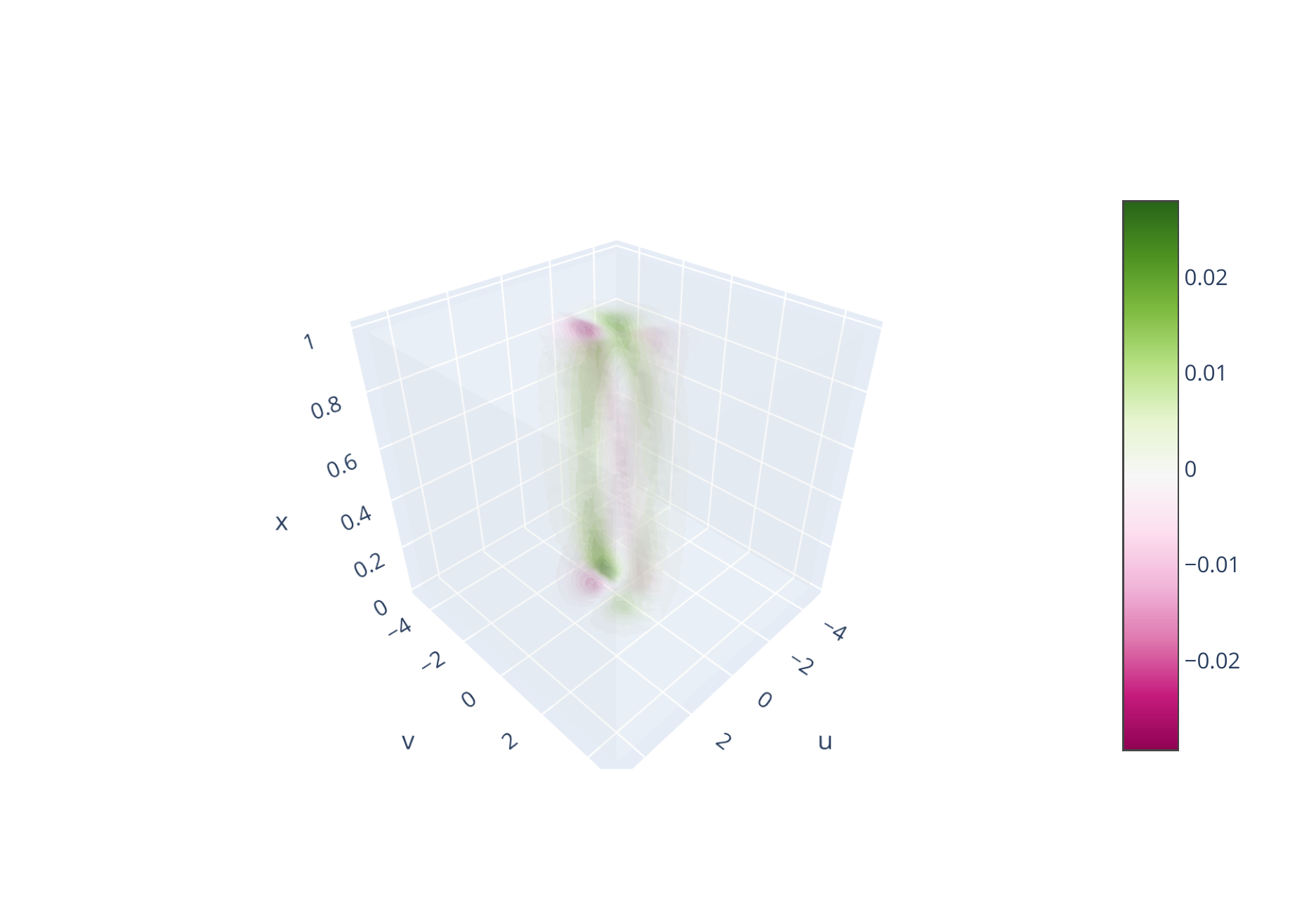}
	}
	\subfigure[BGK]{
		\includegraphics[width=0.47\textwidth]{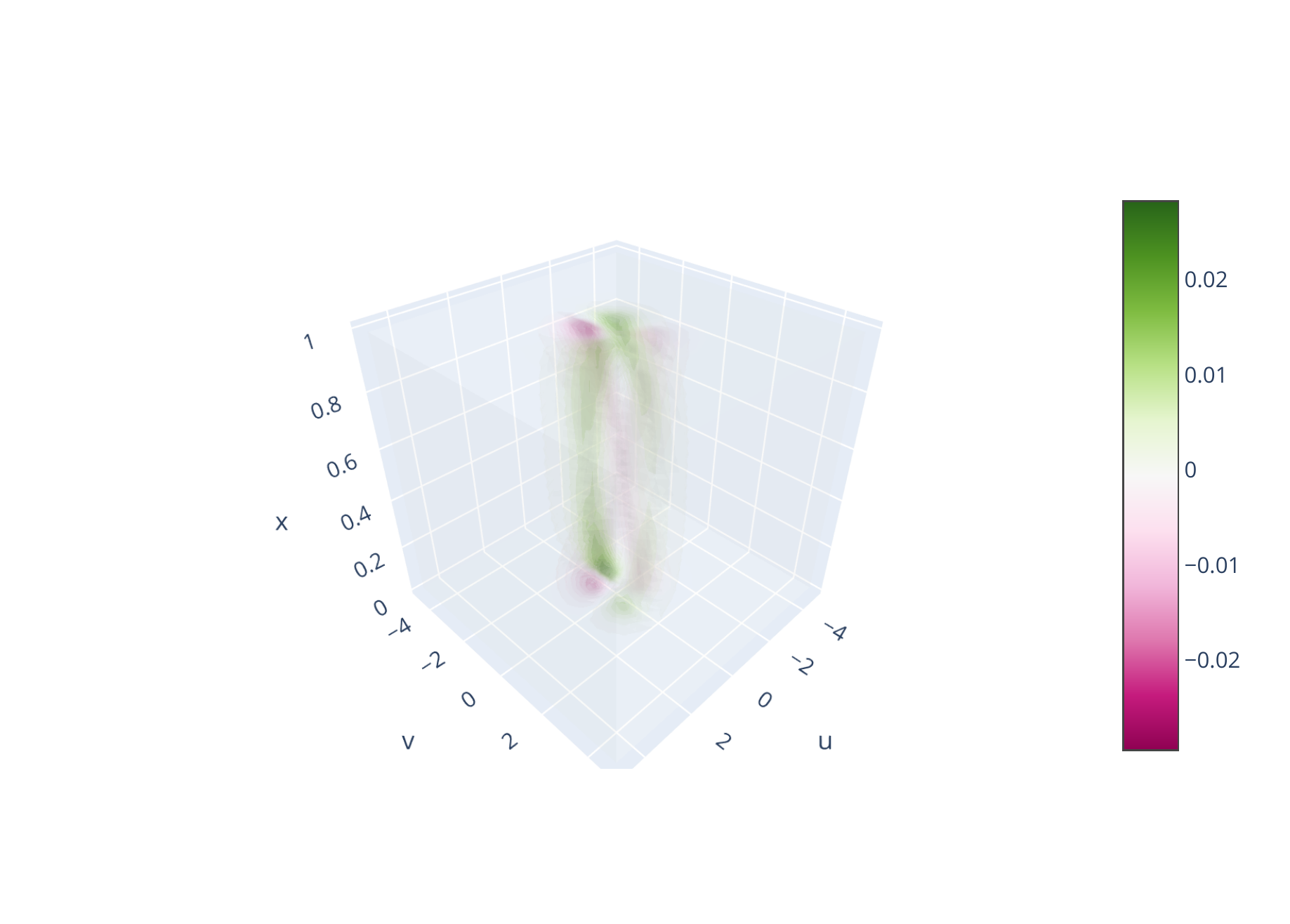}
	}
	\caption{Collision terms along vertical (first row) and horizontal (second row) central lines of the cavity.}
    \label{fig:cavity collision}
\end{figure}

\end{document}